\documentclass[11pt]{article}

\usepackage[T1]{fontenc}
\usepackage{microtype}
\usepackage{amsmath,amsthm}
\usepackage{newpxtext,newpxmath}
\usepackage[margin=15pt,font=small,labelfont={bf}]{caption}
\usepackage[margin=1in]{geometry}
\usepackage[english]{babel}
\usepackage[utf8]{inputenc}

\usepackage[backend=biber,style=alphabetic,maxalphanames=4,minalphanames=4,maxbibnames=99,minbibnames=99]{biblatex}
\DeclareLabelalphaTemplate{
  \labelelement{
    \field[final]{shorthand}
    \field{label}
    \field[strwidth=3,strside=left,ifnames=1]{labelname}
    \field[varwidthlist,names=3]{labelname}
  }
  \labelelement{
    \field[strwidth=2,strside=right]{year}
  }
}

\usepackage{graphicx}
\usepackage{algorithm}
\usepackage[noend]{algpseudocode}

\AtBeginDocument{
  }
\usepackage{xcolor}
\usepackage[colorlinks=true,
    citecolor=darkgreen,
    linkcolor=darkred
]{hyperref}
\usepackage{aliascnt}
\definecolor{darkgreen}{rgb}{0,0.5,0.5}
\definecolor{darkred}{rgb}{0.5,0,0}
\newcommand\drop[1]{}
\newcommand*{\newaliastheorem}[3]{
  \newaliascnt{#1}{#2}
  \newtheorem{#1}[#1]{#3}
  \aliascntresetthe{#1}
  \expandafter\newcommand\csname #1autorefname\endcsname{#3}
}
\newtheorem{theorem}{Theorem}[section]
\newaliastheorem{proposition}{theorem}{Proposition}
\newaliastheorem{lemma}{theorem}{Lemma}
\newaliastheorem{corollary}{theorem}{Corollary}
\newaliastheorem{definition}{theorem}{Definition}
\newaliastheorem{conjecture}{theorem}{Conjecture}
\newaliastheorem{example}{theorem}{Example}

\title{Online Graph Coloring for $k$-Colorable Graphs}
\author{Ken-ichi Kawarabayashi\thanks{National Institute of Informatics \& The University of Tokyo, Tokyo, Japan, \texttt{k\_keniti@nii.ac.jp}.\\ 
    Supported by JSPS Kakenhi 22H05001, JP25K24465, and by JST ASPIRE JPMJAP2302.}
    \and 
    Hirotaka Yoneda\thanks{The University of Tokyo, Tokyo, Japan, \texttt{squar37@gmail.com}.\\ 
    Supported by JSPS JP25K24465, JST ASPIRE JPMJAP2302 and by JST ACT-X JPMJAX25CT.}
    \and 
    Masataka Yoneda\thanks{The University of Tokyo, Tokyo, Japan, \texttt{e869120@gmail.com}.\\ 
    Supported by JSPS JP25K24465 and JST ASPIRE JPMJAP2302.}}
\date{}
\begin{document}
\maketitle

\pagenumbering{gobble}

\begin{abstract}
We study the problem of online graph coloring for $k$-colorable graphs. The best previously known deterministic algorithm uses $\widetilde{O}(n^{1-\frac{1}{k!}})$ colors for general $k$ and $\widetilde{O}(n^{5/6})$ colors for $k = 4$, both given by Kierstead in 1998. In this paper, we finally break this barrier, achieving the first major improvement in nearly three decades.
Our results are summarized as follows:

\begin{enumerate}
    \item \textbf{$k \geq 5$ case.} We provide a deterministic online algorithm to color $k$-colorable graphs with $\widetilde{O}(n^{1-\frac{1}{k(k-1)/2}})$ colors, significantly improving the current upper bound of $\widetilde{O}(n^{1-\frac{1}{k!}})$ colors. Our algorithm also matches the best-known bound for $k = 4$ ($\widetilde{O}(n^{5/6})$ colors).
    \item \textbf{$k = 4$ case.} We provide a deterministic online algorithm to color $4$-colorable graphs with $\widetilde{O}(n^{14/17})$ colors, improving the current upper bound of $\widetilde{O}(n^{5/6})$ colors. 
    \item \textbf{$k = 2$ case.} We show that for randomized algorithms, the upper bound is $1.034 \log_2 n + O(1)$ colors and the lower bound is $\frac{91}{96} \log_2 n - O(1)$ colors. This means that we close the gap to a factor of $1.09$.
\end{enumerate}

With our algorithm for the $k \geq 5$ case, we also obtain a deterministic online algorithm for graph coloring that achieves a competitive ratio of $O(\frac{n}{\log \log n})$, which improves the best-known result of $O(\frac{n \log \log \log n}{\log \log n})$ by Kierstead.

For the bipartite graph case ($k = 2$), the limit of online deterministic algorithms is known: any deterministic algorithm requires $2 \log_2 n - O(1)$ colors. Our results imply that randomized algorithms can perform slightly better but still have a limit.
\end{abstract}

\newpage
{
    \setcounter{tocdepth}{2}
    \addtocontents{toc}{\protect\linespread{1.08}\protect\selectfont} 
    \tableofcontents
    \clearpage
}
\pagenumbering{arabic}
\setcounter{page}{1}

\section{Introduction}\label{sec:intro}

Graph coloring is the problem of coloring each vertex of a graph $G = (V, E)$ with the minimum number of colors, so that no two adjacent vertices have the same color. It is one of the most fundamental problems in graph theory and algorithms.

We study online algorithms for graph coloring, where the input graph is revealed over time. In the \emph{online coloring problem}, vertices arrive one by one, together with their incident edges. Upon the arrival of each vertex, we must immediately assign its color (which must differ from the neighbors' colors) before the next vertex arrives. The challenge is to design a coloring strategy that minimizes the number of colors used.

The online coloring problem appears to be even harder than graph coloring. The most intuitive strategy is perhaps the First-Fit algorithm, which repeatedly assigns the least-indexed available color to each arriving vertex. However, this strategy cannot bound the number of colors, even for bipartite graphs (\autoref{fig:intro-firstfit}).

\begin{figure}[htbp]
  \centering
  \includegraphics[width=\linewidth]{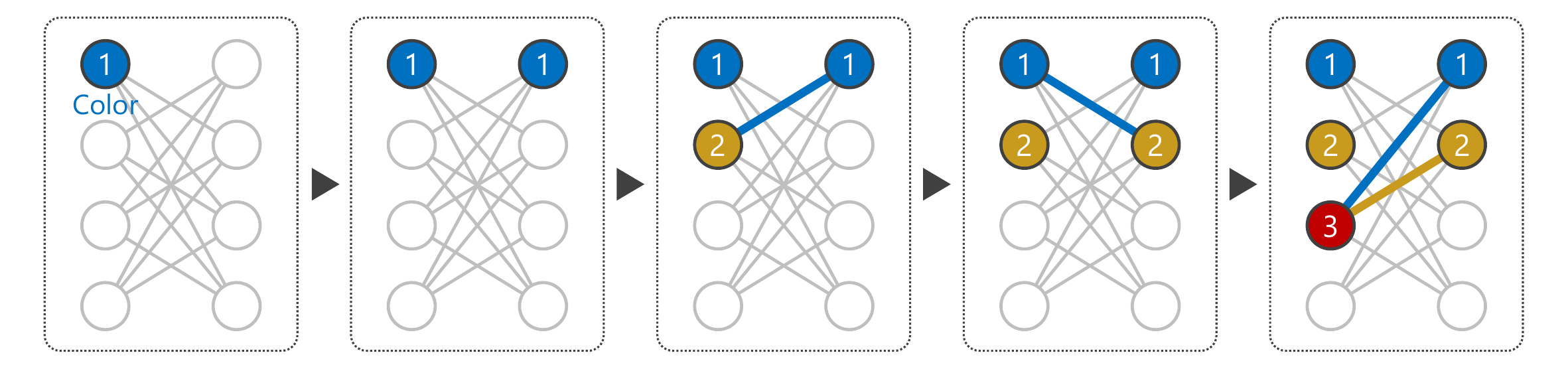}
  \caption{A worst-case input on which First-Fit uses $\frac{1}{2} n$ colors even for bipartite graphs \cite{Ira94}.}
  \label{fig:intro-firstfit}
\end{figure}

\subsection{Previous Studies}

As background information, the graph coloring problem is known to be hard to solve, or even to approximate, in polynomial time. It is inapproximable within a factor $n^{1-\varepsilon}$ for any $\varepsilon > 0$, unless ZPP = NP \cite{FEIGE1998187}. The best known approximation factor is $O(\frac{n (\log \log n)^2}{(\log n)^3})$ \cite{HALLDORSSON199319}. Even for 3-colorable graphs, the best approximation algorithm achieves $\widetilde{O}(n^{0.1954})$ colors \cite{KTY24, BHL26}.

\paragraph{Online coloring.}
Generally, the main research direction in online algorithms is to minimize the competitive ratio.\footnote{The competitive ratio for an algorithm is the maximum value of ``the algorithm's performance (i.e., the number of colors used) divided by the optimal value (i.e., chromatic number)'' for all possible inputs.} For deterministic online coloring, Lov\'{a}sz, Saks, Trotter (1989) \cite{LST89} gave a deterministic online coloring algorithm with a competitive ratio of $O(\frac{n}{\log^* n})$, and Kierstead (1998) \cite{Kie98} improved this ratio (though not explicitly stated) to $O(\frac{n \log \log \log n}{\log \log n})$. For randomized algorithms, Vishwanathan (1992) \cite{Vis92} gave a randomized algorithm with a competitive ratio of $O(\frac{n}{(\log n)^{1/2}})$, and Halld{\'o}rsson (1997) \cite{HALLDORSSON1997265} improved this ratio to $O(\frac{n}{\log n})$. However, the lower bound of $\Omega(\frac{2n}{(\log n)^2})$ is known for both deterministic and randomized algorithms \cite{HS94}.

\paragraph{Algorithms for $k$-colorable graphs.}
Given very strong negative results for approximating the chromatic number in the online graph coloring problem (even worse than graph coloring), it is natural to consider the online graph coloring for $k$-colorable graphs, where $k$ is a constant.

For $k \geq 3$, a deterministic algorithm that uses $O(n^{1-\frac{1}{k!}})$ colors is given by Kierstead (1998) \cite{Kie98}, and a randomized algorithm that uses $\widetilde{O}(n^{1-\frac{1}{k-1}})$ colors is given by Vishwanathan (1992) \cite{Vis92}. For special cases when $k = 3$ and $k = 4$, deterministic algorithms that use $\widetilde{O}(n^{2/3})$ colors and $\widetilde{O}(n^{5/6})$ colors are also given by Kierstead (1998) \cite{Kie98}. However, no improvement has been made for deterministic algorithms since 1998, and a significant gap remains between the upper and lower bounds. Surprisingly, the only known lower bound is $\Omega((\log n)^{k-1})$ colors \cite{Vis92}.

For $k = 2$, a deterministic algorithm that uses $2 \log_2 (n+1)$ colors is found by Lov\'{a}sz, Saks, Trotter (1989) \cite{LST89}. For the lower bound of deterministic algorithms, after a series of studies \cite{Bea76, BBH+14, GKM+14}, Gutowski, Kozik, Micek, Zhu (2014) \cite{GKM+14} finally showed that achieving $2 \log_2 n - 10$ colors is impossible. Note that the lower bound is $\log_2 n$ even for trees \cite{Bea76, GL88}. However, the performance of randomized algorithms remains open, where Vishwanathan \cite{Vis92} gives a lower bound of $\frac{1}{72} \log_2 n$ colors.

\vspace{-0.4em}

\paragraph{Related works.}
Deterministic and randomized online algorithms to color the following graph classes achieve  $\Theta(\log n)$ colors: tree, planar, bounded-treewidth, and disk graphs \cite{DBLP:conf/esa/AlbersS17}. Also, the First-Fit algorithm performs better for $d$-degenerate graphs (i.e., every induced subgraph has a vertex of degree at most $d$), which uses only $O(d \log n)$ colors. For more references on online graph colorings, we refer the reader to the survey by Kierstead \cite{Kie05}.

A related, recently well-studied topic is \emph{online edge coloring}; edges arrive one by one, and we must color each edge immediately upon arrival. For example, a deterministic algorithm that only uses $(1+o(1)) \Delta$ colors when the maximum degree $\Delta$ is $\omega(\log n)$, is known \cite{blikstad2024online,blikstad2025deterministic,blikstad2025online}.

\vspace{-0.4em}

\paragraph{Deterministic vs. Randomized.}
Generally speaking, deterministic online algorithms have large gaps compared to randomized online algorithms for many problems. For example, deterministic online caching algorithms cannot provide better worst-case guarantees than known trivial algorithms, but randomization may yield significantly better results (see \cite{timbook}, Chapter 24).

The difficulty for deterministic online algorithms is that, for worst-case analysis, we must assume ``adversarial attack'' in the context of machine learning (this assumption is typical in the adversarial bandit problem \cite{ACF+02}). This means that in online coloring, the adversary may construct the worst-case input for the next vertex by observing the colors assigned to the previous vertices.

\subsection{Our Contributions}\label{subsec:intro-contributions}

We study online coloring for $k$-colorable graphs, where $k$ is a constant. Recall that the best-known deterministic algorithms use $\widetilde{O}(n^{2/3})$ colors for $k = 3$, $\widetilde{O}(n^{5/6})$ colors for $k = 4$, and $\widetilde{O}(n^{1-\frac{1}{k!}})$ colors for any $k \geq 5$, all given by Kierstead (1998) \cite{Kie98}. For $k = 2$, the best-known algorithm, deterministic or randomized, uses $2 \log_2 (n+1)$ colors \cite{LST89}. Nearly thirty years later, this paper finally makes progress. We improve the results for all $k \geq 2$ except $k = 3$.

\begin{enumerate}
    \item \textbf{Case $k \geq 5$.} We propose a deterministic algorithm for general $k \geq 2$ that uses $\widetilde{O}(n^{1-\frac{1}{k(k-1)/2}})$ colors (\autoref{thm:k5-col}). This matches the best-known bound for $k = 2, 3, 4$ and improves the bound for $k \geq 5$. Our algorithm also improves the competitive ratio for deterministic online coloring, from $O(\frac{n \log \log \log n}{\log \log n})$ \cite{Kie98} to $O(\frac{n}{\log \log n})$ (\autoref{thm:k5-competitive}).
    \item \textbf{Case $k = 4$.} We propose a deterministic algorithm for $k = 4$ that uses $\widetilde{O}(n^{14/17})$ colors (\autoref{thm:k4-col}), which improves the best-known bound by a factor of $\widetilde{O}(n^{1/102})$.
    \item \textbf{Case $k = 2$.} We propose a randomized algorithm for $k = 2$ (bipartite graphs) that uses $1.034 \log_2 n + O(1)$ colors (\autoref{thm:upperbound}), which is better than the deterministic bound. We also show the lower bound that any randomized algorithm uses $\frac{91}{96} \log_2 n - O(1)$ colors (\autoref{thm:lowerbound1}), closing the upper/lower bound gap to a factor of $1.09$.
\end{enumerate}

We also note that by applying our $k = 4$ result to the general case, we obtain an algorithm for general $k \geq 4$ that uses $\widetilde{O}(n^{1-\frac{1}{k(k-1)/2-1/3}})$ colors (\autoref{col:k4-improvement}). See \autoref{tab:intro} for comparison.

\begin{table}[htbp]
    \centering
    \begin{tabular}{|c||c|c|c|c|c|} \hline
        $k$ & 3 & 4 & 5 & 6 & 7 \\
        Previous Results & $\widetilde{O}(n^{0.6667})$ & $\widetilde{O}(n^{0.8334})$ & $\widetilde{O}(n^{0.9917})$ & $\widetilde{O}(n^{0.9987})$ & $\widetilde{O}(n^{0.9999})$ \\
        Our Results & --- & $\widetilde{O}(n^{0.8236})$ & $\widetilde{O}(n^{0.8966})$ & $\widetilde{O}(n^{0.9319})$ & $\widetilde{O}(n^{0.9517})$ \\
        \hline
    \end{tabular}
    \caption{Comparison of previous best results \cite{Kie98} and our results for $k \leq 7$.}
    \label{tab:intro}
\end{table}

\subsection{Overview of Our Techniques}

We now provide an overview of our techniques for the results stated above.

\subsubsection{Algorithm for General $k$ (\autoref{sec:k5})}

We aim for $O(n^{1-\varepsilon})$ colors for $k$-colorable graphs, for some $\varepsilon > 0$. For each vertex $v$, we color $v$ with the First-Fit strategy if $v$ can be colored by some of the colors $1, \dots, n^{1-\varepsilon}$, and otherwise we must pay for new colors. But in the latter case, $v$ has neighbors of each color $1, \dots, n^{1-\varepsilon}$. Hence, $|N_S(v)| \geq n^{1-\varepsilon}$ holds for each $v \in T$, where $S$ is the vertices colored by First-Fit and $T$ is the rest of the vertices. We effectively color the vertices of $T$ using this large-degree condition. A major pitfall in designing these online algorithms is the temporal ambiguity between past and future vertices, making the proofs fragile. We bypass this issue by reducing it to the situation where $S$ is given in advance before the vertices in $T$ arrive, which we call the \emph{S-T problem}.

We take an approach that restructures Kierstead's algorithm \cite{Kie98} to obtain a ``good subset'' $S' \subseteq S$ that enables effective coloring of future vertices. We say that $S'$ is an ``$\ell$-color set'' if there \emph{exists} a $k$-coloring of $G$ in which at most $\ell$ colors appear in $S'$. The best scenario is that $S'$ is a 1-color set $(\ell = 1)$; in this case, the graph $G[N_T(S')]$ is $(k-1)$-colorable. In particular, all the ``future'' vertices that are adjacent to a vertex in $S'$ can be $(k-1)$-colored. This allows us to apply our algorithm inductively on $k$, and hence we can (inductively) color the vertices in $N_T(S')$.

As the original $S$ is a $k$-color set $(\ell = k)$, the question is how to obtain a good subset with a smaller $\ell$. For $\ell = k-1$, we pick a vertex $v_1 \in T$, and then $S_1 := N_S(v_1)$ is indeed a $(k-1)$-color set. For $\ell = k-2$, we pick two \emph{adjacent} vertices $v_2, v_3 \in T$, and then either $S_2 := N_{S_1}(v_2)$ or $S_3 := N_{S_1}(v_3)$ is indeed a $(k-2)$-color set (see \autoref{fig-614}). In general, given an $\ell$-color set $S'$, we pick a small subset $X_T \subseteq T$ such that $G[X_T]$ is non-$(k-\ell)$-colorable as a ``witness'', and then $N_{S'}(t)$ is an $(\ell-1)$-color set for at least one $t \in X_T$. In this way, going down for $\ell = k, \dots, 2$, we finally obtain a 1-color set $S'$. Then, as argued in the previous paragraph, we can color the vertices in $N_T(S') \subseteq T$ using an algorithm for $(k-1)$-colorable graphs. Since each vertex has large degree, we are guaranteed to obtain a subset that is large enough, and therefore we do not pay too many colors for coloring vertices in $T$. Also, some readers may worry that many mistakes might occur (e.g., either $S_2$ or $S_3$ may not be a $(k-2)$-color set), but we can bound the number of mistakes as long as the witness $X_T$ is small enough, and therefore do not pay many colors for mistakes.

However, we must consider the case that such a witness $X_T$ does not exist. Hence, we must build an online coloring algorithm for ``locally $(k-\ell)$-colorable graphs'' where every small (constant-order) subgraph is $(k-\ell)$-colorable. Kierstead bypassed this situation by using the entire subset as a witness, but this caused many ``mistakes'' and hit a factorial barrier on the number of colors used, with $O(n^{1-\frac{1}{k!}})$ colors.

The most important idea is to introduce the concept of \emph{level-$\ell$ sets}, a localized version of $\ell$-color sets. By proving that we can bound the size of the witness subset, we tightly control the number of spawned sub-instances (hence, unlike Kierstead's algorithm, we can bound the number of ``mistakes''). This leads to $O(n^{1-\frac{1}{k(k-1)/2+1}})$ colors for locally $k$-colorable graphs, and with this result, we achieve $O(n^{1-\frac{1}{k(k-1)/2}})$ colors for $k$-colorable graphs.

\begin{figure}[htpb]
  \centering
  \includegraphics[width=\linewidth]{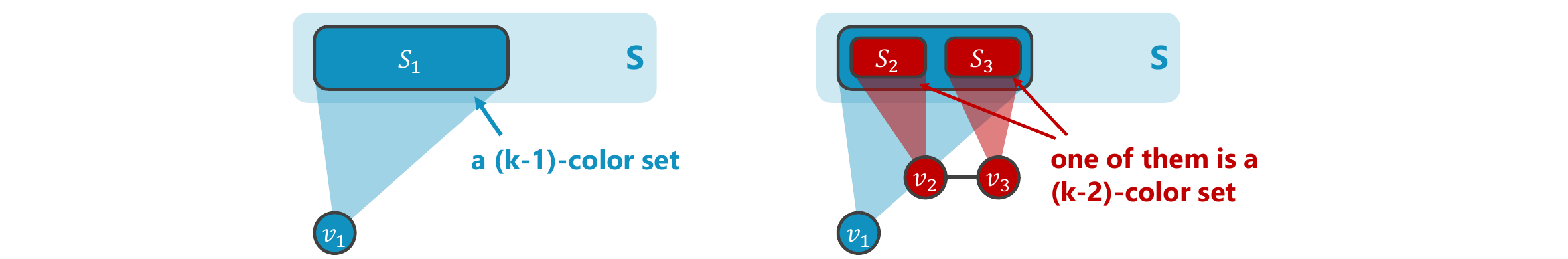}
  \caption{The creation of $(k-1)$-color and $(k-2)$-color sets. Note that the entire set $S$ is a $k$-color set.}
  \label{fig-614}
\end{figure}

\vspace{-0.5em}

\subsubsection{Algorithm for $k = 4$ (\autoref{sec:k4})}

The key to our algorithm is the use of second-neighborhood structure. If we run the algorithm above (for general $k$) with $k = 4$, we obtain 1-color set of $\Omega(n^{1/2})$ vertices, 2-color sets of $\Omega(n^{2/3})$ vertices, and 3-color sets of $\Omega(n^{5/6})$ vertices. If any of the larger sets (e.g., 3-color sets of $\Omega(n^{0.84})$ vertices) were obtained, there would be a good chance for improving the algorithm.

We find such sets with our novel \emph{double greedy method}, which applies First-Fit twice to exploit the second-neighborhood structure. Consider the following strategy for each vertex $v$: (1) if $v$ can be colored by the first $n^{5/6}$ colors, color it with the smallest one (First-Fit), (2) otherwise, if $v$ can be colored by the next $n^{5/6}$ colors, color it with the smallest one (First-Fit), and (3) otherwise, color $v$ with a different algorithm. Let $S$ be the vertices colored by the first First-Fit, let $T$ be the vertices colored by the second First-Fit, and let $U$ be the rest of the vertices.\footnote{Indeed, our algorithm only use $n^{14/17}, n^{13/17}$ colors in $S, T$, respectively, but we assume $n^{5/6}$ colors for simplicity.} We can find a 1-color set of $\Omega(n^{1/2})$ vertices in $T$, using the large degrees from $U$ to $T$. Let $T' \subseteq T$ be the obtained 1-color set. Then, $N_S(T')$ is a 3-color set of $\Omega(n)$ vertices in most cases (\autoref{fig-621} left), due to the large degrees from $T$ to $S$. This is progress toward the future vertices arriving.

However, there are some exceptional cases in which we cannot obtain a 3-color set of $\Omega(n)$ vertices; in the worst case, it can be as small as $\Theta(n^{5/6})$ vertices. This happens when there exists a dense subgraph between $S$ and $T$. In such cases, we can pick a pair of vertices $u_1, u_2 \in T$ where $Z = N_S(u_1) \cap N_S(u_2)$ is large, that is, $\Omega(n^{2/3+\varepsilon})$ vertices for some $\varepsilon > 0$ (\autoref{fig-621} right). If there exists a 4-coloring of the given graph in which distinct colors are used for $u_1$ and $u_2$, then $Z$ is indeed a 2-color set, leading to make progress (for the future vertices to arrive). We call this method \emph{Common \& Simplify technique}. With these two techniques, we eventually achieve $\widetilde{O}(n^{14/17})$ colors, improving on Kierstead's result \cite{Kie98} by a factor of $\widetilde{O}(n^{1/102})$.

\begin{figure}[ht]
  \centering
  \includegraphics[width=0.98\linewidth]{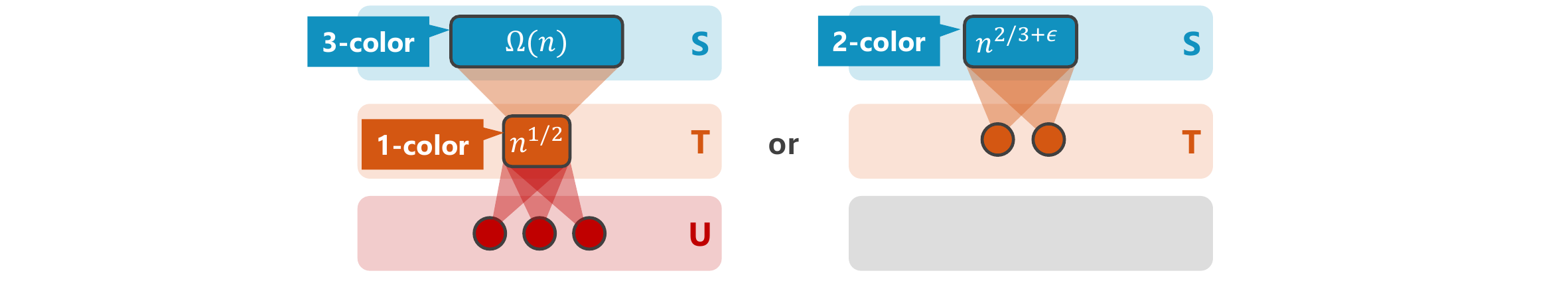}
  \caption{A sketch of our techniques in $k = 4$ case.}
  \label{fig-621}
\end{figure}

\subsubsection{Upper Bound for $k=2$ (\autoref{sec:k2-upper})}

We apply randomization to Lov\'{a}sz, Saks, Trotter's algorithm \cite{LST89} and show this leads to better performance. By analyzing how quickly the ``third color'' is forced to be used, we obtain an upper bound of $1.096 \log_2 n + O(1)$ colors. For a better result, we analyze how quickly the $(2L+1)$-th color is forced to be used for $L \geq 2$. We use a computer-aided approach based on dynamic programming. Eventually, with careful analysis up to $L = 10$, we improve the upper bound to $1.034 \log_2 n$.

\subsubsection{Lower Bound for $k=2$ (\autoref{sec:k2-lower})}

We consider a binary-tree-like graph (\autoref{fig-641}) and use the classic ``potential function'' argument. We introduce a potential function to represent the required number of colors, and show that any online algorithm would increase the potential by ``$\frac{3}{4}$ colors'' per depth, leading to a lower bound of $\frac{3}{4} \log_2 n - O(1)$ colors. For a better result, we note that when the potential increase is smaller at one depth, a larger increase typically occurs at the next depth. We analyze how potential increases occur across two consecutive depths. However, the setting is more complex and requires a computer to brute-force all the essential cases. In the end, we obtain a lower bound of $\frac{91}{96} \log_2 n - O(1)$ colors.

\begin{figure}[htpb]
  \centering
  \includegraphics[width=0.95\linewidth]{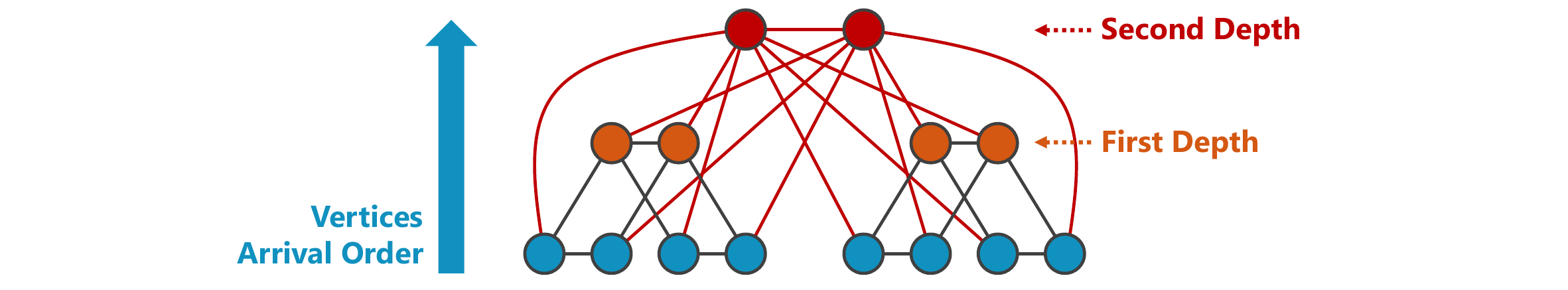}
  \caption{The input for the depth-$2$ case, which we consider for the lower bound.}
  \label{fig-641}
\end{figure}

\newpage
\section{Preliminaries}\label{sec:prelim}

\paragraph{Basic notation.}
The set of vertices adjacent to a vertex $v \in V$ is denoted by $N(v)$. The set of vertices adjacent to at least one vertex in $S \subseteq V$ is denoted by $N(S) := \bigcup_{v \in S} N(v)$. For a vertex set $X \subseteq V$, we denote $N_X(v) := N(v) \cap X$ and $N_X(S) := N(S) \cap X$. The induced subgraph of a vertex set $S \subseteq V$ is denoted by $G[S]$. The chromatic number of $G$ is denoted by $\chi(G)$.

In the $O, \Omega, \Theta$-notation, we use $\widetilde{O}(f(n)), \widetilde{\Omega}(f(n))$, and $\widetilde{\Theta}(f(n))$ to hide polylogarithmic factors in $n$. Also, we assume all logarithms are in base $2$, unless stated otherwise (i.e., $\log n := \log_2 n$).

\paragraph{Online coloring.}
We now formally define the \emph{online coloring problem}. An undirected graph $G = (V, E)$ is initially empty (i.e., $V = \emptyset$). Vertices are added to $G$ one by one, together with their incident edges. The final number of vertices $n$ is unknown until the end.

Immediately after a vertex $v$ (and its incident edges) is added, we must assign a color $c(v) \in \mathbb{N}$ to $v$. This color must differ from the colors of all its neighbors. Note that the next vertex arrives only after coloring $v$. The goal is to minimize the number of colors used, that is, $|\{c(v) : v \in V\}|$.

\paragraph{First-Fit algorithm.}
The most straightforward strategy for online coloring is to assign each new vertex the smallest available color. This strategy is called First-Fit.

\paragraph{$\ell$-color set.}
We define the notion of ``$\ell$-color set'', a concept used extensively throughout our online coloring algorithms. For a graph $G = (V, E)$ and a vertex set $S \subseteq V$, we say that ``$S$ is an $\ell$-color set in some $k$-coloring of $G$'' if:
\begin{quote}
    There exists a $k$-coloring of $G$, say $\varphi: V \to \{1, \dots, k\}$, such that $|\{\varphi(x): x \in S\}| \leq \ell$, i.e., at most $\ell$ distinct colors appear in $S$. (Note that this does not need to hold for all $k$-colorings $\varphi$.)
\end{quote}
Otherwise, we say that ``$S$ is not an $\ell$-color set in any $k$-coloring of $G$''. Let us observe that when $\ell = k$, $S$ is an $\ell$-color set in some $k$-coloring of $G$ if and only if $G$ is $k$-colorable.

\paragraph{Key inequality for analysis.}
In this paper, we sometimes use the following inequality:
\begin{lemma}
    \label{lem:prelim-ineq}
    Let $b$ be a constant with $0 \leq b \leq 1$. If non-negative real numbers $x_1, \dots, x_k \ (k \leq k_{\max})$ satisfy $x_1 + \dots + x_k \leq s_{\max}$, then $x_1^b + \dots + x_k^b \leq s_{\max}^b \cdot k_{\max}^{1-b}$.
\end{lemma}

\begin{proof}
    For completeness, we give a proof here. Given that $0 \leq b \leq 1$, the function $\varphi(x) = x^b$ is concave. By Jensen's inequality $\varphi(\mathbb{E}[X]) \geq \mathbb{E}[\varphi(X)]$ (where $X$ is a random variable), we have:
    \begin{equation*}
        \left(\frac{x_1 + \dots + x_k}{k}\right)^b \geq \frac{x_1^b + \dots + x_k^b}{k}
    \end{equation*}
    Hence, $x_1^b + \dots + x_k^b \leq (x_1 + \dots + x_k)^b \cdot k^{1-b} \ (\leq s_{\max}^b \cdot k_{\max}^{1-b})$, which proves the lemma. The equality holds if $k = k_{\max}$ and $x_1 = \dots = x_k = \frac{s_{\max}}{k_{\max}}$.
\end{proof}

\section{Online Coloring for $k$-Colorable Graphs}\label{sec:k5}

In this section, we present a deterministic online algorithm that colors any $k$-colorable graph with $\widetilde{O}(n^{1-\frac{1}{k(k-1)/2}})$ colors, for any constant $k \geq 2$. This improves the previous result of $\widetilde{O}(n^{1-\frac{1}{k!}})$ colors by Kierstead \cite{Kie98}.

Technically, we only consider the case when the final number of vertices $n$ for $G$ is known in advance. The following lemma can justify this.

\begin{lemma}
    \label{lem:k5-n-is-known}
    Let $f: \mathbb{N} \to \mathbb{R}_{\geq 0}$ be a non-decreasing unbounded function. For a specific graph class (e.g., $k$-colorable graphs), suppose there exists a deterministic online coloring algorithm for $n$-vertex graphs that uses at most $f(n)$ colors, for every $n \in \mathbb{N}$. Then there exists a deterministic online coloring algorithm for this graph class that uses at most $4f(n)$ colors, even when the final number of vertices, $n$, is unknown.
\end{lemma}

\begin{proof}
    We run the following process for $i = 0, 1, 2, \dots$:
    \begin{quote}
        Let $t_i$ be the largest integer such that $f(t_i) \leq 2^i$ (or $t_i = 0$ if no such integer exists). For the next $t_i$ vertices (or until all vertices arrive), we use an algorithm for $t_i$-vertex graphs, using completely new colors.\footnote{For example, if $f(n) = \sqrt{n}$, we use an algorithm for $1$-vertex graphs for the first $1$ vertex, then an algorithm for $4$-vertex graphs for the next $4$ vertices, then an algorithm for $16$-vertex graphs for the next $16$ vertices, and so on.}
    \end{quote}
    Let $i' := \lceil \log_2 f(n) \rceil$. Since $f(n) \leq 2^{i'}$, we have $t_{i'} \geq n$, so the process with $i > i'$ never occurs. Thus, we used at most $f(t_0) + \dots + f(t_{i'}) \leq 2^0 + \dots + 2^{i'} < 2^{i'+1} = 2^{\lceil \log_2 f(n) \rceil + 1} \leq 4f(n)$ colors.
\end{proof}

\subsection{Research Background}
\label{subsec:k5-background}

We first review Kierstead's \cite{Kie98} framework, in an original restructured way, for the techniques that achieve $\widetilde{O}(n^{2/3})$ colors for the $k = 3$ case and $\widetilde{O}(n^{5/6})$ colors for the $k = 4$ case.

We attempt to color a $k$-colorable graph online, with $\widetilde{O}(n^{1-\varepsilon})$ colors (for some $\varepsilon > 0$). The first step is to perform First-Fit using $\lceil n^{1-\varepsilon} \rceil$ colors. When a vertex $v$ arrives, if $v$ can be colored with any of the colors $1, \dots, \lceil n^{1-\varepsilon} \rceil$, we color $v$ with the smallest possible color. Otherwise, we must color $v$ using a different algorithm. In this case, $|N_{V_{\mathrm{FF}}}(v)| \geq n^{1-\varepsilon}$ holds, where $V_{\mathrm{FF}}$ is the current set of vertices colored by First-Fit. We must employ a strategy to color $v$ by exploiting the dense structure between $V_{\mathrm{FF}}$ and $V \setminus V_{\mathrm{FF}}$.

To this end, this process of coloring the remaining vertices reduces to the following ``S-T problem'', a variant of the online coloring problem. Essentially, the S-T problem corresponds to the setting in which the First-Fit vertices arrive first, and the rest of the vertices arrive later.

\begin{definition}[S-T problem]
    We initially specify a vertex set $S \ (\neq \emptyset)$ and the density parameter $\mu \geq 1$. Initially, a graph $G = (V, E)$ consists of the vertex set $S$ and no edges. Vertices are added to $G$ one by one, together with their incident edges. We denote the current set of added vertices by $T := V \setminus S$. Each added vertex $t \in T$ satisfies $|N_S(t)| \geq \frac{|S|}{\mu}$. Immediately after a vertex $t \in T$ (and its incident edges) is added, we must assign its color $c(t) \in \mathbb{N}$. This color must differ from the colors of all its neighbors in $T$. Note that in the S-T problem, the vertices in $S$ will not be colored. See \autoref{fig:s-t-problem}.
\end{definition}

\begin{figure}[t]
    \centering
    \includegraphics[width=0.9\linewidth]{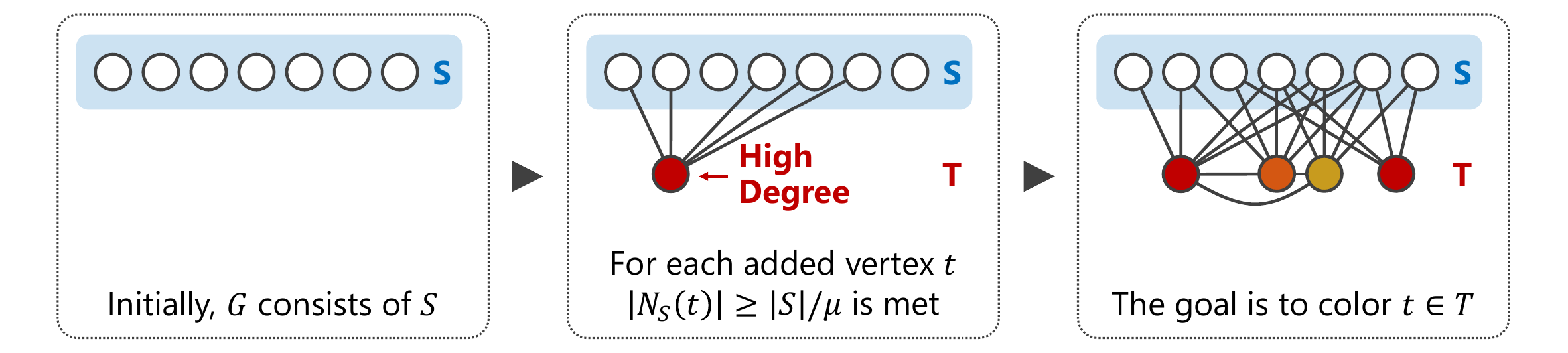}
    \caption{The sketch of the S-T problem. The vertex set $S$, uncolored, is given first, and each new vertex $t \in T$ satisfies $|N_S(t)| \geq \frac{|S|}{\mu}$.}
    \label{fig:s-t-problem}
\end{figure}

\paragraph{On the S-T problem.}
The S-T problem must be treated as an independent problem, not as a ``subroutine'' for the online coloring. It lies in parallel with the online coloring problem; e.g., there are no constraints on the colorability of $G$ by default. Therefore, like we consider online coloring for $k$-colorable graphs, we also consider the S-T problem for $k$-colorable graphs. We solve the original online coloring problem by reducing it to the S-T problem with $|S| = n, |T| \leq n, \mu = n^{\varepsilon}$. We explain the reduction in \autoref{subsec:k5-st-reduction}, which is non-trivial but relatively straightforward.

\paragraph{Using the concept of $\ell$-color set.}
Now, our main focus is to solve the S-T problem for $k$-colorable graphs. Kierstead's idea was to reduce it to an easier family of subproblems: the case of the S-T problem where $S$ is guaranteed to be an $\ell$-color set in some $k$-coloring of $G$. Note that the $\ell = k$ case is the original S-T problem.

Suppose we have a vertex $t' \in T$, and let $S' := N_S(t')$. Then, $S'$ is a $(k-1)$-color set in some $k$-coloring of $G$. Therefore, if a vertex $t \in T$ with sufficiently many neighbors in $S'$ arrives in the future, we can color $t$ using an algorithm for the S-T problem between $S'$ and the set of such vertices $t$, which corresponds to the $\ell = k-1$ case (see \autoref{fig:reduction-example}). 

More specifically, we will solve the $\ell = k-1$ case using the algorithm for $\ell = k-2$, and so on, ultimately reducing to the $\ell = 1$ case. For the $\ell = 1$ case, where $S$ is guaranteed to be a $1$-color set, clearly $G[T]$ is $(k-1)$-colorable. Therefore, in general, we reduce the original problem to ``online coloring for $(k-1)$-colorable graphs'', which can be solved inductively on $k$.

\begin{figure}[t]
    \centering
    \includegraphics[width=0.9\linewidth]{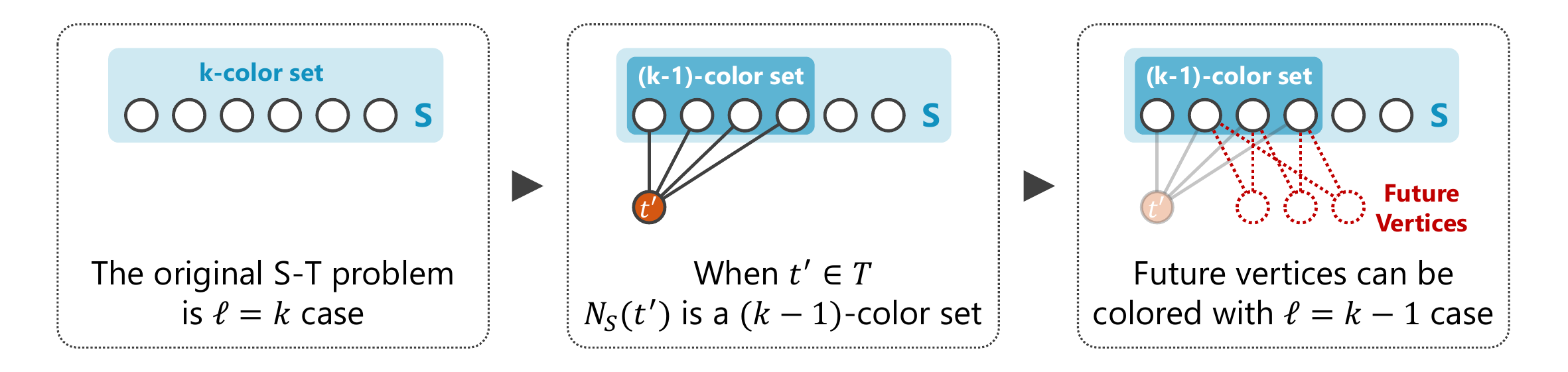}
    \caption{The sketch of the procedure when we reduce the $\ell = k$ problem to the $\ell = k-1$ problem. Note that in the $\ell = t$ case, $S$ is a $t$-color set.}
    \label{fig:reduction-example}
\end{figure}

\paragraph{Locally $k$-colorable graphs.}
In Kierstead's framework, it is necessary to extract a small non-$(k-\ell)$-colorable subgraph from $T$, in order to reduce the $\ell$-color set case to the $(\ell-1)$-color set case. Hence, we need to construct an online coloring algorithm for graphs with no such small subgraphs, which are defined below as \emph{locally $(k-\ell)$-colorable graphs}.

\begin{definition}
    A graph $G = (V, E)$ is locally $k$-colorable if $\chi(G[X]) \leq k$ for all $X \subseteq V$ such that $|X| < 2^{2^k}$.
\end{definition}

Note that locally $k$-colorable graphs are not necessarily $k$-colorable. Indeed, $\chi(G)$ can be arbitrarily large compared to $k$, even for $k = 2$; for example, there exists a graph with arbitrarily large girth and arbitrarily large chromatic number, by a classical result of Erd\H{o}s \cite{Erdos1959}.

For $k = 1$, the graph has no edges, so we can color it online with one color. For $k = 2$, only $O(n^{1/2})$ colors are needed, as shown by the following stronger theorem of Kierstead \cite{Kie98}.

\begin{theorem}[\cite{Kie98}]
    \label{thm:kierstead-2good}
    There is a deterministic online algorithm that colors any graph that contains neither $3$-cycles nor $5$-cycles using only $O(n^{1/2})$ colors.
\end{theorem}

When $k \geq 3$, in the original problem for $k$-colorable graphs, Kierstead could achieve $\widetilde{O}(n^{2/3})$ colors for $k = 3$, $\widetilde{O}(n^{5/6})$ colors for $k = 4$, but did not obtain a comparable result for $k = 5$ (he barely achieved $\widetilde{O}(n^{119/120})$ colors).

\subsection{New Concept: Level-$\ell$ Set}

In this section, we devise an algorithm to color locally $k$-colorable graphs online with $O(n^{1-\frac{1}{k(k-1)/2+1}})$ colors. The starting point is to consider whether we can use a framework similar to that in \autoref{subsec:k5-background}. To adapt this framework, the most important idea is to define the ``locally'' version of $\ell$-color set. We define the corresponding concept, \emph{level-$\ell$ set}, as follows.

\begin{definition}
    For a graph $G = (V, E)$ and a vertex set $S \subseteq V$, we say that ``$S$ is a level-$\ell$ set of $G$'' (in the context of locally $k$-colorable graphs) if:
    \begin{quote}
        For all $X \subseteq V$ such that $|X| < 2^{2^k - 2^{k-\ell} + 1}$, $X \cap S$ is an $\ell$-color set in some $k$-coloring of $G[X]$.
    \end{quote}
    Let us observe that when $\ell = k$, $S$ is a level-$\ell$ set of $G$ if and only if $G$ is locally $k$-colorable.
\end{definition}

Then we can consider a family of subproblems: the S-T problem, where $S$ is guaranteed to be a level-$\ell$ set of $G$. The S-T problem for locally $k$-colorable graphs is equivalent to the $\ell = k$ case. We reduce this problem to the $\ell = k-1$ case, then $\ell = k-2$, ..., and finally to $\ell = 1$. For the $\ell = 1$ case, $G[T]$ is locally $(k-1)$-colorable (as we prove in \autoref{lem:k5-base}). Therefore, in general, we reduce the original problem to ``online coloring for locally $(k-1)$-colorable graphs'', which can be shown inductively on $k$. We explain the details later, but this is the general strategy to solve the problem.

\paragraph{Remarks on thresholds.}
The reader may wonder why the size thresholds for locally $k$-colorable graphs and level-$\ell$ sets are double-exponential. As we detail in \autoref{subsec:k5-remarks}, when local colorability fails, our algorithm recursively extracts uncolorable witness subgraphs. This double-exponential limit is the strict combinatorial ceiling required to absorb the cascading union bounds of these recursive extractions (we calculate the exact maximum witness sizes in \autoref{tab:k5-threshold}).

\subsection{The Big Picture}

In order to obtain an algorithm, we formally define the problems we need to solve.

\begin{itemize}
    \item Problem $\textsf{LOC}(k)$: The online coloring problem where, if $G$ is no longer locally $k$-colorable (i.e., we find $X \subseteq V$ such that $|X| < 2^{2^k}, \chi(G[X]) > k$), we can terminate the problem by outputting $X$ as a witness.
    \item Problem $\textsf{ST}(k)$: The S-T problem (with density parameter $\mu$) where, if $G$ is no longer locally $k$-colorable (i.e., we find $X \subseteq V$ such that $|X| < 2^{2^k}, \chi(G[X]) > k$), we can terminate the problem by outputting $X$ as a witness. We also denote this problem by $\textsf{ST}(k, \mu)$.
    \item Problem $\textsf{SP}(k, \ell)$: The S-T problem (with density parameter $\mu$) where, if $S$ is no longer a level-$\ell$ set of $G$ (i.e., we find $X \subseteq V$ such that $|X| < 2^{2^k - 2^{k-\ell} + 1}$ and $X \cap S$ is not an $\ell$-color set in any $k$-coloring of $G[X]$), we can terminate the problem by outputting $X$ as a witness. We also denote this problem by $\textsf{SP}(k, \ell, \mu)$.
\end{itemize}

Essentially, $\textsf{LOC}(k)$ is the online coloring problem for locally $k$-colorable graphs, $\textsf{ST}(k)$ is the S-T problem for locally $k$-colorable graphs, and $\textsf{SP}(k, \ell)$ is the S-T problem where $S$ is guaranteed to be a level-$\ell$ set of $G$. However, for our purpose, as given above, we allow inputs that violate these assumptions. To handle such cases, we introduce exception handling: ``output $X$ as a witness and terminate the problem''. \textbf{Even in this case, we need to color the last vertex we receive.}

Additionally, we define the problem $\textsf{LOC}'(k)$ to be the version of $\textsf{LOC}(k)$ where the final number of vertices $n$ is known in advance. By \autoref{lem:k5-n-is-known}, $\textsf{LOC}'(k)$ is equally difficult as $\textsf{LOC}(k)$.

\paragraph{The proof roadmap.}
The goal is to solve $\textsf{LOC}(k)$ deterministically for constant $k \geq 2$. As a premise, we assume that the easier problems, $\textsf{LOC}(0), \dots, \textsf{LOC}(k-1)$, can all be solved with $O(n^{1-\alpha})$ colors ($0 \leq \alpha \leq 1$). The proof roadmap, as depicted in \autoref{fig:loc-roadmap}, is to design an algorithm to solve $\textsf{SP}(k, 1)$ first, then $\textsf{SP}(k, 2), \dots, \textsf{SP}(k, k)$, and finally $\textsf{LOC}(k)$.

\begin{figure}[b]
    \centering
    \includegraphics[width=0.9\linewidth]{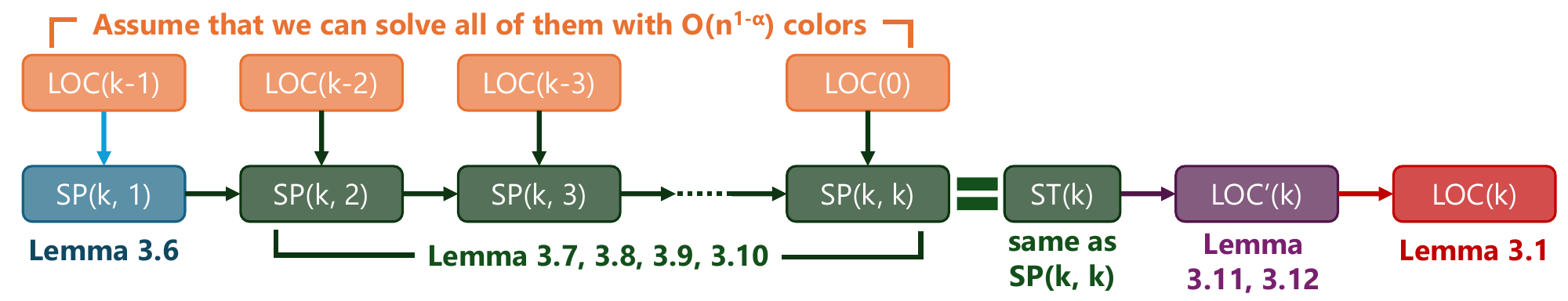}
    \caption{The roadmap for solving $\textsf{LOC}(k)$; e.g., we solve $\textsf{SP}(k, 2)$ using $\textsf{LOC}(k-2)$ and $\textsf{SP}(k, 1)$.}
    \label{fig:loc-roadmap}
\end{figure}

We will eventually achieve $O(|T|^{1-\alpha} \mu^{(\ell-1)\alpha})$ colors for $\textsf{SP}(k, \ell, \mu)$, and $O(n^{1-\frac{\alpha}{1+(k-1)\alpha}})$ colors for $\textsf{LOC}(k)$. Once these results are shown, given that $\textsf{LOC}(0), \textsf{LOC}(1)$ can be solved with $O(1)$ colors, we can calculate that $\textsf{LOC}(2)$ can be solved in $O(n^{1/2})$ colors (by $\alpha = 1$), $\textsf{LOC}(3)$ in $O(n^{3/4})$ colors (by $\alpha = \frac{1}{2}$), ..., and $\textsf{LOC}(k)$ in $O(n^{1-\frac{1}{k(k-1)/2+1}})$ colors for any constant $k \geq 2$. Note that since every $k$-colorable graph is locally $k$-colorable, we also achieve $O(n^{1-\frac{1}{k(k-1)/2+1}})$ colors for online coloring of $k$-colorable graphs (slightly worse than our best result).

\paragraph{Important remarks.}
Each problem listed above will be solved by creating instances $I$ of another problem. We denote the graph $G$ and the vertex sets $V, S, T$ for this instance as $G(I), V(I), S(I), T(I)$.

\subsection{Solution for $\textsf{SP}(k, 1)$}

We first solve $\textsf{SP}(k, 1)$. The idea is simply to color the vertices in $T$ using $\textsf{LOC}(k-1)$.

\begin{lemma}
    \label{lem:k5-base}
    $\textnormal{\textsf{SP}}(k, 1)$ can be solved with $O(|T|^{1-\alpha})$ colors.
\end{lemma}

\begin{proof}
    We color the vertices in $T$ by running $\textsf{LOC}(k-1)$. When $\textsf{LOC}(k-1)$ terminates by outputting $X_T \subseteq T$, it satisfies $|X_T| < 2^{2^{k-1}}, \chi(G[X_T]) > k-1$. For each $t \in X_T$, let $\mathrm{up}(t)$ be one of the vertices in $N_S(t)$. This can be defined because $|N_S(t)| \geq \frac{|S|}{\mu} > 0$. Let $X_S := \{\mathrm{up}(t) : t \in X_T\}$.

    If there exists a $k$-coloring of $G[X_S \cup X_T]$ in which all vertices in $X_S$ have the same color, then the vertices in $X_T$ must have the remaining $k-1$ colors. However, this contradicts $\chi(G[X_T]) > k-1$. Therefore, $X_S$ is not a $1$-color set in any $k$-coloring of $G[X_S \cup X_T]$. Hence, we can output $X := X_S \cup X_T$ and terminate $\textsf{SP}(k, 1)$. The size $|X| \leq 2 |X_T| < 2^{2^{k-1}+1}$ also satisfies the requirement. Since $\textsf{LOC}(k-1)$ uses $O(n^{1-\alpha})$ colors, we only use $O(|T|^{1-\alpha})$ colors to solve $\textsf{SP}(k, 1)$.
\end{proof}

\subsection{Solution for $\textsf{SP}(k, \ell)$ for $\ell \geq 2$: Algorithm}

Next, we describe an algorithm for solving $\textsf{SP}(k, \ell, \mu)$ where $\ell \in \{2, \dots, k\}$.

\paragraph{Overview.}
We assume that we already have a solution for $\textsf{SP}(k, \ell-1)$. First, we color the vertices in $T$ by running $\textsf{LOC}(k-\ell)$. The key intuition is as follows:

\begin{quote}
    Suppose we find a small subset $X_T \subseteq T$ such that $\chi(G[X_T]) > k-\ell$. Then, as long as $S$ itself is a level-$\ell$ set, the subset $N_S(t)$ forms a level-$(\ell-1)$ set for at least one vertex $t \in X_T$. These subsets are easier to handle (inductively on $\ell$) and lead to some profit for future vertices.
\end{quote}

When $\textsf{LOC}(k-\ell)$ terminates by outputting $X_T \subseteq T$, then for each vertex $t' \in X_T$, we simultaneously initiate a new instance $I$ of the problem $\textsf{SP}(k, \ell-1, 2\mu)$ where $S(I) := N_S(t')$. We use completely disjoint color sets for these instances. Note that we set the density parameter to $2\mu$ for these instances. Here, we create $|X_T| < 2^{2^{k-\ell}}$ instances, which is $O(1)$ (and hence we only pay an $O(1)$ factor for the number of colors used). These instances are used for processing future vertices: when a new vertex $t \in T$ arrives and satisfies $|N_{S(I)}(t)| \geq \frac{|S(I)|}{2\mu}$, we can color $t$ using the instance $I$. Importantly, as long as $S$ is a level-$\ell$ set of $G$, at least one instance $I$ survives forever (\autoref{lem:k5-sp-output}). We also initiate a new instance of $\textsf{LOC}(k-\ell)$ for coloring the rest of the vertices in $T$, using completely different colors. If the new $\textsf{LOC}(k-\ell)$ also terminates, we do the same procedure again. The key is that this procedure occurs at most $2\mu$ times (\autoref{lem:k5-sp-size}). This is why we can solve the problem more effectively when $\mu$ is small.

Overall, when a new vertex $t \in T$ arrives, we color $t$ using an ``active'' instance (which is not yet terminated) of $\textsf{SP}(k, \ell-1, 2\mu)$ if possible, or using $\textsf{LOC}(k-\ell)$ otherwise; see \autoref{fig:k5-sp-procedure}.

\begin{figure}[htbp]
  \centering
  \includegraphics[width=\linewidth]{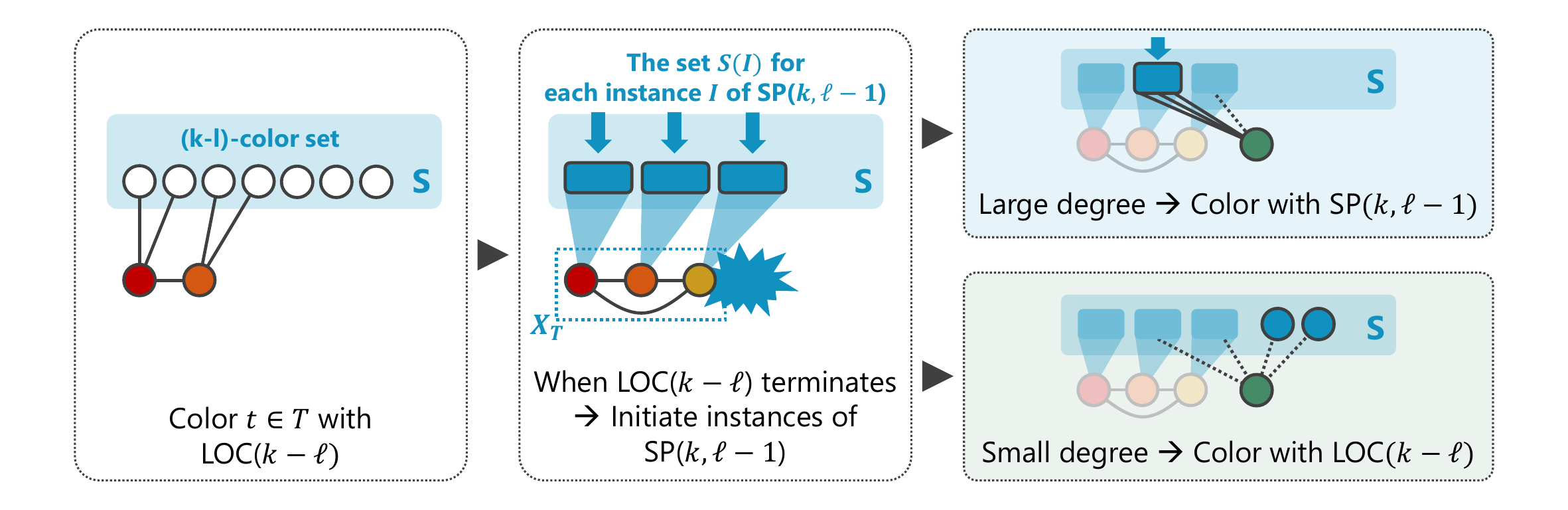}
  \caption{The sketch of the procedure when $\ell = k-2$. In the middle and the right figure, each big rounded blue square in $S$ is the set of adjacent vertices for each $t \in X_T$. Note that when $\ell = k-2$ and $\textsf{LOC}(k-\ell)$ terminates, the set $X_T$ is not 2-colorable.} \label{fig:k5-sp-procedure}
\end{figure}

\begin{figure}[t]
    \centering
    \includegraphics[width=1\linewidth]{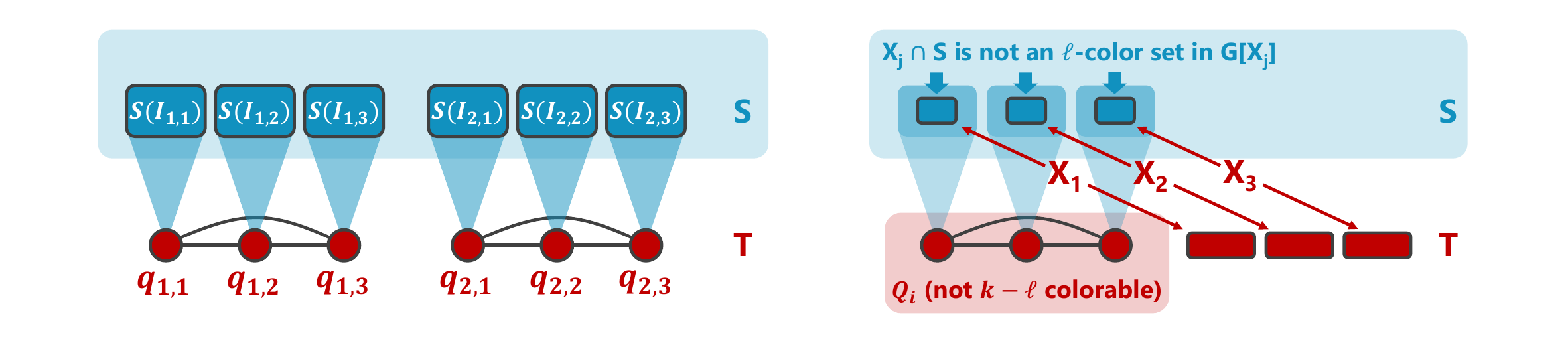}
    \caption{Left: an example of relationships between $Q_i$ and $S(I_{i,j})$. Right: an example of the graph $G[X]$ for the output $X$ when $\textsf{SP}(k, \ell)$ terminates, for $k-\ell = 2$. The output set $X$ consists all of the dark (red and dark blue) vertices.}
    \label{fig:k5-sp-output}
\end{figure}

\paragraph{Implementation.}
We denote the created instances of $\textsf{LOC}(k-\ell)$ by $I^{\textsf{LOC}}_1, I^{\textsf{LOC}}_2, \dots$, in the chronological order. When the $r$-th instance $I^{\textsf{LOC}}_r$ terminates, we let $Q_r = \{q_{r,1}, \dots, q_{r,|Q_r|}\}$ for its output, and $I_{r,1}, \dots, I_{r,|Q_r|}$ for the instances of $\textsf{SP}(k, \ell-1, 2\mu)$ created at that time, so that $S(I_{r,j}) = N_S(q_{r,j})$ for each $j$; see \autoref{fig:k5-sp-output} left (note that the $(r+1)$-th instance $I^{\textsf{LOC}}_{r+1}$ is created after this event). \textbf{We use completely disjoint colors for all of these instances.}

The algorithm for solving $\textsf{SP}(k, \ell, \mu)$ starts by initiating the instance $I^{\textsf{LOC}}_1$. Whenever a vertex $t \in T$ arrives, we do the following procedure:
\begin{enumerate}
    \item If there exists an active instance $I_{i,j}$ (i.e., not terminated) such that $|N_{S(I_{i,j})}(t)| \geq \frac{|S(I_{i,j})|}{2\mu}$:
    \begin{enumerate}
        \item We choose one such instance $I_{i,j}$ arbitrarily, and color $t$ using this instance $I_{i,j}$ (line 4). Specifically, we add $t$ (and its incident edges toward $S(I_{i,j}) \cup T(I_{i,j})$) to $I_{i,j}$, and let it color (using our algorithm for $\textsf{SP}(k, \ell-1)$). Back in the original problem, we color $t$ according to its color in $I_{i,j}$.
        \item If $I_{i,1}, \dots, I_{i,|Q_i|}$ are all terminated (i.e., $I_{i,j}$ terminates as a result of step 1 (a) and this was the last active instance among them), then $S$ is no longer a level-$\ell$ set of $G$. We output $X := X_1 \cup \dots \cup X_{|Q_i|} \cup Q_i$ and terminate the original problem $\textsf{SP}(k, \ell, \mu)$, where $X_j$ is the output set of $I_{i,j}$ (lines 5--7); see \autoref{fig:k5-sp-output} right.
    \end{enumerate}
    \item Otherwise, we color $t$ using instance $I^{\textsf{LOC}}_r$, where $I^{\textsf{LOC}}_r$ is the currently running instance of $\textsf{LOC}(k-\ell)$ (line 9). Specifically, we add $t$ (and its incident edges toward $V(I^{\textsf{LOC}}_r)$) to $I^{\textsf{LOC}}_r$, and let it color (using our algorithm for $\textsf{LOC}(k-\ell)$). If $I^{\textsf{LOC}}_r$ terminates as a result, let $Q_r = \{q_{r,1}, \dots, q_{r,|Q_r|}\}$ be its output. Then, for each $j = 1, \dots, |Q_r|$, we initiate an instance $I_{r,j}$ of $\textsf{SP}(k, \ell-1, 2\mu)$ so that $S(I_{r,j}) := N_S(q_{r,j})$ (line 11). We also initiate a new instance $I^{\textsf{LOC}}_{r+1}$ of $\textsf{LOC}(k-\ell)$, which increments $r$ by $1$ (line 12).
\end{enumerate}
The pseudocode of this algorithm is given in \autoref{alg:k5-sp}.

\begin{algorithm}[t]
\caption{An algorithm for solving $\textsf{SP}(k, \ell, \mu)$ where $\ell \in \{2, \dots, k\}$}
\begin{algorithmic}[1]
    \State $I^{\mathrm{LOC}}_1 \gets$ a new instance of $\textsf{LOC}(k-\ell)$, and $r \gets 1$
    \For {each arrival of $t \in T$}
        \If {there exists an active instance $I_{i,j}$ such that $|N_{S(I_{i,j})}(t)| \geq \frac{|S(I_{i,j})|}{2\mu}$}
            \State color $t$ using $I_{i,j}$
            \If {$I_{i,1}, \dots, I_{i,|Q_i|}$ are all terminated}
                \State $X_j \gets$ the output set for $I_{i, j}$, for $j = 1, \dots, |Q_i|$
                \State \Return $X \gets (X_1 \cup \cdots \cup X_{|Q_i|}) \cup Q_i$ and terminate $\textsf{SP}(k, \ell, \mu)$
            \EndIf
        \Else
            \State color $t$ using $I^{\mathrm{LOC}}_r$
            \If {$I^{\mathrm{LOC}}_r$ terminates by outputting $Q_r = \{q_{r,1}, \dots, q_{r,|Q_r|}\} \subseteq T$}
                \State $I_{r,j} \gets$ a new instance of $\textsf{SP}(k, \ell-1, 2\mu)$ where $S(I_{r,j}) := N_S(q_{r,j})$, for $j = 1, \dots, |Q_r|$
                \State $I^{\mathrm{LOC}}_{r+1} \gets$ a new instance of $\textsf{LOC}(k-\ell)$, and $r \gets r+1$
            \EndIf
        \EndIf
    \EndFor
\end{algorithmic}
\label{alg:k5-sp}
\end{algorithm}

\subsection{Solution for $\textsf{SP}(k, \ell)$ for $\ell \geq 2$: Correctness and Analysis}

First, we prove that the algorithm for $\textsf{SP}(k, \ell)$ works correctly, that is, the output set $X$ satisfies the conditions.

\begin{lemma}
    \label{lem:k5-sp-output}
    If instances $I_{i,1}, \dots, I_{i,|Q_i|}$ are all terminated, the output set $X := X_1 \cup \dots \cup X_{|Q_i|} \cup Q_i$ satisfies the requirements.
\end{lemma}

\begin{proof}
    Suppose that $X \cap S$ is an $\ell$-color set in some $k$-coloring of $G[X]$. Then, there exists a $k$-coloring of $G[X]$, say $\varphi: X \to \{1, \dots, k\}$, such that $\varphi(s) \in \{1, \dots, \ell\}$ for all $s \in X \cap S$. By the assumption of $\textsf{LOC}(k-\ell)$, we have $\chi(G[Q_i]) > k-\ell$, so there exists $q_{i, j} \in Q_i$ such that $\varphi(q_{i,j}) \in \{1, \dots, \ell\}$.

    For such $q_{i,j}$, each $s \in X_j \cap S$ has color $\varphi(s) \in \{1, \dots, \ell\} \setminus \{\varphi(q_{i,j})\}$, because $s$ is adjacent to $q_{i,j}$ (since $s \in X_j \cap S \subseteq S(I_{i,j}) = N_S(q_{i,j})$). Hence, $X_j \cap S$ is an $(\ell-1)$-color set in some $k$-coloring of $G[X_j]$. However, this contradicts the assumption of $\textsf{SP}(k, \ell-1)$. Therefore, $X \cap S$ is not an $\ell$-color set in any $k$-coloring of $G[X]$.

    Next, we show that the size of $X$ satisfies the requirement. We have the following:
    \begin{equation*}
        |X| \leq (|X_1| + \cdots + |X_{|Q_i|}|) + |Q_i| \leq \left(2^{2^k-2^{k-\ell+1}+1}-1\right) \cdot \left(2^{2^{k-\ell}}-1\right) + \left(2^{2^{k-\ell}}-1\right) < 2^{2^k-2^{k-\ell}+1}
    \end{equation*}
    Note that $|X_j| < 2^{2^k-2^{k-\ell+1}+1}$ and $|Q_i| < 2^{2^{k-\ell}}$ follow from the assumptions for $\textsf{SP}(k, \ell-1)$ and $\textsf{LOC}(k-\ell)$, respectively.
\end{proof}

Next, we analyze the number of colors used in this procedure. The important point is to see how many times the instances of $\textsf{SP}(k, \ell-1)$ are created. The following lemma is the key to bounding this number.

\begin{lemma}
    \label{lem:k5-sets}
    Let $\mu \geq 1$ be a real number. Let $A_1, \dots, A_m \subseteq S$ be sets such that $|A_i| \geq \frac{|S|}{\mu}$ for each $i$, and $|A_i \cap A_j| \leq \frac{|A_i|}{2\mu}$ for each $i, j$ with $i < j$. Then, $m < 2\mu$ holds.
\end{lemma}

\begin{proof}
    We analyze the size of $A_1 \cup \dots \cup A_t$ for an integer $t \leq m$. If $1 - \frac{t-1}{2\mu} \geq 0$, then we have:
    \begin{multline*}
        |A_1 \cup \dots \cup A_t| = \sum_{i=1}^t |A_i \setminus (A_1 \cup \dots \cup A_{i-1})| \geq \sum_{i=1}^t (|A_i| - |A_i \cap A_1| - \dots - |A_i \cap A_{i-1}|) \\
        \geq \sum_{i=1}^t \left(|A_i| - \frac{|A_1|}{2\mu} - \dots - \frac{|A_{i-1}|}{2\mu}\right) \geq \sum_{i=1}^t \left(1-\frac{t-i}{2\mu}\right) |A_i| \geq \sum_{i=1}^t \left(1-\frac{t-i}{2\mu}\right) \frac{|S|}{\mu} = \frac{t(4\mu-t+1)}{4\mu^2} |S|
    \end{multline*}
    By substituting $t = \lceil 2\mu \rceil$ (which is valid when $m \geq 2\mu$), we have $|A_1 \cup \dots \cup A_t| > |S|$. This is a contradiction. Therefore, $m < 2\mu$ holds.
\end{proof}

\begin{lemma}
    \label{lem:k5-sp-size}
    $r < 2\mu + 1$.
\end{lemma}

\begin{proof}
    While the algorithm for $\textsf{SP}(k, \ell)$ has not yet terminated, for each $i = 1, \dots, r-1$, there exists an active instance $I_{i,t_i}$ for some $t_i$. Let $A_i := S(I_{i,t_i}) \ (= N_S(q_{i,t_i}))$.
    
    For each $i$, we have $|A_i| \geq \frac{|S|}{\mu}$, by the input condition for $\textsf{SP}(k, \ell, \mu)$. Also, for each $i, j \ (i < j)$, we have $|A_i \cap A_j| \leq \frac{|A_i|}{2\mu}$. This is because, given that the vertex $q_{j,t_j}$ arrives after the instance $I_{i,t_i}$ is created, if $|A_i \cap A_j| > \frac{|A_i|}{2\mu}$, then $q_{j,t_j}$ should have been colored using the instance $I_{i,t_i}$, which is a contradiction. Therefore, by applying \autoref{lem:k5-sets}, we obtain $r < 2\mu+1$.
\end{proof}

\begin{lemma}
    \label{lem:k5-sp}
    We can solve $\textnormal{\textsf{SP}}(k, \ell)$ with $O(|T|^{1-\alpha} \mu^{(\ell-1)\alpha})$ colors, for $\ell = 1, \dots, k$.
\end{lemma}

\begin{proof}
    We prove by induction on $\ell$. The $\ell = 1$ case is already shown in \autoref{lem:k5-base}. For $\ell \geq 2$, when solving $\textsf{SP}(k, \ell)$, the number of instances of $\textsf{SP}(k, \ell-1)$ we have created is at most $(r-1) \cdot \max_i |Q_i| \leq 2\mu \cdot 2^{2^{k-\ell}} = O(\mu)$. By the induction hypothesis, each instance $I_{i,j}$ uses $O(|T(I_{i,j})|^{1-\alpha} (2\mu)^{(\ell-2)\alpha})$ colors. The sum of these values is bounded by $O(|T|^{1-\alpha} \mu^{(\ell-1)\alpha})$, by $\sum_{i,j} |T(I_{i,j})| \leq |T|$ and \autoref{lem:prelim-ineq}. Also, the number of instances of $\textsf{LOC}(k-\ell)$ we have created is at most $r < 2\mu+1$. Since $\textsf{LOC}(k-\ell)$ can be solved with $O(n^{1-\alpha})$ colors, the total number of colors used here is $O(|T|^{1-\alpha} \mu^{\alpha})$, by the same argument. Therefore, we only use $O(|T|^{1-\alpha} \mu^{(\ell-1)\alpha})$ colors in total to solve $\textsf{SP}(k, \ell)$.
\end{proof}

\subsection{From the S-T Problem to Online Coloring}
\label{subsec:k5-st-reduction}

The remaining task is to adapt the framework from the S-T problem to the online coloring problem. To solve $\textsf{LOC}'(k)$, we set the parameter $\mu \geq 1$ and apply the First-Fit strategy using the first $\lceil \frac{n}{\mu} \rceil$ colors. We use the S-T problem framework to color the vertices that are not colored by First-Fit. In this subsection, we formally describe the algorithm to solve $\textsf{LOC}'(k)$ for $k \geq 2$.

First, immediately after $\textsf{LOC}'(k)$ starts, we create an instance $I$ of the problem $\textsf{ST}(k, \mu)$; we prepare the ``auxiliary'' vertices $s_1, \dots, s_n$, and set $S(I) := \{s_1, \dots, s_n\}$. We denote the vertices that will arrive at $T(I)$ by $t_1, t_2, \dots$, in the arrival order. The vertices in $\textsf{LOC}'(k)$ and in $\textsf{ST}(k, \mu)$ are different, but they are corresponding. We denote the current correspondence by $\sigma: V \to \{s_1, \dots, s_n, t_1, t_2, \dots\}$. After that, for each arrival of $v \in V$ in $\textsf{LOC}'(k)$, we do the following:

\begin{enumerate}
    \item If $v$ can be colored by colors $1, \dots, \lceil \frac{n}{\mu} \rceil$, we color $v$ with the smallest available color (First-Fit). Then, let $\sigma(v) := s_i$, where $v$ is the $i$-th vertex colored by First-Fit.
    \item Otherwise, let $\sigma(v) := t_i$, where $v$ is the $i$-th vertex not colored by First-Fit. Then:
    \begin{itemize}
        \item Add vertex $t_i$ to the instance $I$, along with edge $\{t_i, \sigma(v')\}$ for each $v' \in N(v)$ (the vertices adjacent to $v$ in the original problem), and let it color (using our algorithm for $\textsf{ST}(k)$, which is equivalent to $\textsf{SP}(k, k)$). This is a valid input to $I$; $|N_S(t_i)| \geq \frac{|S|}{\mu}$ must hold, because $v$ is adjacent to at least $\frac{n}{\mu}$ First-Fit vertices.
        \item We color $v$ with color $\lceil \frac{n}{\mu} \rceil + c$, where $c$ is the color of vertex $t_i$ in the instance $I$.
    \end{itemize}
    \item If $I$ terminates by outputting $X_I \subseteq S(I) \cup T(I)$, we terminate $\textsf{LOC}'(k)$ by outputting $X := \{u \in V : \sigma(u) \in X_I\}$.
\end{enumerate}

Note that, for each $u, v \in V$, the condition that edge $\{\sigma(u), \sigma(v)\}$ is added to $I$ is not only because $\{u, v\} \in E$, but also because the vertex which arrives later among $u$ and $v$ is not colored by First-Fit. This means that some edges are deleted when they are added to $I$. Next, we prove the correctness.

\begin{lemma}
    \label{lem:k5-loc-output}
    The algorithm for $\textnormal{\textsf{LOC}}'(k)$ colors the graph correctly, and the output $X$ satisfies the condition.
\end{lemma}

\begin{proof}
    We prove that two adjacent vertices $ u$ and $ v$ are colored with different colors. Let $V_{\mathrm{FF}}$ be the set of vertices colored by First-Fit.
    \begin{itemize}
        \item Case $u, v \in V_{\mathrm{FF}}$: Suppose $u$ comes before $v$. When we color $v$ with First-Fit, we choose a color different from $u$'s.
        \item Case $u, v \notin V_{\mathrm{FF}}$: In $I$, $\sigma(u)$ and $\sigma(v)$ must be colored by different colors, because edge $\{\sigma(u), \sigma(v)\}$ is added to $I$. Therefore, $u$ and $v$ are colored with different colors.
        \item Otherwise: The vertices in $V_{\mathrm{FF}}$ are colored with color $\lceil \frac{n}{\mu} \rceil$ or earlier, and the other vertices are colored with color $\lceil \frac{n}{\mu} \rceil + 1$ or later. Hence, $u$ and $v$ are colored with different colors.
    \end{itemize}
    
    We also prove that the output set $X$ satisfies the condition. First, $(G(I))[X_I]$ is not $k$-colorable, by the assumption of $\mathsf{ST}(k)$. Let $X'_I$ be the set of vertices obtained by removing from $X_I$ those vertices not yet mapped by $\sigma$. Then, $(G(I))[X_I']$ is also not $k$-colorable, because we only remove isolated vertices. Moreover, $G[X]$ is a graph obtained by adding some edges to a graph isomorphic to $(G(I))[X_I']$, so $G[X]$ is also not $k$-colorable. Also, $|X| = |X'_I| \leq |X_I| < 2^{2^k}$. Therefore, $X$ satisfies the condition to terminate $\textsf{LOC}'(k)$.
\end{proof}

Finally, we show the number of colors for $\textsf{LOC}'(k)$ by setting the parameter $\mu$ optimally.

\begin{lemma}
    \label{lem:k5-loc}
    $\textnormal{\textsf{LOC}}'(k)$ can be solved with $O(n^{1-\frac{\alpha}{1+(k-1)\alpha}})$ colors. The same result also holds for $\textsf{LOC}(k)$.
\end{lemma}

\begin{proof}
    $\textsf{ST}(k)$ is equivalent to $\textsf{SP}(k, k)$, so by \autoref{lem:k5-sp}, we can solve $\textsf{ST}(k)$ with $O(|T|^{1-\alpha} \mu^{(k-1)\alpha})$ colors. To solve $\textsf{LOC}'(k)$, we need additional $\lceil \frac{n}{\mu} \rceil$ colors for First-Fit. Therefore, by $|T| \leq n$, the total number of colors we use is $O(n^{1-\alpha} \mu^{(k-1)\alpha} + \frac{n}{\mu})$. By setting $\mu = n^{\frac{\alpha}{1+(k-1)\alpha}}$, the number of colors becomes $O(n^{1-\frac{\alpha}{1+(k-1)\alpha}})$. With \autoref{lem:k5-n-is-known}, we can also solve $\textsf{LOC}(k)$ using $O(n^{1-\frac{\alpha}{1+(k-1)\alpha}})$ colors.
\end{proof}

\begin{theorem}
    \label{thm:k5-loc}
    For all $k \geq 0$, $\textnormal{\textsf{LOC}}(k)$ can be solved with $O(n^{1-\frac{1}{k(k-1)/2+1}})$ colors.
\end{theorem}

\begin{proof}
    First, we can solve $\textsf{LOC}(0)$ with $1$ color, because when the first vertex $v$ arrives, we can output $X := \{v\}$ and terminate the problem. Also, we can solve $\textsf{LOC}(1)$ with $2$ colors; we can use the same color until an edge appears, and when an edge $\{u, v\}$ appears, we can output $X := \{u, v\}$ and terminate the problem (we color the last vertex using the second color).

    Next, we consider the case $k \geq 2$. We prove by induction on $k$. \autoref{lem:k5-loc} shows that $\textsf{LOC}(k)$ can be solved with $O(n^{1-\frac{\alpha}{1+(k-1)\alpha}})$ colors, where we can apply this lemma with $\alpha = \frac{1}{(k-1)(k-2)/2+1}$, by the assumption of induction. Therefore, $\textsf{LOC}(k)$ can be solved with $O(n^{1-\frac{1}{k(k-1)/2+1}})$ colors.
\end{proof}

\subsection{The Improvement}
\label{subsec:k5-improvement}

In the previous subsections, we presented a deterministic online algorithm for coloring locally $k$-colorable graphs using $O(n^{1-\frac{1}{k(k-1)/2+1}})$ colors. This result means that we can also color $k$-colorable graphs with $O(n^{1-\frac{1}{k(k-1)/2+1}})$ colors. Compared to the current best result of $O(n^{1-\frac{1}{k!}})$ colors, the factorial term became quadratic, which is already a big improvement.

We can further improve this result. The key is that, while only an $O(n^{1/2})$-color algorithm for $\textsf{LOC}(2)$ has been found in the previous subsections, there is an $O(\log n)$-color algorithm for bipartite graphs ($k = 2$ case), by Lov\'{a}sz, Saks, Trotter \cite{LST89}. We improve the exponent by using the following result as a base case.

\begin{theorem}[\cite{LST89}]
    \label{thm:bipartite}
    There exists a deterministic online algorithm to color bipartite graphs (i.e., $2$-colorable graphs) with $2 \log_2 (n+1)$ colors.
\end{theorem}

\paragraph{Problems to solve.}
To construct an algorithm, we formally define the problems we need to solve. The key is that we consider the S-T problem in which $S$ is an $\ell$-color set in some $k$-coloring of $G$. The ultimate goal is to solve $\textsf{COL}(k)$ for $k \geq 3$.

\begin{itemize}
    \item Problem $\textsf{COL}(k)$: The online coloring problem where, if $G$ is no longer $k$-colorable, we can terminate the problem by declaring this fact.
    \item Problem $\textsf{ST}^+(k)$: The S-T problem where, if $G$ is no longer $k$-colorable, we can terminate the problem by declaring this fact. We also denote this problem by $\textsf{ST}^+(k, \mu)$.
    \item Problem $\textsf{SP}^+(k, \ell)$: The S-T problem where, if $S$ is no longer an $\ell$-color set in any $k$-coloring of $G$, we can terminate the problem by declaring this fact. We also denote this problem by $\textsf{SP}^+(k, \ell, \mu)$.
\end{itemize}

\paragraph{The improved algorithm.}
We follow the framework used to solve the problem for locally $k$-colorable graphs. First, we can solve $\textsf{SP}^+(k, 1)$ by coloring $T$ with $\textsf{COL}(k-1)$ (instead of $\textsf{LOC}(k-1)$). This is possible because, if $G[T]$ becomes non-$(k-1)$-colorable, then $S$ is not a $1$-color set in any $k$-coloring of $G$ (by a similar argument to \autoref{lem:k5-base}). Next, we can solve $\textsf{SP}^+(k, \ell)$ (for $\ell \geq 2$) by the same algorithm as \autoref{alg:k5-sp}, that creates instances of $\textsf{SP}^+(k, \ell-1)$ (instead of $\textsf{SP}(k, \ell-1)$). Note that we color vertices in $T$ with $\textsf{LOC}(k-\ell)$ also in this case. When instances $I_{i, 1}, \dots, I_{i, |Q_i|}$ are all terminated, we simply declare that ``$S$ is not an $\ell$-color set in any $k$-coloring of $G$'' (instead of creating the output $X$). This is justified by the following lemma.

\begin{lemma}
    In the algorithm for $\textnormal{\textsf{SP}}^+(k, \ell)$ with $\ell \geq 2$, if instances $I_{i, 1}, \dots, I_{i, |Q_i|}$ are all terminated, then $S$ is not an $\ell$-color set in any $k$-coloring of $G$.
\end{lemma}

\begin{proof}
    Suppose that $S$ is an $\ell$-color set in some $k$-coloring of $G$. Then, there exists a $k$-coloring of $G$, say $\varphi: V \to \{1, \dots, k\}$, such that $\varphi(s) \in \{1, \dots, \ell\}$ for all $s \in S$. By the assumption of $\textsf{LOC}(k-\ell)$, we have $\chi(G[Q_i]) > k-\ell$, so there exists $q_{i, j} \in Q_i$ that $\varphi(q_{i,j}) \in \{1, \dots, \ell\}$.

    For such $q_{i,j}$, each $s \in S(I_{i,j})$ has color $\varphi(s) \in \{1, \dots, \ell\} \setminus \{\varphi(q_{i,j})\}$, because $s$ is adjacent to $q_{i,j}$ due to $s \in S(I_{i,j}) = N_S(q_{i,j})$. Hence, $S(I_{i,j})$ is an $(\ell-1)$-color set in some $k$-coloring of $G[S(I_{i,j}) \cup T(I_{i,j})]$. However, this contradicts the assumption of $\textsf{SP}^+(k, \ell-1)$. Therefore, $S$ is not an $\ell$-color set in any $k$-coloring of $G$.
\end{proof}

$\textsf{ST}^+(k)$ is equivalent to $\textsf{SP}^+(k, k)$. We solve $\textsf{COL}'(k)$ using the same algorithm as in \autoref{subsec:k5-st-reduction}, that creates an instance of $\textsf{ST}^+(k)$ (instead of $\textsf{ST}(k)$). If this instance terminates, then $G$ is no longer $k$-colorable (by a similar argument to \autoref{lem:k5-loc-output}), so we can terminate the problem. Given this algorithm for $\textsf{COL}'(k)$, we can solve $\textsf{COL}(k)$ via \autoref{lem:k5-n-is-known}.

Next, we analyze the number of colors used.

\begin{lemma}
    \label{lem:k5-col}
    Suppose that $\textnormal{\textsf{COL}}(k-1)$ and $\textnormal{\textsf{LOC}}(0), \dots, \textnormal{\textsf{LOC}}(k-2)$ can all be solved with $\widetilde{O}(n^{1-\alpha})$ colors $(0 \leq \alpha \leq 1)$. Then, $\textnormal{\textsf{SP}}^+(k, \ell, \mu)$ can be solved with $\widetilde{O}(|T|^{1-\alpha} \mu^{(\ell-1)\alpha})$ colors, and $\textnormal{\textsf{COL}}(k)$ can be solved with $\widetilde{O}(n^{1-\frac{\alpha}{1+(k-1)\alpha}})$ colors.
\end{lemma}

\begin{proof}
    The structure of the algorithm is the same as that for locally $k$-colorable graphs, except for the terminating condition. Therefore, the same results as \autoref{lem:k5-base}, \autoref{lem:k5-sp}, and \autoref{lem:k5-loc} hold for the number of colors used. Note that if we change the $O(\cdot)$ notation to $\widetilde{O}(\cdot)$, the arguments still hold.
\end{proof}

\begin{theorem}
    \label{thm:k5-col}
    For all $k \geq 2$, $\textnormal{\textsf{COL}}(k)$ can be solved with $\widetilde{O}(n^{1-\frac{1}{k(k-1)/2}})$ colors.
\end{theorem}

\begin{proof}
    First, we can solve $\textsf{COL}(2)$ with $2 \log_2 (n+1) + 1$ colors by \autoref{thm:bipartite} (one extra color is for the last vertex when the graph becomes non-bipartite), which is $\widetilde{O}(1)$ colors. Next, we consider the case $k \geq 3$. We prove by induction on $k$. \autoref{lem:k5-col} shows that $\textsf{COL}(k)$ can be solved with $\widetilde{O}(n^{1-\frac{\alpha}{1+(k-1)\alpha}})$ colors, where we can use this lemma with $\alpha = \frac{1}{(k-1)(k-2)/2}$, by the induction hypothesis. Therefore, $\textsf{COL}(k)$ can be solved with $\widetilde{O}(n^{1-\frac{1}{k(k-1)/2}})$ colors.
\end{proof}

\subsection{The Competitive Ratio}

We presented a deterministic online algorithm to color locally $k$-colorable graphs with $O(n^{1-\frac{1}{k(k-1)/2+1}})$ colors (\autoref{thm:k5-loc}). This algorithm can also be used to color $k$-colorable graphs. To analyze the competitive ratio, we need to examine the constant factor in the number of colors used.

\begin{lemma}
    \label{lem:k5-const}
    For $k \geq 0$, $\textnormal{\textsf{LOC}}(k)$ can be solved with $2^{2^{k+3}} n^{1-\frac{1}{k(k-1)/2+1}}$ colors.
\end{lemma}

\begin{proof}
    Given that $\textsf{LOC}(0)$ and $\textsf{LOC}(1)$ can be solved with $1$ and $2$ colors, respectively (see \autoref{thm:k5-loc}), the lemma holds for $k = 0, 1$. Hence, we consider the case $k \geq 2$. We assume that $\textsf{LOC}(0), \dots, \textsf{LOC}(k-1)$ can all be solved with $c n^{1-\alpha}$ colors ($c \geq 1, \alpha = \frac{1}{(k-1)(k-2)/2+1}$).

    We prove that, for $\ell = 1, \dots, k$, $\textsf{SP}(k, \ell, \mu)$ can be solved with $(2^{2^{k+2} - 2^{k-\ell+2}}) c |T|^{1-\alpha} \mu^{(\ell-1)\alpha}$ colors, by induction on $\ell$. For the $\ell = 1$ case, $\textsf{SP}(k, 1, \mu)$ can be solved with $c |T|^{1-\alpha}$ colors (see \autoref{lem:k5-base}). For the $\ell \geq 2$ case, we follow the proof of \autoref{lem:k5-sp}. The number of instances of $\textsf{SP}(k, \ell-1, 2\mu)$ is at most $2\mu \cdot 2^{2^{k-\ell}}$, so the number of colors used in this part is at most
    \begin{align*}
        (2^{2^{k+2} - 2^{k-\ell+3}}) c |T|^{1-\alpha} (2\mu)^{(\ell-2)\alpha} \cdot (2\mu \cdot 2^{2^{k-\ell}})^\alpha & = (2^{2^{k+2} - 2^{k-\ell+3}}) 2^{2^{k-\ell} \alpha} 2^{(\ell-1)\alpha} c |T|^{1-\alpha} \mu^{(\ell-1)\alpha} \\
        & \leq (2^{2^{k+2} - 2^{k-\ell+3}}) 2^{2^{k-\ell}} \cdot 2 \cdot c |T|^{1-\alpha} \mu^{(\ell-1)\alpha}
    \end{align*}
    where the inequality follows from $2^{(\ell-1)\alpha} \leq 2$, by $(k-1)\alpha = \frac{k-1}{(k-1)(k-2)/2+1} \leq 1$. Also, the number of instances of $\textsf{LOC}(k-\ell)$ is at most $2\mu+1$, so the number of colors used in this part is at most
    \begin{equation*}
        c |T|^{1-\alpha} (2\mu+1)^\alpha \leq c |T|^{1-\alpha} (3\mu)^{\alpha} \leq 3 \cdot c |T|^{1-\alpha} \mu^{(\ell-1)\alpha}
    \end{equation*}
    where the inequalities follow from $\mu \geq 1, \alpha \leq 1, \ell \geq 2$. In total, we use at most $\{(2^{2^{k+2} - 2^{k-\ell+3}}) 2^{2^{k-\ell}} \cdot 2 + 3\} c |T|^{1-\alpha} \mu^{(\ell-1)\alpha} \leq 2^{2^{k+2} - 2^{k-\ell+2}} c |T|^{1-\alpha} \mu^{(\ell-1)\alpha}$, which matches the desired bound. Therefore, it needs  $2^{2^{k+2}-4} c |T|^{1-\alpha} \mu^{(k-1)\alpha}$ colors to solve $\textsf{ST}(k)$. To solve $\textsf{LOC}(k)$, by the proofs of \autoref{lem:k5-loc} and \autoref{lem:k5-n-is-known}, we set $\mu = n^{\frac{1}{k(k-1)/2+1}}$ and use at most $4 \cdot \{(2^{2^{k+2}-4} c + 1) n^{1-\frac{1}{k(k-1)/2+1}} + 1\}$ colors, which is not larger than $2^{2^{k+2}} c n^{1-\frac{1}{k(k-1)/2+1}}$.

    Finally, we prove that $\textsf{LOC}(k)$ can be solved with $2^{2^{k+3}} n^{1-\frac{1}{k(k-1)/2+1}}$ colors, by induction on $k$. By the induction hypothesis, we can use the argument above with $c = 2^{2^{k+2}}$. Then, we can derive that $\textsf{LOC}(k)$ can be solved with $2^{2^{k+2}} c n^{1-\frac{1}{k(k-1)/2+1}} = 2^{2^{k+3}} n^{1-\frac{1}{k(k-1)/2+1}}$ colors.
\end{proof}

Finally, we prove the main result of this subsection. The previous best-known competitive ratio was $O(\frac{n \log \log \log n}{\log \log n})$ by Kierstead \cite{Kie98}, which follows from the deterministic online algorithm to color $k$-colorable graphs with $n^{1-\frac{1}{k!}}$ colors. We improve this ratio by a factor of $\log \log \log n$.

\begin{theorem}
    \label{thm:k5-competitive}
    There exists a deterministic online coloring algorithm with competitive ratio $O(\frac{n}{\log \log n})$.
\end{theorem}

\begin{proof}
    First, we prove that, for all $k \leq \frac{1}{2} \log \log n$, the following holds for sufficiently large $n$:
    \begin{equation*}
        2^{2^{k+3}} \cdot n^{1-\frac{1}{k(k-1)/2+1}} < \frac{n}{\log \log n}
    \end{equation*}
    Taking logarithms of both sides, we obtain:
    \begin{align*}
        & 2^{k+3} + \left(1-\frac{1}{k(k-1)/2+1}\right)\log n < \log n - \log \log \log n \\
        (\Leftrightarrow) \ & 2^{k+3} + \log \log \log n < \frac{\log n}{k(k-1)/2+1}
    \end{align*}
    For $k \leq \frac{1}{2} \log \log n$, the left-hand side is $O(\sqrt{\log n})$ by $2^{k+2} = O(\sqrt{\log n})$, but the right-hand side is $\Omega(\frac{\log n}{(\log \log n)^2})$. Therefore, the inequality holds for sufficiently large $n$. With \autoref{lem:k5-const}, $\textsf{LOC}(k)$, for $k \leq \frac{1}{2} \log \log n$, only uses $\frac{n}{\log \log n}$ colors, when $n$ is sufficiently large.
    
    Now, we describe an algorithm that works even when $k$ is unknown. We run $\textsf{LOC}(1)$, $\textsf{LOC}(2)$, $\textsf{LOC}(3)$, $\dots$, in this order, each using completely disjoint colors. For each $i$, we start coloring vertices with $\textsf{LOC}(i)$ only after $\textsf{LOC}(i-1)$ terminates (this only happens when the graph becomes non-locally-$(i-1)$-colorable, which is non-$(i-1)$-colorable). If $\chi(G) < \frac{1}{2} \log \log n$, the number of colors used is at most $\chi(G) \cdot \frac{n}{\log \log n}$ by the discussion above. Otherwise, we obviously use at most $n$ colors. Hence, this algorithm achieves a competitive ratio of $O(\frac{n}{\log \log n})$.

    $\textsf{LOC}(k)$ runs in polynomial time, even when $k$ is not constant. To prove this, it suffices to show that the total number of created instances of any problem is bounded by a polynomial. For the algorithms for all problems, each vertex triggers the creation of at most one instance, provided that $q_{i,j}$ triggers the creation of $I_{i,j}$ in \autoref{alg:k5-sp}. Also, the recursion depth when solving $\textsf{LOC}(k)$ is $O(k^2)$. Hence, there are at most $O(k^2)$ instances per each vertex in $\textsf{LOC}(k)$, which is $O(nk^2)$ instances in total.
\end{proof}

\subsection{Additional Remarks}
\label{subsec:k5-remarks}

\paragraph{Comparison with Kierstead's algorithm.}
Kierstead discovered an algorithm for $k$-colorable graphs that uses $\widetilde{O}(n^{2/3})$ colors for $k = 3$ and $\widetilde{O}(n^{5/6})$ colors $k = 4$, using the framework similar to what we mentioned in \autoref{subsec:k5-background} and \autoref{subsec:k5-improvement}. However, for $k \geq 5$, he did not find a decent algorithm for $\textsf{LOC}(k-\ell)$ (when $k-\ell \geq 3$), so he instead used $\textsf{COL}(k-\ell)$ for coloring a subset of vertices in $T$ and used the entire subset as a ``witness'' for initiating sub-instances. This created a lot of instances, so he only achieved $O(n^{1-\frac{1}{k!}})$ colors. Our fundamental contribution is to give an effective algorithm for $\textsf{LOC}(k)$, which is achieved by introducing a new concept of level-$\ell$ set, the ``locally'' version of $\ell$-color set.

We also applied Kierstead's framework to understand the problem better. Kierstead used a structure called ``witness tree'' to maintain sub-instances, but we take an inductive approach instead. We also introduced the S-T Problem and presented a reduction in \autoref{subsec:k5-st-reduction} to clarify the essential part that must be processed online. Note that Kierstead proved that, for ``every'' $k$-coloring of $G$ (in the original problem $\textsf{COL}(k)$), there exists an instance $I$ of $\textsf{SP}(k, \ell)$ that at most $\ell$ distinct colors appear on $S(I)$ (in fact, this also holds in our algorithm). Instead, we defined an $\ell$-color set to have at most $\ell$ distinct colors in ``some'' $k$-coloring, which is a weaker definition, but this change enabled a ``bottom-up'' inductive approach.

\paragraph{Thresholds for locally $k$-colorable graphs.}
We defined locally $k$-colorable graphs so that every subgraph with order less than $2^{2^k}$ is $k$-colorable. We note that our algorithm for $\textsf{LOC}(k)$ actually returns a smaller non-$k$-colorable subgraph. Let $f(k)$ and $g(k, \ell)$ be the maximum size of the output set in $\textsf{LOC}(k)$ and $\textsf{SP}(k, \ell)$, respectively. By \autoref{lem:k5-base} and \autoref{lem:k5-sp-output}, we have:
\begin{equation*}
    g(k, 1) = 2f(k-1), \quad g(k, \ell) = (g(k, \ell-1) + 1) \cdot  f(k-\ell) \ (\ell \in \{2, \dots, k\}), \quad f(k) = g(k, k)
\end{equation*}
The values for $k \leq 5$ are mentioned in \autoref{tab:k5-threshold}. Note that $f(2) = 5$ matches Kierstead's result that colors $C_3, C_5$-free graphs online with $O(n^{1/2})$ colors (\autoref{thm:kierstead-2good}, \cite{Kie98}). They are much less than $2^{2^k}$, but still grow double-exponentially on $k$, which is the reason that we set the threshold to $2^{2^k}$ for $\textsf{SP}(k)$ and to $2^{2^k - 2^{k-\ell}}$ for $\textsf{SP}(k, \ell)$.

\begin{table}[htbp]
    \centering
    \caption{The maximum size for the output set: $f(k)$ for $\textsf{LOC}(k)$, and $g(k, \ell)$ for $\textsf{SP}(k, \ell)$.}
    \begin{tabular}{|c|c|ccccc|}
        $k$ & $f(k)$ & $g(k, 1)$ & $g(k, 2)$ & $g(k, 3)$ & $g(k, 4)$ & $g(k, 5)$ \\ \hline
        0 & 1 & - & - & - & - & - \\
        1 & 2 & - & - & - & - & - \\
        2 & 5 & 4 & 5 & - & - & - \\
        3 & 23 & 10 & 22 & 23 & - & - \\
        4 & 473 & 46 & 235 & 472 & 473 & - \\
        5 & 217823 & 946 & 21781 & 108910 & 217822 & 217823 \\
    \end{tabular}
    \label{tab:k5-threshold}
\end{table}

\paragraph{On the competitive ratio.}
The bottleneck of the competitive ratio $O(\frac{n}{\log \log n})$ is in the double-exponential constant factor in the number of colors used ($2^{2^{k+3}}$ in \autoref{lem:k5-const}), not in the exponent $1-\frac{1}{k(k-1)/2}$. If we just improve the factor to $2^{\mathrm{poly}(k)}$, the competitive ratio would improve to $O(\frac{n}{(\log n)^c})$ for some $c > 0$, which would close the gap to the $\Omega(\frac{n}{(\log n)^2})$ lower bound \cite{HS94}. We also note that the reason for this factor is that the output size of $\textsf{LOC}(k)$ is double-exponential.

\section{Online Coloring for $4$-Colorable Graphs}\label{sec:k4}

In this section, we present a deterministic online algorithm for coloring a 4-colorable graph $G$ with $n$ vertices using $\widetilde{O}(n^{14/17})$ colors, improving the previous bound of $\widetilde{O}(n^{5/6})$ colors \cite{Kie98}.

\subsection{The Overall Idea}\label{sec:k4-intro}

First, we apply First-Fit with $n^{14/17}$ colors. Then, the problem of coloring the remaining vertices reduces to solving $\textsf{ST}^+(4)$ (where $|S| = n, |T| \leq n$) with $\mu = n^{3/17}$, by the discussion in \autoref{subsec:k5-st-reduction} and \autoref{subsec:k5-improvement}. Therefore, all we need is to solve $\textsf{ST}^+(4, n^{3/17})$ using $\widetilde{O}(n^{14/17})$ colors. The current result in \autoref{lem:k5-col} only achieves $\widetilde{O}(|T|^{2/3} \mu) = \widetilde{O}(n^{43/51})$ colors, so we need to further improve it. As this problem is challenging, we first consider the special case satisfying the following property, in which there is no dense subgraph in $G$.

\begin{description}
    \item[\textbf{No-dense property.}] There are no subsets $S_D \subseteq S$ and $T_D \subseteq T$ that satisfy the following conditions: (i) $|T_D| \leq n^{8/17}$, (ii) $|S_D| \leq n^{8/17}|T_D|$, and (iii) $|N_{S_D}(t)| \geq \frac{1}{3} n^{14/17}$ for all $t \in T_D$.
\end{description}

The steps for solving this special case are shown in \autoref{fig-506}. First, we obtain a 1-color set $T_D \subseteq T$ such that $|T_D| = \Omega(n^{7/17})$ (\autoref{fig-506}A). To obtain a large 1-color set, we apply First-Fit on $T$, and use the same approach as \autoref{sec:k5} that utilizes the large degrees. Since First-Fit is run in two stages, we call this technique the \emph{double greedy method}. Then, assuming the no-dense property, we obtain a large 3-color set in $S$ (\autoref{fig-506}B), since $N_S(T_D)$ is a 3-color set of $\Omega(n^{15/17})$ vertices. By repeating this process, we can cover $S$ with $O(n^{2/17})$ 3-color sets (\autoref{fig-506}C). Afterward, we color the remaining vertices using $\textsf{SP}^+(4, 3)$ based on these 3-color sets (\autoref{fig-506}D). The advantage is that we only need to create $O(n^{2/17})$ instances of $\textsf{SP}^+(4, 3)$ (instead of $\mu = n^{3/17}$ instances), resulting in an improved number of colors.

\begin{figure}[t]
  \centering
  \includegraphics[width=\linewidth]{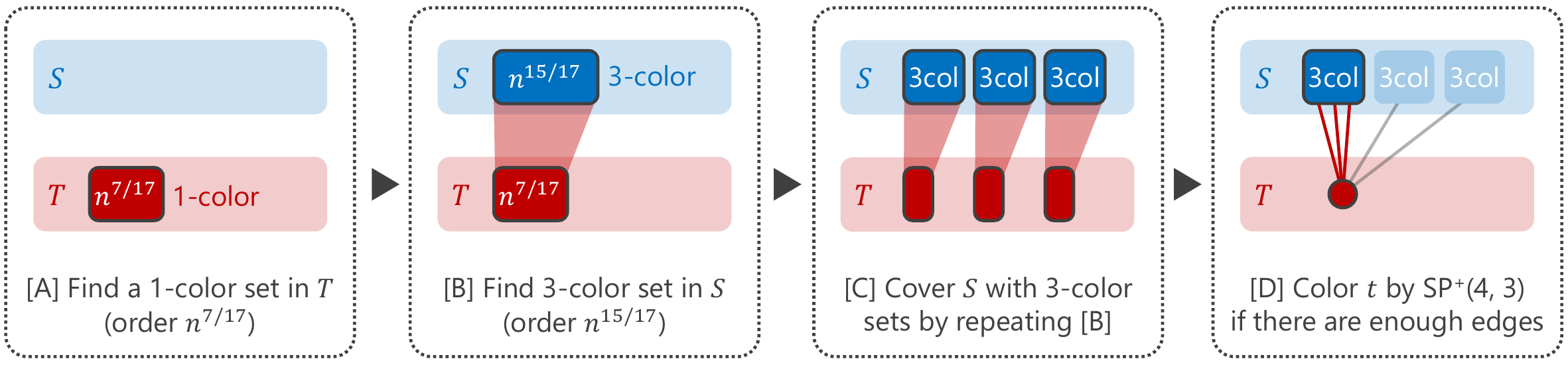}
  \caption{The sketch of our algorithm for the no-dense case.} \label{fig-506}
\end{figure}

However, the no-dense property does not always hold. In the worst case, $|S_D| = n^{14/17}$ even though $|T_D| = n^{8/17}$. To address this, we use our novel \emph{Common \& Simplify technique} to reduce $\textsf{ST}^+(4)$ to the case with the no-dense property. The core idea is to pick two highly common vertices $u_1, u_2 \in T$ such that $|N_S(u_1) \cap N_S(u_2)| = \widetilde{\Omega}(n^{12/17})$ (we can show that such a pair exists if the no-dense property does not hold; see \autoref{lem:k4-ab-free}). If there exists a 4-coloring of $G$ in which different colors are used for $u_1$ and $u_2$, $N_S(u_1) \cap N_S(u_2)$ is a 2-color set of $\widetilde{\Omega}(n^{12/17})$ vertices (\autoref{fig-520}), where we can use them for $\textsf{SP}^+(4, 2)$. Note that sometimes $u_1$ and $u_2$ may be assigned the same color, but we can force $(u_1, u_2)$ to be different colors with a ratio of at least $\frac{1}{4}$ (see \autoref{alg:division}).

In \autoref{subsec:k4-cand}, we introduce the ``monochromatic candidates problem'' which is used as a tool in the proof. In \autoref{subsec:k4-nodense} and \autoref{subsec:k4-nodense-proof}, we present an algorithm for the no-dense property case. In \autoref{subsec:k4-dense}, we introduce the Common \& Simplify technique and finally present a deterministic online algorithm that colors any 4-colorable graph with $\widetilde{O}(n^{14/17})$ colors.

\begin{figure}[t]
  \centering
  \includegraphics[width=\linewidth]{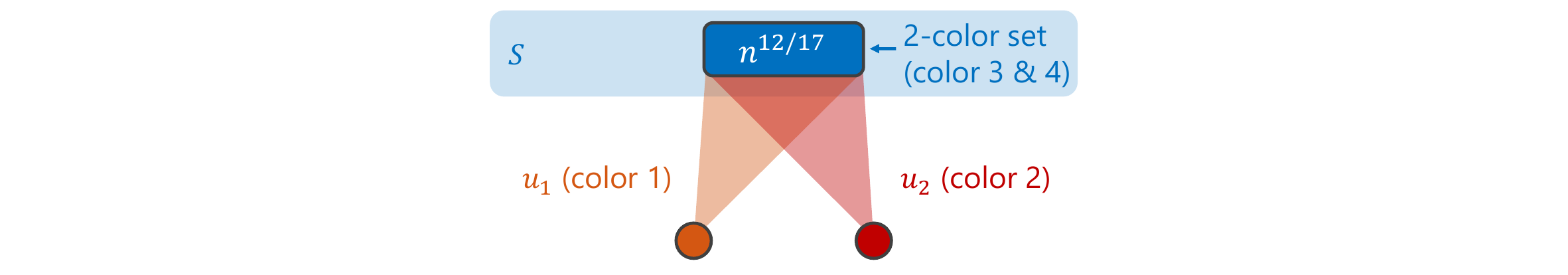}
  \caption{The core idea of the Common \& Simplify technique. If $u_1$ and $u_2$ are colored by color 1 and 2 in some 4-coloring of $G$, then $N_S(u_1) \cap N_S(u_2)$ is forced to be a 2-color set.} \label{fig-520}
\end{figure}

\subsection{Monochromatic Candidates Problem} \label{subsec:k4-cand}

Before solving the no-dense case of $\textsf{ST}^+(4)$, we define the following ``monochromatic candidates problem'', which is a key tool for finding candidate 1-color sets using a small number of colors. We also define the corresponding subproblem. Note that $\textsf{CAND}(k)$ is equivalent to $\textsf{CANDSP}(k, k)$.

\begin{itemize}
    \item Problem $\textsf{CAND}(k)$: The S-T problem where we can terminate the problem by outputting a set $\mathcal{A} = \{A_1, \dots, A_m\} \ (A_i \subseteq S)$ such that: (i) $m = O(1)$, (ii) $|A_i| \geq \frac{|S|}{2^{(k-1)(k-2)/2} \mu^{k-1}}$ for each $i$, and (iii) for \emph{every} $k$-coloring $\varphi: V \to \{1, \dots, k\}$ of $G$, at least one of $A_1, \dots, A_m$ is monochromatic, i.e., there exists an $i$ such that $|\{\varphi(v) : v \in A_i\}| = 1$.\footnote{Note that unlike the definition of ``1-color set'', we define ``monochromatic'' with respect to a fixed $k$-coloring $\varphi$.} We also denote this problem by $\textsf{CAND}(k, \mu)$.
    \item Problem $\textsf{CANDSP}(k, \ell)$: The S-T problem where we can terminate the problem by outputting a set $\mathcal{A} = \{A_1, \dots, A_m\} \ (A_i \subseteq S)$ such that: (i) $m = O(1)$, (ii) $|A_i| \geq \frac{|S|}{2^{(\ell-1)(\ell-2)/2} \mu^{\ell-1}}$ for each $i$, and (iii) for \emph{every} $k$-coloring $\varphi: V \to \{1, \dots, k\}$ of $G$ such that at most $\ell$ distinct colors appear in $S$, at least one of $A_1, \dots, A_m$ is monochromatic. We also denote this problem by $\textsf{CANDSP}(k, \ell, \mu)$.
\end{itemize}

\paragraph{Algorithm.}
First, in $\textsf{CANDSP}(k, 1)$, we can output $\mathcal{A} := \{S\}$ and terminate the problem immediately after its initiation. Next, we can solve $\textsf{CANDSP}(k, \ell)$ (for $\ell \geq 2$) using the same algorithm as \autoref{alg:k5-sp}, which creates instances of $\textsf{CANDSP}(k, \ell-1, 2\mu)$ (instead of $\textsf{SP}(k, \ell-1, 2\mu)$). Note that we color vertices in $T$ with $\textsf{LOC}(k-\ell)$ also in this case. When instances $I_{i,1}, \dots, I_{i,|Q_i|}$ are all terminated, we output $\mathcal{A} := \mathcal{A}_1 \cup \dots \cup \mathcal{A}_{|Q_i|}$, where $\mathcal{A}_j$ is the output of $I_{i,j}$. Especially for the $\ell = 2$ case, when $I^{\textsf{LOC}}_1$ finishes and the instances $I_{1,1}, \dots, I_{1,|Q_1|}$ are created, these instances all terminate immediately, resulting in the termination of $\textsf{CANDSP}(k, 2)$. Therefore, in $\textsf{CANDSP}(k, 2)$, we only create one instance of $\textsf{LOC}(k-2)$.

\begin{lemma}
    In the algorithm for $\textnormal{\textsf{CANDSP}}(k, \ell)$ with $\ell \geq 2$, if instances $I_{i,1}, \dots, I_{i,|Q_i|}$ are all terminated, then the output $\mathcal{A} := \mathcal{A}_1 \cup \dots \cup \mathcal{A}_{|Q_i|}$ meets the requirements.
\end{lemma}

\begin{proof}
    For condition (i), since $|\mathcal{A}_j| = O(1)$ by the assumption of $\textsf{CANDSP}(k, \ell-1)$, we have $|\mathcal{A}| \leq |\mathcal{A}_1| + \dots + |\mathcal{A}_{|Q_i|}| \leq |Q_i| \cdot O(1) = O(1)$. For condition (ii), each element in $\mathcal{A}_j$ has size at least $\frac{|S(I_{i,j})|}{2^{(\ell-2)(\ell-3)/2} (2\mu)^{\ell-2}} \geq \frac{|S|}{2^{(\ell-1)(\ell-2)/2} \mu^{\ell-1}}$ by the assumption of $\textsf{CANDSP}(k, \ell-1, 2\mu)$ and $|S(I_{i,j})| \geq \frac{|S|}{\mu}$.
    
    Next, we consider condition (iii). Let $\varphi: V \to \{1, \dots, k\}$ be a coloring of $G$ such that at most $\ell$ distinct colors appear in $S$. We assume, without loss of generality, that $\varphi(s) \in \{1, \dots, \ell\}$ for all $s \in S$. Since $\chi(G[Q_i]) > k-\ell$, there exists $q_{i,j} \in Q_i$ such that $\varphi(q_{i,j}) \in \{1, \dots, \ell\}$. For such $q_{i,j}$, only $\ell-1$ distinct colors (i.e., colors in $\{1, \dots, \ell\} \setminus \{\varphi(q_{i,j})\}$) appear in $S(I_{i,j})$, due to $S(I_{i,j}) = N_S(q_{i,j})$. Hence, by the assumption of $\textsf{CANDSP}(k, \ell-1)$, at least one set in $\mathcal{A}_j$ is monochromatic. Therefore, at least one set in $\mathcal{A} := \mathcal{A}_1 \cup \dots \cup \mathcal{A}_{|Q_i|}$ is monochromatic.
\end{proof}

Since $\textsf{CAND}(k)$ is equivalent to $\textsf{CANDSP}(k, k)$, this gives an algorithm to solve $\textsf{CAND}(k)$. Next, we analyze the number of colors used.

\begin{lemma}
    We can solve $\textnormal{\textsf{CANDSP}}(k, \ell)$ with $O(|T|^{1-\alpha} \mu^{(\ell-2)\alpha})$ colors, for $\ell = 2, \dots, k$, assuming that $\textnormal{\textsf{LOC}}(0), \dots, \textnormal{\textsf{LOC}}(k-2)$ can all be solved with $\widetilde{O}(n^{1-\alpha})$ colors.
\end{lemma}

\begin{proof}
    We prove it by induction on $\ell$. For $\ell = 2$, we color vertices in $T$ using one instance of $\textsf{LOC}(k-2)$, so we use $O(|T|^{1-\alpha})$ colors. For $\ell \geq 3$, we color vertices in $T$ using at most $2\mu+1$ instances of $\textsf{LOC}(k-\ell)$ and at most $2\mu$ instances of $\textsf{CANDSP}(k, \ell-1)$. Therefore, by a similar argument to the proof of \autoref{lem:k5-sp}, we use at most $O(|T|^{1-\alpha} \mu^{(\ell-2)\alpha})$ colors.
\end{proof}

Therefore, $\textsf{CAND}(4)$ can be solved with $O(|T|^{1/2} \mu)$ colors, which uses fewer colors than $\text{ST}^+(4)$ (with $\widetilde{O}(|T|^{2/3} \mu)$ colors). Each set in the output has size at least $\frac{|S|}{8 \mu^3}$. Also, the size of the output is $m \leq 5 \cdot 2 \cdot 1 = 10$, since the maximum sizes of the output sets for $\textsf{LOC}(2), \textsf{LOC}(1), \textsf{LOC}(0)$ are $5, 2, 1$, respectively; see \autoref{tab:k5-threshold}.

\subsection{No-Dense Case: Algorithm} \label{subsec:k4-nodense}

Next, we show an algorithm to solve the no-dense case of $\textsf{ST}^+(4)$ where $|S| = n, |T| \leq n, \mu = n^{3/17}$.

\paragraph{Overview.}
The strategy is to obtain monochromatic sets in $T$ using $\textsf{CAND}(4)$. We apply First-Fit using $n^{13/17}$ colors (\autoref{fig-518} left; let $T_{\mathrm{FF}}$ be the set of vertices colored by First-Fit), and we use $\textsf{CAND}(4)$ to color the remaining vertices, exploiting their large degrees. Although vertices in $T \setminus T_{\mathrm{FF}}$ may arrive earlier than those of $T_{\mathrm{FF}}$, we can reduce this to the S-T problem setting by the arguments in \autoref{subsec:k5-st-reduction}.

If the instance of $\textsf{CAND}(4)$ returns ``candidate monochromatic sets'' $\mathcal{A} = \{A_1, \dots, A_m\}$, then for each $i$, consider initiating a new instance $I$ of $\textsf{SP}^+(4, 3)$ where $S(I) := N_S(A_i)$. This relies on the assumption that at least one of $A_1, \dots, A_m$ is monochromatic in every $4$-coloring of a subgraph of $G$; then, as long as $G$ is $4$-colorable, at least one of $A_1, \dots, A_m$ is a $1$-color set in some $4$-coloring of $G$, and hence at least one of $N_S(A_1), \dots, N_S(A_m)$ is a $3$-color set in some $4$-coloring of $G$ (\autoref{fig-518} middle). These instances are used to process future vertices: when a new vertex $t \in T$ arrives and $|N_{S(I)}(t)|$ is large enough, we can color $t$ using instance $I$ (\autoref{fig-518} right). We also initiate a new instance of $\textsf{CAND}(4)$ for coloring the rest of the vertices in $T$, using completely different colors, and on some occasions, we also refresh First-Fit and start using the next $n^{13/17}$ colors.

\begin{figure}[t]
  \centering
  \includegraphics[width=0.9\linewidth]{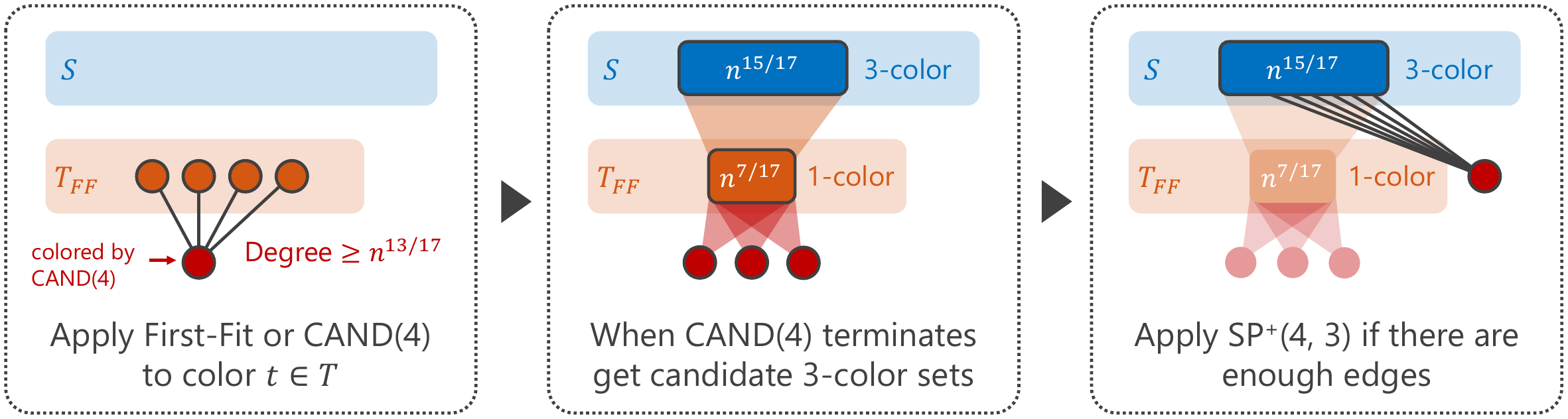}
  \caption{The sketch of our algorithm for the no-dense case.} \label{fig-518}
\end{figure}

\paragraph{Parameters.} Next, we describe the parameter settings for $\textsf{CAND}(4)$ and $\textsf{SP}^+(4, 3)$:
\begin{itemize}
    \item When we create an instance $J$ of $\textsf{CAND}(4)$, we always set $|S(J)| = n^{16/17}$ and $\mu(J) = 2n^{3/17}$. $|N_{S(J)}(t)| \geq \frac{1}{2} n^{13/17}$ holds for each $t \in T(J)$. The output $\mathcal{A} = \{A_1, \dots, A_m\}$ satisfies $|A_i| \geq \frac{|S(J)|}{8\mu(J)^3} = \frac{1}{64} n^{7/17}$. This uses $O(|T|^{1/2} n^{3/17})$ colors.
    \item When we create an instance $I$ of $\textsf{SP}^+(4, 3)$, we always set $|S(I)| = n^{15/17}$ and $\mu(I) = 6400 n^{3/17}$. $|N_{S(I)}(t)| \geq \frac{1}{6400} n^{12/17}$ holds for each $t \in T(I)$. This uses $\widetilde{O}(|T|^{2/3} n^{2/17})$ colors by \autoref{lem:k5-col}.
\end{itemize}

The parameter settings for $\textsf{CAND}(4)$ imply that we can continue coloring while $|T_{\mathrm{FF}}| < n^{16/17}$. When $|T_{\mathrm{FF}}|$ reaches $n^{16/17}$, we refresh First-Fit and start using completely new colors.

\paragraph{Implementation.}

We denote the created instances of $\textsf{CAND}(4)$ by $J_1, J_2, \dots$ in chronological order. When the $r$-th instance $J_r$ terminates, we let $\mathcal{A}_r := \{A_{r,1}, \dots, A_{r,|\mathcal{A}_r|}\}$ for its output, and $I_{r, 1}, \dots, I_{r,|\mathcal{A}_r|}$ for the instances of $\textsf{SP}^+(4, 3)$ created at that time. We use completely disjoint colors for all of these instances as well as First-Fit.

When initiating an instance $J_r$ of $\textsf{CAND}(4)$, we set $S(J_r)$ to have size of $n^{16/17}$, where vertices of $S(J_r)$ correspond to vertices in $T_{\mathrm{FF}} \setminus T_{\mathrm{NG}}$. We establish a correspondence in the same manner as in \autoref{subsec:k5-st-reduction}. Here, $T_{\mathrm{FF}}$ is the set of vertices colored by the current session of First-Fit, and $T_{\mathrm{NG}}$ is defined as follows, where $S_{\mathrm{used}} := \bigcup \{S(I_{i,j}): \text{$I_{i,j}$ is active}\}$:
\begin{equation}
    \label{eq:k4-t-ng}
    T_{\mathrm{NG}} := \left\{t \in T_{\mathrm{FF}} : |N_{S \setminus S_{\mathrm{used}}} (t)| < \frac{1}{2} n^{14/17}\right\}
\end{equation}
Note that some vertices in $S(J_r)$ may not be mapped to any vertices in $T_{\mathrm{FF}} \setminus T_{\mathrm{NG}}$, but they are isolated vertices, so they will not be included in the output $\mathcal{A}_r$ (due to the algorithm in \autoref{subsec:k4-cand}).

The algorithm starts by initiating the instance $J_1$. Whenever a vertex $t \in T$ arrives, we do the following procedure:

\begin{enumerate}
    \item If there exists an active instance $I_{i,j}$ such that $|N_{S(I_{i,j})}(t)| \geq \frac{1}{6400} n^{12/17}$, we choose one such instance $I_{i,j}$ arbitrarily, and color $t$ using this instance $I_{i,j}$ (line 6). If $I_{i,1}, \dots, I_{i,|\mathcal{A}_i|}$ are all terminated as a result, then we declare that $G$ is not $4$-colorable and terminate the original problem $\textsf{ST}^+(4)$ (line 8).
    \item Otherwise, if $|N_{T_{\mathrm{FF}}}(t)| < n^{13/17}$, we color $t$ with First-Fit (lines 9--11).
    \item Otherwise, we color $t$ using instance $J_r$, where $J_r$ is the currently running instance of $\textsf{CAND}(4)$ (line 13). If $J_r$ terminates as a result, let $\mathcal{A}_r = \{A_{r,1}, \dots, A_{r,|\mathcal{A}_r|}\}$ be its output. Then, for each $j = 1, \dots, |\mathcal{A}_r|$, we initiate an instance $I_{r,j}$ of $\textsf{SP}^+(4, 3)$ so that $S(I_{r,j}) := [N_{S \setminus S_{\mathrm{used}}}(A_{r,j})]$, where $[X]$ is defined to be the first $\frac{1}{64} n^{15/17}$ elements of $X \subseteq S$ (line 15).

    For each $t \in A_{r,j}$, $|N_{S \setminus S_{\mathrm{used}}}(t)| \geq \frac{1}{2} n^{14/17}$ holds due to $A_{r, j} \subseteq T_{\mathrm{FF}} \setminus T_{\mathrm{NG}}$. Hence, given that $|A_{r,j}| \geq \frac{1}{64} n^{7/17}$, we can apply the no-dense property with $T_D := (\text{the first $\frac{1}{64} n^{7/17}$ vertices of $A_{r,j}$})$ and $S_D := N_{S \setminus S_{\mathrm{used}}}(A_{r,j})$, and we obtain $|N_{S \setminus S_{\mathrm{used}}}(A_{r,j})| \geq \frac{1}{64} n^{15/17}$.
    \item If $|T_{\mathrm{FF}}|$ reaches $n^{16/17}$, or at least $n^{1/17}$ instances of $\textsf{CAND}(4)$ have been terminated after the last refresh of First-Fit, we refresh First-Fit; $T_{\mathrm{FF}}$ is reset to empty, and we use new $n^{13/17}$ colors in the future First-Fit (lines 17--18).
    \item If $J_r$ terminates in step 3 or First-Fit is refreshed in step 4, we initiate a new instance $J_{r+1}$ of $\textsf{CAND}(4)$, which increments $r$ by $1$ (lines 19--20). This means that, in case we refresh First-Fit by $|T_{\mathrm{FF}}|$ reaching $n^{16/17}$, the instance $J_r$ becomes ``inactive'' without being terminated.
\end{enumerate}

Note that, during the run of an instance of $\textsf{CAND}(4)$, no element in $T_{\mathrm{FF}} \setminus T_{\mathrm{NG}}$ gets deleted, as $S_{\mathrm{used}}$ changes only when $\textsf{CAND}(4)$ terminates. The pseudocode for this algorithm is given in \autoref{alg:no-dense-case}.

\begin{algorithm}[htbp]
\caption{The solution for $\textsf{ST}^+(4)$ with $|S| = n, |T| \leq n, \mu = n^{3/17}$ when the no-dense property is met. For the parameters used in $\textsf{CAND}(4)$ and $\textsf{SP}^+(4, 3)$, see \textbf{Parameters} paragraph. We set the $S$ for $\textsf{CAND}(4)$ to be $T_{\mathrm{FF}} \setminus T_{\mathrm{NG}}$ (via the reduction in \autoref{subsec:k5-st-reduction}), where $T_{\mathrm{NG}}$ is defined in \autoref{eq:k4-t-ng}.}
\begin{algorithmic}[1]
    \State $J_1 \gets$ a new instance of $\textsf{CAND}(4)$, and $r \gets 1, c \gets 0$
    \State $T_{\mathrm{FF}} \gets \emptyset$
    \For {each arrival of $t \in T$}
        \State $S_{\mathrm{used}} \gets \bigcup \{S(I_{i,j}): \text{$I_{i,j}$ is active}\}$
        \If {there exists an active instance $I_{i,j}$ such that $|N_{S(I_{i,j})}(t)| \geq \frac{1}{6400} n^{12/17}$}
            \State color $t$ using $I_{i,j}$
            \If {$I_{i,1}, \dots, I_{i,|\mathcal{A}_i|}$ are all terminated}
                \State \Return that $G$ is not $4$-colorable and terminate $\textsf{ST}^+(4)$
            \EndIf
        \ElsIf {$|N_{T_{\mathrm{FF}}}(t)| < n^{13/17}$}
            \State color $t$ with First-Fit
            \State $T_{\mathrm{FF}} \gets T_{\mathrm{FF}} \cup \{t\}$ \Comment{$T_{\mathrm{NG}}$ is also updated}
        \Else
            \State color $t$ using $J_r$
            \If {$J_r$ terminates by outputting $\mathcal{A}_r = \{A_{r,1}, \dots, A_{r,|\mathcal{A}_r|}\}$}
                \State $I_{r,j} \gets$ a new instance of $\textsf{SP}^+(4, 3)$ where $S(I_{r,j}) := [N_{S \setminus S_{\mathrm{used}}}(A_{r,j})]$, for $j = 1, \dots, |\mathcal{A}_r|$
                \State $c \gets c + 1$ \Comment{For $X \subseteq S$, we define $[X]$ to be the first $\frac{1}{64} n^{15/17}$ vertices of $X$}
            \EndIf
        \EndIf
        \If {$|T_{\mathrm{FF}}| \geq n^{16/17}$ or $c \geq n^{1/17}$}
            \State refresh First-Fit, and $T_{\mathrm{FF}} \gets \emptyset, c \gets 0$
        \EndIf
        \If {$J_r$ is terminated or First-Fit is refreshed in this iteration}
            \State $J_{r+1} \gets$ a new instance of $\textsf{CAND}(4)$, and $r \gets r+1$
        \EndIf
    \EndFor
\end{algorithmic}
\label{alg:no-dense-case}
\end{algorithm}

\subsection{No-Dense Case: Correctness \& Analysis} \label{subsec:k4-nodense-proof}

Next, we prove that the algorithm works correctly. The missing piece is to show that, when we color $t$ using $J_r$ (line 13), it meets the degree requirement (\autoref{lem:k4-nodense-degree}). After that, we analyze the number of colors used, and show that \autoref{alg:no-dense-case} achieves the target of $\widetilde{O}(n^{14/17})$ colors (\autoref{lem:k4-nodense}).

\begin{lemma}
    \label{lem:k4-nodense-num}
    Let $X := \{i \in \{1, \dots, r\} : \text{$J_i$ outputted the set $\mathcal{A}_i$ and terminated}\}$. Then, $|X| \leq 64n^{2/17}$.
\end{lemma}

\begin{proof}
    For each $i \in X$, there exists an active instance $I_{i,t_i}$ for some $t_i$. Let $R_i := S(I_{i,t_i})$. Then, for each $i, j \in X \ (i < j)$, $R_i \cap R_j = \emptyset$, because when the instance $I_{j,t_j}$ is created, $R_i = S(I_{i,t_i})$ has already been contained in $S_{\mathrm{used}}$. Also, we have $|R_i| = \frac{1}{64} n^{15/17}$. Therefore, $|X| \leq 64 n^{2/17}$.
\end{proof}

\begin{lemma}
    \label{lem:k4-nodense-ng}
    $|T_{\mathrm{NG}}| < n^{8/17}$ holds under the no-dense property.
\end{lemma}

\begin{proof}
    Let $J_{r'}$ be the instance of $\textsf{CAND}(4)$ running right before the last refresh of First-Fit. Then, each $t \in T_{\mathrm{FF}}$ arrives after $J_{r'}$ becomes inactive. Let $S'_{\mathrm{used}} := \bigcup \{S(I_{i,j}) : \text{$i \leq r'$ and $I_{i,j}$ is active}\}$. Since (i) $|X| \leq 64 n^{2/17}$ (by \autoref{lem:k4-nodense-num}), (ii) $|\mathcal{A}_i| \leq 10$, and (iii) when $I_{i,j} \ (i \leq r')$ is active, $|N_{S(I_{i,j})}(t)| < \frac{1}{6400} n^{12/17}$ for each $t \in T_{\mathrm{FF}}$ (because otherwise it should be colored by $I_{i,j}$), we have for each $t \in T_{\mathrm{FF}}$:
    \begin{equation*}
        |N_{S'_{\mathrm{used}}}(t)| \leq 64 n^{2/17} \cdot 10 \cdot \frac{1}{6400} n^{12/17} = \frac{1}{10} n^{14/17}
    \end{equation*}
    Therefore, by \autoref{eq:k4-t-ng}, we have for each $t \in T_{\mathrm{NG}}$:
    \begin{equation*}
        |N_{S_{\mathrm{used}} \setminus S'_{\mathrm{used}}}(t)| = |N_S(t)| - |N_{S \setminus S_{\mathrm{used}}}(t)| - |N_{S'_{\mathrm{used}}}(t)| \geq n^{14/17} - \frac{1}{2} n^{14/17} - \frac{1}{10} n^{14/17} = \frac{2}{5} n^{14/17}
    \end{equation*}
    This fits the degree requirement for the no-dense property. Suppose that $|T_{\mathrm{NG}}| \geq n^{8/17}$, and let $T_D :=$ (the first $n^{8/17}$ vertices of $T_{\mathrm{NG}}$) and $S_D := S_{\mathrm{used}} \setminus S'_{\mathrm{used}}$. Since $c \leq n^{1/17}$ and $|S(I_{i,j})| = \frac{1}{64} n^{15/17}$:
    \begin{equation*}
        |S_D| \leq n^{1/17} \cdot 10 \cdot \frac{1}{64} n^{15/17} = \frac{5}{32} n^{16/17}
    \end{equation*}
    which contradicts the no-dense property between $T_D$ and $S_D$. Therefore, $|T_{\mathrm{NG}}| < n^{8/17}$.
\end{proof}

\begin{lemma}
    \label{lem:k4-nodense-degree}
    When we color $t$ using instance $J_r$, $|N_{S(J_r)}(t)| \geq \frac{1}{2} n^{13/17}$ holds.
\end{lemma}

\begin{proof}
    Since we did not color $t$ with First-Fit, $|N_{T_{\mathrm{FF}}}(t)| \geq n^{13/17}$ holds. Also, $|T_{\mathrm{NG}}| < n^{8/17}$ by \autoref{lem:k4-nodense-ng}. Therefore, $|N_{S(J_r)}(t)| = |N_{T_{\mathrm{FF}} \setminus T_{\mathrm{NG}}}(t)| \geq n^{13/17} - n^{8/17} \geq \frac{1}{2} n^{13/17}$.
\end{proof}

\begin{lemma}
    \label{lem:k4-nodense}
    \autoref{alg:no-dense-case} uses $\widetilde{O}(n^{14/17})$ colors.
\end{lemma}

\begin{proof}
    First, we analyze the number of colors used in First-Fit. We refresh First-Fit when $|T_{\mathrm{FF}}|$ becomes $n^{16/17}$ (which can happen at most $n^{1/17}$ times), or $n^{1/17}$ instances of $\textsf{CAND}(4)$ have terminated (which can happen at most $64 n^{1/17}$ times due to \autoref{lem:k4-nodense-num}), so the refresh happens at most $65 n^{1/17}$ times. Since each First-Fit uses $n^{13/17}$ colors, we use $O(n^{14/17})$ colors in total here.

    Next, we analyze for $\textsf{CAND}(4)$. By \autoref{lem:k4-nodense-num}, and additionally $|T_{\mathrm{FF}}|$ reaches $n^{16/17}$ for only $n^{1/17}$ times, we have $r = O(n^{2/17})$. Also, each instance $J_i$ of $\textsf{CAND}(4)$ uses $O(|T(J_i)|^{1/2} n^{3/17})$ colors. Thus, we use $O(n^{25/34})$ colors in total here, by \autoref{lem:prelim-ineq}.

    Finally, we analyze for $\textsf{SP}^+(4, 3)$. Since $|X| \leq 64 n^{2/17}$ (\autoref{lem:k4-nodense-num}) and $|\mathcal{A}_i| \leq 10$, we create at most $640n^{2/17}$ instances of $\textsf{SP}^+(4, 3)$. Also, each instance $I_{i,j}$ of $\textsf{SP}^+(4, 3)$ uses $\widetilde{O}(|T(I_{i,j})|^{2/3} n^{2/17})$ colors. Thus, we use $\widetilde{O}(n^{14/17})$ colors in total here, by \autoref{lem:prelim-ineq}.
\end{proof}

\subsection{Avoiding the Dense Case}\label{subsec:k4-dense}

The last step is to reduce the original coloring problem to the case with the no-dense property. Here, we consider applying First-Fit with $2n^{14/17}$ colors (instead of $n^{14/17}$) to solve online coloring of $4$-colorable graphs. Thus, we need to reduce $\textsf{ST}^+(4, \frac{1}{2} n^{3/17})$ to ``the problem $\textsf{ST}^+(4, n^{3/17})$ with the no-dense property''. We realize this based on the following branch framework (\autoref{alg:branch}).

\begin{algorithm}[htbp]
\caption{The branch framework for $\textsf{ST}^+(4)$}
\begin{algorithmic}[1]
    \For {each arrival of $t \in T$}
        \State select either $T_A$ or $T_B$ ``wisely''
        \If {$T_A$ is selected}
            \State color $t$ with the algorithm for the no-dense case (\autoref{alg:no-dense-case})
        \Else
            \State color $t$ with a ``special procedure''
        \EndIf
    \EndFor
\end{algorithmic}
\label{alg:branch}
\end{algorithm}

In the branch framework, the no-dense property must hold between $S$ and $T_A$, and there must exist a deterministic online algorithm that colors the vertices in $T_B$ (line 6) using $\widetilde{O}(n^{14/17})$ colors. Formally, we must solve the following problem:

\paragraph{Division problem.}
Consider the problem $\textsf{ST}^+(4)$ where $|S| = n, |T| \leq n, \mu = \frac{1}{2} n^{3/17}$ (then $|N_S(t)| \geq 2n^{14/17}$ for each $t \in T$). Our task is to divide $T$ into $T_A$ and $T_B$ online, that is, once we receive a vertex $t \in T$, we must immediately decide whether it belongs to $T_A$ or $T_B$ (see \autoref{fig-510}). If we select $T_B$, we must color $t$. Our objective is to ensure that the no-dense property holds between $S$ and $T_A$, and that the vertices in $T_B$ are colored with $\widetilde{O}(n^{14/17})$ colors.

Also, to achieve the no-dense property, immediately after a vertex $t \in T$ arrives (i.e., before selecting $T_A$ or $T_B$), we are allowed to delete some edges from $t$ into $S$. $|N_S(t)| \geq n^{14/17}$ must still hold after the deletion (otherwise this would violate the requirements for the no-dense case).

\begin{figure}[t]
  \centering
  \includegraphics[width=\linewidth]{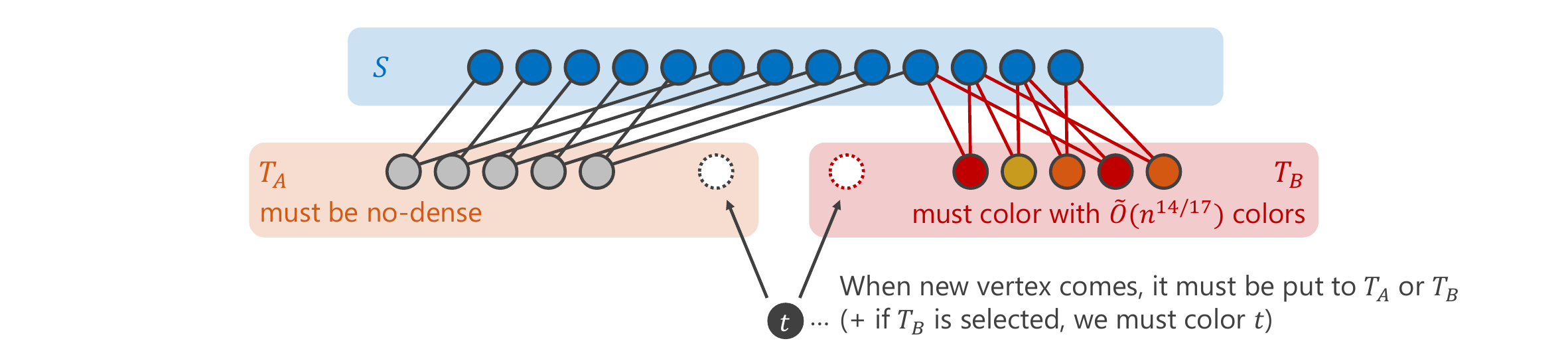}
  \caption{A sketch of the division problem.} \label{fig-510}
\end{figure}

If the division problem is solved, we obtain a deterministic algorithm to solve $\textsf{COL}(4)$ with $\widetilde{O}(n^{14/17})$ colors. To solve the division problem, we first show an important definition and lemma.

\begin{definition}\label{def:k4-beta-common}
    We say that two vertices $u_1, u_2 \in T$ are $\beta$-common if $|N_S(u_1) \cap N_S(u_2)| \geq n^{\beta}$. We say that $T_A$ is $(\alpha, \beta)$-free if there is no vertex set $\{t_1, \dots, t_{n^\alpha+1}\} \subseteq T_A$ such that $t_1, \dots, t_{n^{\alpha}+1}$ arrive in this order, and for all $i$, $t_i$ and $t_{n^{\alpha}+1}$ are $\beta$-common.
\end{definition}

\begin{lemma}\label{lem:k4-ab-free}
    Suppose that the no-dense property does not hold between $S$ and $T_A$. Then, $T_A$ is not $(\alpha_i, \beta_i)$-free for at least one $i$, where $\Delta = \log n, K = \frac{2 \log \log n}{\log n}$, and:
    \begin{equation}
        \label{eq:k4-common-ab}
        \left(\alpha_i, \beta_i\right) = \left(\frac{8}{17} - \left(\frac{2}{17} + K\right) \cdot \frac{\Delta - i}{\Delta}, \frac{14}{17} - \left(\frac{2}{17} + K\right) \cdot \frac{i}{\Delta}\right) \quad (i = 0, \dots, \Delta)
    \end{equation}
\end{lemma}

\begin{proof}
    Since the no-dense property does not hold, there exist subsets $S_D \subseteq S, T_D \subseteq T_A$ such that $|T_D| = n^a$ and $|S_D| \leq n^{a+8/17}$ hold for some $a \leq \frac{8}{17}$, and $|N_{S_D}(t)| \geq \frac{1}{3} n^{14/17}$ for each $t \in T_D$. It suffices to consider the ``hardest'' case that $|N_{S_D}(t)| = \frac{1}{3} n^{14/17}$, because $T_A$ would not \emph{newly} become non-$(\alpha_i, \beta_i)$-free by removing the excess edges.
    
    Suppose for a contradiction that $T_A$ is $(\alpha_i, \beta_i)$-free for all $i$.

    We consider the sum of $c(t_1, t_2) = |N_{S_D}(t_1) \cap N_{S_D}(t_2)|$ over all $\{t_1, t_2\} \subseteq T_D$ (with $t_1$ arriving earlier than $t_2$) in two ways. First, the sum is equal to the number of tuples $(t_1, t_2, s) \ (\{t_1, t_2\} \in T_D, s \in S_D$ such that $s$ is adjacent to both $t_1$ and $t_2$). For each $s$, there are $\frac{1}{2} |N_{T_D}(s)| (|N_{T_D}(s)| - 1)$ ways to choose $\{t_1, t_2\}$ to meet the condition. Hence, the following equation holds:
    \begin{equation*}
        \sum_{\{t_1, t_2\} \subseteq T_D} c(t_1, t_2) = \sum_{s \in S_D} \frac{1}{2} |N_{T_D}(s)| (|N_{T_D}(s)| - 1).
    \end{equation*}
    
    There are $\frac{1}{3} n^{a+14/17}$ edges between $S_D$ and $T_D$. Therefore, $\sum_{s \in S_D} |N_{T_D}(s)| = \frac{1}{3} n^{a+14/17}$. Combining this with $|S_D| \leq n^{a+8/17}$, we obtain $\sum_{s \in S_D} |N_{T_D}(s)|^2 \geq \frac{1}{9} n^{a+20/17}$.\footnote{For any real numbers $x_1, \dots, x_k$, the inequality $x_1^2 + \dots + x_k^2 \geq \frac{1}{k} (x_1 + \dots + x_k)^2$ holds.} Therefore, 
    \begin{equation*}
        \sum_{\{t_1, t_2\} \subseteq T_D} c(t_1, t_2) = \Omega(n^{a+20/17}).
    \end{equation*}

    Second, given that $T_A$ is $(\alpha_i, \beta_i)$-free, for each $t_2$, the number of $t_1$'s that $c(t_1, t_2) \geq n^{\beta_i}$ is at most $n^{\alpha_i}$. This means that the number of $\{t_1, t_2\}$'s that $c(t_1, t_2) \geq n^{\beta_i}$ is at most $n^{a + \alpha_i}$. Let $[c_1, \dots, c_m]$ denote the list of $c(t_1, t_2)$'s for all $\{t_1, t_2\}$'s, sorted in descending order. Then, $c_j \leq n^{\beta_i}$ must hold for $j = n^{a + \alpha_i}+1, \dots, n^{a+\alpha_{i+1}}$. For $j = 1, 2, \dots, n^{a+\alpha_0}$, $c_j \leq n^{\beta_0}$ must hold, by the assumption $|N_{S_D}(t)| = \frac{1}{3} n^{14/17} \ (\leq n^{\beta_0})$. Therefore, the following holds (note that $m \leq n^{a+\alpha_\Delta}$, and $n^{\alpha_i+\beta_i} = n^{20/17-K} = \frac{n^{20/17}}{(\log n)^2}$ for each $i$):
    \begin{align*}
        \sum_{i=1}^m c_i & \leq n^{a+\alpha_0} \cdot n^{\beta_0} + \sum_{i=0}^{\Delta-1} (n^{a+\alpha_{i+1}} - n^{a+\alpha_i}) \cdot n^{\beta_i} \\
        & = n^{a+20/17-K} + \sum_{i=0}^{\Delta-1} (n^{(2/17+K)/\Delta}-1) \cdot n^{a+20/17-K} \\
        & = O(n^{a+20/17}/\log n).
    \end{align*}
    
    Thus, the sum of $c(t_1, t_2)$'s is both $\Omega(n^{a+20/17})$ and $O(\frac{n^{a+20/17}}{\log n})$, which is a contradiction. Therefore, \autoref{lem:k4-ab-free} follows.
\end{proof}

By \autoref{lem:k4-ab-free}, the division problem is solved if $T_A$ is $(\alpha_i, \beta_i)$-free for all $i = 0, \dots, \Delta$. Now, we aim to achieve $(\alpha_i, \beta_i)$-free, instead of the no-dense property. We call this strategy the \emph{Common \& Simplify technique}.

\subsection{Common \& Simplify Technique: Algorithm}

Since handling multiple pairs at once is quite challenging, we consider an algorithm that achieves $(\alpha, \beta)$-free for a single pair $(\alpha, \beta)$.

\paragraph{Overview.}
Recall that $\textsf{SP}^+(4, 2)$ can be solved using $\widetilde{O}(|T|^{2/3} \mu^{1/3})$ colors (\autoref{lem:k5-col}). The strategy is to create $2$-color sets using the dense structure. The key intuition is as follows:

\begin{quote}
    Using the vertices in $T$ that have many common neighbors, we aim to obtain large 2-color sets. Indeed, if two $\beta$-common vertices $u_1, u_2 \in T$ have different colors in some 4-coloring of $G$, then $N_S(u_1) \cap N_S(u_2)$ is a 2-color set in some 4-coloring of $G$, with $n^{\beta}$ vertices, which is sufficiently large.
\end{quote}

When a vertex $t \in T$ arrives, we normally put $t$ into $T_A$, but in case that putting it would make $T_A$ not $(\alpha, \beta)$-free, we select $T_B$ instead. In such a case, we create the set $D := \{u \in T_A: \text{$u$ and $t$ are $\beta$-common}\}$, which we call a \emph{dense group}, and it satisfies $|D| \geq n^{\alpha}$ (\autoref{fig-513} left).

This dense group is used for the future vertices $t \in T$ that are $\beta$-common to some vertices in $D$. Let $F$ be the set of such vertices. As long as there is no edge inside $F$, we can color each vertex in $F$ using the same color (\autoref{fig-513} middle). If an edge $\{u_1, u_2\}$ appears inside $F$, we can create candidate 2-color sets. Indeed, when $u_j$ is $\beta$-common with $w_j \in D$ (for $j = 1, 2$) and each vertex in $F$ is $\beta$-common with $x$, we can prove that at least one of
\begin{equation}
    \label{eq:k4-common-2color}
    [N_S(u_1) \cap N_S(w_1)], [N_S(w_1) \cap N_S(x)], [N_S(x) \cap N_S(w_2)], [N_S(w_2) \cap N_S(u_2)]
\end{equation}
must form a 2-color set (\autoref{fig-513} right), where $[X]$ is the first $n^{\beta}$ vertices of $X \subseteq V$. This is because $u_1$ and $u_2$ always have different colors in any 4-coloring of $G$, and therefore at least one of $(u_1, w_1), (w_1, x), (x, w_2), (w_2, u_2)$ has different colors in some 4-coloring of $G$. We initiate instances of $\textsf{SP}^+(4, 2)$ based on these sets, and use them to color the future vertices effectively.

\begin{figure}[t]
  \centering
  \includegraphics[width=\linewidth]{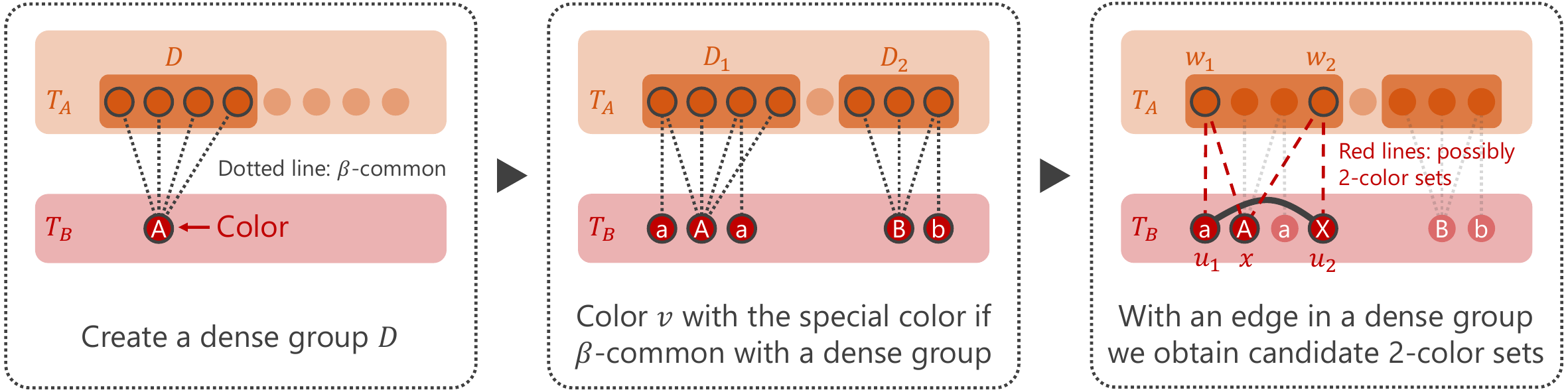}
  \caption{The sketch of the algorithm between $T_A$ and $T_B$. The dotted line indicates that two vertices are $\beta$-common (it does not always mean they are adjacent). Colors \texttt{A}, \texttt{a}, \texttt{B}, \texttt{b}, and \texttt{X} are different.} \label{fig-513}
\end{figure}

\paragraph{Parameters.}
When we create an instance $I$ of $\textsf{SP}^+(4, 2)$, we always set $|S(I)| = n^{\beta}$ and $\mu(I) = 8 n^{3/17} \log n$. $|N_{S(I)}(t)| \geq \frac{n^{\beta-3/17}}{8 \log n}$ holds for each $t \in T(I)$. This uses $\widetilde{O}(|T|^{2/3} n^{1/17})$ colors by \autoref{lem:k5-col}.

\paragraph{Implementation.}
We denote the created instances of $\textsf{SP}^+(4, 2)$ by $I_{i,j}$, where $i$ follows chronological order, and $I_{i,1}, \dots, I_{i,4}$ are created at the same time. We also denote the vertices that caused the creation of dense groups by $x_1, x_2, \dots$. We denote by $D_i$ the dense group created by $x_i$, and $F_i$ for the set of vertices which are $\beta$-common to some vertices in $D_i$ (and colored with the same color).

We also maintain a counter $r$ for the number of times instances of $\textsf{SP}^+(4, 2)$ are created (initially zero), and a counter $c$ for the number of times a dense group is created (initially zero). Whenever a vertex $t \in T$ arrives, we do the following procedure:
\begin{enumerate}
    \item If there exists a non-terminated instance $I_{i,j}$ such that $|N_{S(I_{i,j})}(t)| \geq \frac{n^{\beta-3/17}}{8 \log n}$, we choose one such instance $I_{i,j}$ arbitrarily, and color $t$ using this instance $I_{i,j}$. If $I_{i,1}, \dots, I_{i,4}$ are all terminated as a result, then we declare that $G$ is not $4$-colorable and terminate the original problem $\textsf{ST}^+(4)$.
    \item Otherwise, to ensure the no-dense property between $S$ and $T_A$, we delete the edges between $t$ and $S_{\mathrm{used}}$, where $S_{\mathrm{used}} := \bigcup \{S(I_{i,j}): \text{$I_{i,j}$ is active}\}$. Then:
    \begin{enumerate}
        \item If there exists $i$ such that $t$ is $\beta$-common to some vertices in $D_i$, we put $t$ into $F_i$ and color $t$; we can color all vertices in $F_i$ with the same color, as long as there is no edge inside $F_i$. If there is an edge $\{u_1, u_2\}$ inside $F_i$, we create four instances $I_{r+1,1}, \dots, I_{r+1,4}$ of $\textsf{SP}^+(4, 2)$ as explained in the overview (\autoref{eq:k4-common-2color}), which increments $r$ by $1$. We then ``reset'' all dense groups; we delete all $x_i, D_i, F_i$'s and reset $c$ to $0$.
        \item Otherwise, if $D' := \{u \in T_A : \text{$u$ and $t$ are $\beta$-common}\}$ has size $n^{\alpha}$ or more, we color $t$ using a new color, and create a dense group: $x_{c+1} \gets t, D_{c+1} \gets D', F_{c+1} \gets \emptyset$ and increment $c$ by $1$ (this process is the key to keep $(\alpha, \beta)$-free). If not, we put $t$ into $T_A$.
    \end{enumerate}
\end{enumerate}
The pseudocode of this algorithm is given in \autoref{alg:division}.

\begin{algorithm}[!t]
\caption{The algorithm for the division problem to make it $(\alpha, \beta)$-free}
\begin{algorithmic}[1]
    \State $r \gets 0, c \gets 0$
    \For {each arrival of $t \in T$}
        \If {there exists an active instance $I_{i,j}$ such that $|N_{S(I_{i,j})}(t)| \geq \frac{n^{\beta-3/17}}{8 \log n}$}
            \State color $t$ using $I_{i,j}$
            \If {$I_{i,1}, \dots, I_{i,4}$ are all terminated}
                \State \Return that $G$ is not $4$-colorable and terminate $\textsf{ST}^+(4)$
            \EndIf
        \Else
            \State delete all edges between $t$ and $S_{\mathrm{used}} := \bigcup \{S(I_{i,j}) : \text{$I_{i,j}$ is active}\}$
            \If {there exists $i$ such that $t$ is $\beta$-common to some vertices in $D_i$}
                \State $F_i \leftarrow F_i \cup \{t\}$
                \If {there exist $u_1, u_2 \in F_i$ that are adjacent}
                    \State color $t$ with a new color
                    \State $w_j \gets$ a vertex in $D_i$ that is $\beta$-common with $u_j$, for $j = 1, 2$
                    \State initiate instances $I_{r+1,1}, \dots, I_{r+1,4}$ of $\textsf{SP}^+(4, 2)$ where $S(I_{r+1,j})$ is set to each of $[N_S(u_1) \cap N_S(w_1)], [N_S(w_1) \cap N_S(x_i)], [N_S(x_i) \cap N_S(w_2)], [N_S(w_2) \cap N_S(u_2)]$ \Comment{For $X \subseteq S$, we define $[X]$ to be the first $n^{\beta}$ vertices of $X$}
                    \State $r \gets r+1$, delete all $x_i, D_i, F_i$'s, and then reset $c \gets 0$
                \Else
                    \State color $t$ with the color reserved for $F_i$
                \EndIf
            \Else
                \State $D' := \{u \in T_A : \text{$u$ and $t$ are $\beta$-common}\}$
                \If {$|D'| \geq n^{\alpha}$}
                    \State color $t$ with a new color
                    \State $x_{c+1} \gets t, D_{c+1} \gets D', F_{c+1} \gets \emptyset$, and then $c \gets c+1$
                \Else
                    \State put $t$ into $T_A$
                \EndIf
            \EndIf
        \EndIf
    \EndFor
\end{algorithmic}
\label{alg:division}
\end{algorithm}

\subsection{Common \& Simplify Technique: Correctness \& Analysis}

Next, we prove that the algorithm works correctly. The missing piece is to show that at least one of the four created instances actually admit 2-color sets (\autoref{lem:k4-common-path}), and at least $n^{14/17}$ edges remain when we delete the edges between $t$ and $S_{\mathrm{used}}$ (\autoref{lem:k4-common-edges}). After that, we analyze the number of colors used, and show that \autoref{alg:division} achieves the target of $\widetilde{O}(n^{14/17})$ colors (\autoref{lem:k4-common-colors}).

\begin{lemma}
    \label{lem:k4-common-path}
    If instances $I_{i,1}, \dots, I_{i,4}$ are all terminated, then $G$ is not $4$-colorable.
\end{lemma}

\begin{proof}
    Let $\varphi: V \to \{1, 2, 3, 4\}$ be a $4$-coloring of $G$. Since $I_{i,1}, \dots, I_{i,4}$ are all terminated, each of the four sets in \autoref{eq:k4-common-2color} must have at least $3$ distinct colors under $\varphi$. Let $x, u_1, u_2, w_1, w_2$ be the vertices that created the instances $I_{i,1}, \dots, I_{i,4}$ (the variable names are according to \autoref{eq:k4-common-2color}). Since $u_1$ and $u_2$ are adjacent, $\varphi(u_1) \neq \varphi(u_2)$. Thus, at least one of $\varphi(u_1) \neq \varphi(w_1), \varphi(w_1) \neq \varphi(x), \varphi(x) \neq \varphi(w_2), \varphi(w_2) \neq \varphi(u_2)$ holds. This means that at least one of the four sets in \autoref{eq:k4-common-2color} has at most $2$ distinct colors under $\varphi$, which is a contradiction. Therefore, $G$ is not $4$-colorable.
\end{proof}

\begin{lemma}
    \label{lem:k4-common-num}
    $r \leq n^{1-\beta}$.
\end{lemma}

\begin{proof}
    For each $i = 1, \dots, r$, there exists an active instance $I_{i,t_i}$ for some $t_i$. Let $R_i := S(I_{i,t_i})$. Then, for each $i, j \ (i < j)$, $R_i \cap R_j = \emptyset$. This is because the reset of dense groups happen when $I_{i,t_i}$ is initiated, and hence the vertices $x, u_1, u_2, w_1, w_2$ (from \autoref{eq:k4-common-2color}) that created $I_{j,1}, \dots, I_{j,4}$ all arrive later than initiating $I_{i,t_i}$. Thus, the edges from $\{x, u_1, u_2, w_1, w_2\}$ to $S(I_{i,t_i})$ are all deleted, so by \autoref{eq:k4-common-2color}, $R_j = S(I_{j,t_j})$ does not overlap with $R_i = S(I_{i,t_i})$. Since $|R_i| = n^{\beta}$, we have $r \leq n^{1-\beta}$.
\end{proof}

\begin{lemma}
    \label{lem:k4-common-edges}
    At most $\frac{n^{14/17}}{2 \log n}$ edges are deleted for each $t \in T$.
\end{lemma}

\begin{proof}
    We delete the edges between $t$ and $S_{\mathrm{used}}$. For each active instance $I_{i,j}$, we have $|N_{S(I_{i,j})}(t)| < \frac{n^{\beta-3/17}}{8 \log n}$ (because otherwise $t$ should be colored by $I_{i,j}$). Also, the number such instances $I_{i,j}$ is at most $4n^{1-\beta}$ by \autoref{lem:k4-common-num}. Therefore, at most $\frac{n^{\beta-3/17}}{8 \log n} \cdot 4n^{1-\beta} = \frac{n^{14/17}}{2 \log n}$ edges are deleted.
\end{proof}

\begin{lemma}
    \label{lem:k4-common-colors}
    \autoref{alg:division} uses $\widetilde{O}(n^{14/17})$ colors, if $\alpha + \beta \geq \frac{20}{17} - K$ and $\beta \geq \frac{12}{17} - K$.
\end{lemma}

\begin{proof}
    First, we analyze the number of colors used in $\textsf{SP}^+(4, 2)$. Since we create at most $4n^{1-\beta}$ instances of $\textsf{SP}^+(4, 2)$ by \autoref{lem:k4-common-num}, and each instance $I_{i,j}$ of $\textsf{SP}^+(4, 2)$ uses $\widetilde{O}(|T(I_{i,j})|^{2/3} n^{1/17})$ colors, we use $\widetilde{O}(n^{37/51+(1-\beta)/3})$ colors in total by \autoref{lem:prelim-ineq}, which is $\widetilde{O}(n^{14/17})$.

    Next, we analyze the number of colors used to handle dense groups. For each dense group, three colors will be used: one is for $x_i$, one is for the last vertex in $F_i$, and one is for the other vertices in $F_i$. Also, $D_i \cap D_j = \emptyset$ for each $i, j \ (i < j)$, because otherwise, $x_j$ would have been put into $F_i$ (instead of creating a dense group). Since $|D_i| \geq n^{\alpha}$ for each $i$, we have $c \leq n^{1-\alpha}$. Also, these dense groups are reset at most $n^{1-\beta}$ times by \autoref{lem:k4-common-num}. Therefore, the number of colors used here is $3 \cdot n^{1-\alpha} \cdot n^{1-\beta} = \widetilde{O}(n^{2-(\alpha+\beta)})$, which is $\widetilde{O}(n^{14/17})$.
\end{proof}

\subsection{The Final Build-up}

Finally, we put everything together to show that we can color $4$-colorable graphs online with $\widetilde{O}(n^{14/17})$ colors. We begin by presenting the algorithm for the original division problem.

\begin{lemma}
    \label{lem:k4-division}
    The division problem can be solved with $\widetilde{O}(n^{14/17})$ colors.
\end{lemma}

\begin{proof}
    To solve the division problem, we run instances $I_0, \dots, I_{\Delta}$ of \autoref{alg:division} (with the same $S$) in parallel, where $I_i$ runs the procedure to achieve $(\alpha_i, \beta_i)$-free, so that each instance uses disjoint colors. We put each arriving vertex $t \in T$ for the original division problem into $T(I_0)$. For each $i = 0, \dots, \Delta-1$, if $t$ is put into $T_A(I_i)$, then we put $t$ into $T(I_{i+1})$. If $t$ is put into $T_A(I_{\Delta})$, we put $t$ into $T_A$ in the original division problem.
    
    For each $i$, if $T(I_i)$ is $(\alpha_0, \beta_0), \dots, (\alpha_{i-1}, \beta_{i-1})$-free, then $T(I_i) = T_A(I_{i-1})$ becomes $(\alpha_0, \beta_0), \dots, (\alpha_i, \beta_i)$-free. Thus, by induction on $i$, $T_A = T_A(I_{\Delta})$ is $(\alpha_i, \beta_i)$-free for all $i = 0, \dots, \Delta$. Thus, the no-dense property holds between $S$ and $T_A$ by \autoref{lem:k4-ab-free}, meeting the requirement for the division problem.
    
    For each $t \in T$, since at most $\frac{n^{14/17}}{2 \log n}$ edges are deleted for each run of \autoref{alg:division} (\autoref{lem:k4-common-edges}), at most $\frac{n^{14/17}}{2 \log n} \cdot (\Delta+1) \leq n^{14/17}$ edges are deleted in total. Thus, at least $2n^{14/17} - n^{14/17} = n^{14/17}$ edges remain, which does not violate the input condition for the no-dense case. Moreover, we use $\widetilde{O}(n^{14/17}) \cdot (\Delta+1) = \widetilde{O}(n^{14/17})$ colors due to $\Delta = \log n$ and \autoref{lem:k4-common-colors}.
\end{proof}

\begin{theorem}
    \label{thm:k4-col}
    $\textsf{COL}(4)$ can be solved with $\widetilde{O}(n^{14/17})$ colors.
\end{theorem}

\begin{proof}
    The division problem can be solved with $\widetilde{O}(n^{14/17})$ colors by \autoref{lem:k4-division}. The vertices in $T_A$ can be colored with \autoref{alg:no-dense-case} for the no-dense case, which uses $\widetilde{O}(n^{14/17})$ colors by \autoref{lem:k4-nodense}. We also use $2n^{14/17}$ colors for First-Fit (to reduce $\textsf{COL}'(4)$ to $\textsf{ST}^+(4)$). Thus, $\textsf{COL}'(4)$ can be solved with $\widetilde{O}(n^{14/17})$ colors. By \autoref{lem:k5-n-is-known}, $\textsf{COL}(4)$ can also be solved with $\widetilde{O}(n^{14/17})$ colors.
\end{proof}

In addition, by combining with our results in \autoref{sec:k5}, we can also improve the results for $k \geq 5$. The following \autoref{col:k4-improvement} slightly improves the results for $k \geq 5$ over \autoref{thm:k5-col}. Our best results, compared to Kierstead's algorithm \cite{Kie98}, are shown in \autoref{tab:intro} in \autoref{subsec:intro-contributions}.

\begin{corollary}
    \label{col:k4-improvement}
    For all $k \geq 4$, $\textsf{COL}(k)$ can be solved with $\widetilde{O}(n^{1-\frac{1}{k(k-1)/2-1/3}})$ colors.
\end{corollary}

\begin{proof}
    First, we can solve $\textsf{COL}(4)$ with $n^{14/17}$ colors by \autoref{thm:k4-col}, which proves the statement for $k = 4$. Next, we consider the case $k \geq 5$. We prove this by induction on $k$. \autoref{lem:k5-col} shows that $\textsf{COL}(k)$ can be solved with $\widetilde{O}(n^{1-\frac{\alpha}{1+(k-1)\alpha}})$ colors, where we can substitute $\alpha = \frac{1}{(k-1)(k-2)/2-1/3}$ by the induction hypothesis. Thus, $\textsf{COL}(k)$ can be solved with $\widetilde{O}(n^{1-\frac{1}{k(k-1)/2-1/3}})$ colors.
\end{proof}

\section{Randomized Algorithm for Bipartite Graphs}\label{sec:k2-upper}

In this section, we show a randomized online coloring algorithm for bipartite graphs, against an oblivious adversary (the case where the graph is determined before the coloring begins).

\subsection{The Algorithm by Lov\'{a}sz, Saks, and Trotter}

Firstly, we review a previously-known deterministic online coloring algorithm for bipartite graphs by Lov\'{a}sz, Saks, and Trotter (1989) \cite{LST89}. This algorithm uses $2 \log (n+1)$ colors in the worst case, and is an optimal deterministic algorithm up to a constant number of colors \cite{GKM+14}. The idea of this algorithm is to always try to achieve the ``bipartite coloring'', and once it becomes impossible, we start using a new color. 
Note that the hard case is when there are many components of bipartite graphs, and some of them start merging when a new vertex $v$ arrives. In this case, we may need a new color for $v$. 

Formally, we assign a \emph{level} to each connected component, so that a level-$\ell$ component uses colors $1, 2, \dots, 2\ell$. The meaning of the level is how many times the algorithm fails to achieve the ``bipartite coloring'' and is forced to use a new color. Therefore, we would like to estimate the maximum level.

We call a level-$\ell$ component \emph{matched} if the colors $2\ell-1$ and $2\ell$ are based on the ``correct'' $2$-coloring; that is, distances of any two vertices with color $2\ell-1$ are even, distances of any two vertices with color $2\ell$ are even, and distances of any vertex with color $2\ell-1$ and any vertex with color $2\ell$ are odd. At any time, every component is supposed to be matched. So, when a new vertex $v$ arrives and some component becomes no longer matched, we increase the component's level by $1$ and use a new color for vertex $v$. Otherwise, we color $v$ with either color $2\ell-1$ or $2\ell$, whichever is matched. Examples are shown in \autoref{fig-522}.

\begin{figure}[htbp]
    \centering
    \includegraphics[width=\linewidth]{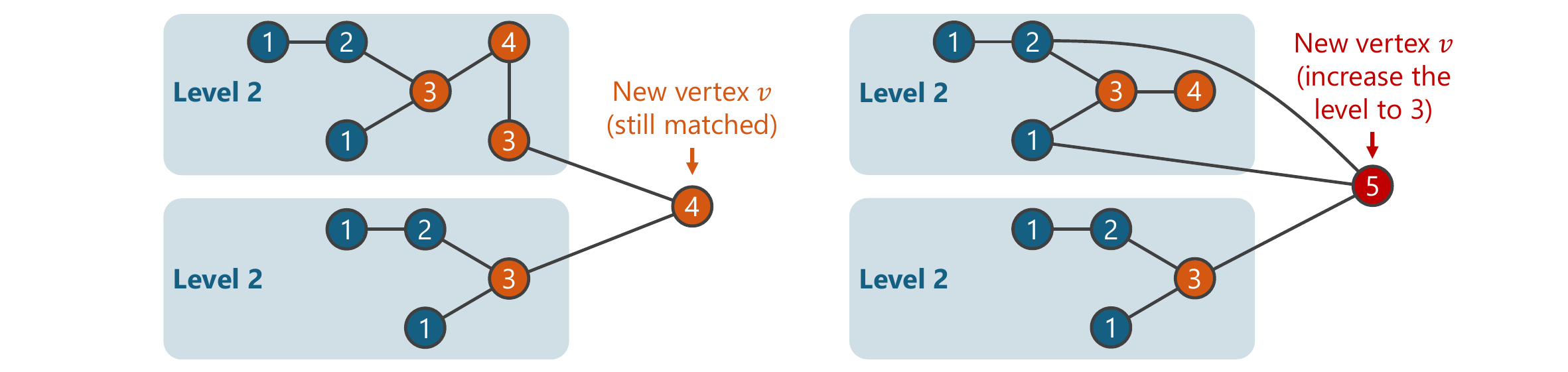}
    \caption{Two examples of the algorithm by Lov\'{a}sz, Saks, and Trotter \cite{LST89}. Left figure shows an example when the component is still matched, and the right figure shows an example when we must increase the level from 2 to 3. The number inside each vertex is the color used.}
    \label{fig-522}
\end{figure}

Lov\'{a}sz, Saks, and Trotter's algorithm is described in \autoref{alg:lst89} (\textsc{LST89}). We suppose that $V = \{v_1, \dots, v_n\}$, arriving in the order $v_1, \dots, v_n$. When vertex $v_i$ arrives, the resulting component including $v_i$ is denoted as $C_i$, and its level, after coloring $v_i$, is denoted as $\ell_i$.

\begin{algorithm}[htbp]
	\caption{$\textsc{LST89}(G = (V, E))$}
    \label{alg:lst89}
    \begin{algorithmic}[1]
        \State $S \gets \emptyset$ \Comment{the set of indices of connected components of the current graph}
        \For{$i = 1, \dots, n$}
            \State $x_{i,1}, \dots, x_{i,k_i} \gets \{j \in S : \exists u \in C_j, uv_i \in E\}$ \Comment{the indices of components that are adjacent to $v_i$}
            \State $C_i \gets C_{x_{i,1}} \cup \dots \cup C_{x_{i,k_i}} \cup \{v_i\}$
            \State $S \gets (S \setminus \{x_{i,1}, \dots, x_{i,k_i}\}) \cup \{i\}$
            \If{$k_i = 0$}
                \State Color $v_i$ by color $1$
                \State $\ell_i \gets 1$
            \Else
                \State $\ell^*_i \gets \max(\ell_{x_{i,1}}, \dots, \ell_{x_{i,k_i}})$
                \If{$C_i$ is matched (as a level-$\ell^*_i$ component)}
                    \State Color $v_i$ by either color $2\ell^*_i-1$ or $2\ell^*_i$ so that $C_i$ remains matched
                    \State $\ell_i \gets \ell^*_i$
                \Else
                    \State Color $v_i$ by color $2\ell^*_i+1$
                    \State $\ell_i \gets \ell^*_i+1$
                \EndIf
            \EndIf
        \EndFor
    \end{algorithmic}
\end{algorithm}

\begin{theorem}[\cite{LST89}] \label{thm:lst89}
    \autoref{alg:lst89} (\textsc{LST89}) uses at most $2 \log (n+1)$ colors for any graph.
\end{theorem}

\begin{proof}
    Let $a_{\ell}$ be the minimum $n$ that level-$\ell$ components can appear. For $\ell \geq 2$, a level-$\ell$ component can be formed only when two or more level-$(\ell-1)$ components are merged into one component. Therefore, $a_{\ell} \geq 2a_{\ell-1} + 1$. Then, $a_{\ell} \geq 2^\ell-1$ follows, so the algorithm uses at most $2 \log (n+1)$ colors.
\end{proof}

\subsection{The Randomized Algorithm and Its Probabilistic Model}

In this subsection, we present a randomized algorithm for online coloring of bipartite graphs. This algorithm is essentially the randomized version of \textsc{LST89}, which makes the following modifications to \autoref{alg:lst89}:
\begin{itemize}
    \item Change line 7 to ``Color $v_i$ by color $1$ or $2$ with probability $\frac{1}{2}$ each''
    \item Change line 15 to ``Color $v_i$ by color $2\ell^*_i+1$ or $2\ell^*_i+2$ with probability $\frac{1}{2}$ each''
\end{itemize}
The new algorithm is denoted by \textsc{RandomizedLST}. The number of colors will be improved from \textsc{LST89}, because the adversarial case, which requires a new color every time, can be avoided in expectation. The goal of this section is to prove that \textsc{RandomizedLST} uses at most $1.034 \log n + O(1)$ colors in expectation.

\paragraph{Probabilistic model of the performance.}
We analyze the performance of \textsc{RandomizedLST} using a rooted forest $T$. When a vertex $v_i$ arrives, the component $C_i$ is formed by merging $v_i$ and some existing connected components, say $C_{x_{i,1}}, \dots, C_{x_{i,k_i}}$. We represent in $T$ this relation of how the components are merged. This is defined to be an $n$-vertex forest where the vertices are labeled $1, \dots, n$, and the children of vertex $i$ are $x_{i,1}, ..., x_{i,k_i}$.\footnote{The constructed $T$ is indeed a rooted forest. It has no cycles due to $x_{i,j} < i$. No vertex is a child of multiple vertices because once a connected component is merged, it is no longer a connected component of the graph.} Then, the probability distribution of $(\ell_1, \dots, \ell_n)$ in \textsc{RandomizedLST} can be simulated by \autoref{alg:leveling_alternative}, which is shown in the following lemma.

\begin{lemma}
    The probability distributions of $(\ell_1, \dots, \ell_n)$ generated by \textsc{RandomizedLST} and by \autoref{alg:leveling_alternative} are the same.
    \label{lem:leveling_lemma}
\end{lemma}

\begin{algorithm}[htbp]
    \caption{Alternative algorithm to generate levels $\ell_1, \dots, \ell_n$}
    \label{alg:leveling_alternative}
    \begin{algorithmic}[1]
        \For{$i = 1, \dots, n$}
            \If{vertex $v_i$ is a leaf in $T$}
                \State $\ell_i \gets 1$
            \Else
                \State $\ell^*_i \gets \max(\ell_{x_{i,1}}, \dots, \ell_{x_{i,k_i}})$
                \State $c_i \gets$ (number of $j$'s that $\ell_{x_{i,j}} = \ell^*_i$)
                \State $\ell_i \gets$ ($\ell^*_i$ with probability $2^{-(c_i-1)}$, and $\ell^*_i+1$ with probability $1-2^{-(c_i-1)}$)
            \EndIf
        \EndFor
    \end{algorithmic}
\end{algorithm}

\begin{proof}
    It suffices to show that in \textsc{RandomizedLST}, the conditional probability that $\ell_i = \ell^*_i$ given $(\ell_1, \dots, \ell_{i-1})$ is always $2^{-(c_i-1)}$, where $\ell^*_i := \max(\ell_{x_{i,1}}, \dots, \ell_{x_{i,k_i}})$ and $c_i :=$ (the number of $j$'s that $\ell_{x_{i, j}} = \ell^*_i$). We relate a coloring of $v_1, \dots, v_{i-1}$ to the colorings obtained by ``flipping'' the color of all the vertices in an arbitrary $\mathcal{C}' \subseteq \{C_{x_{i, j}} : \ell_{x_{i, j}} = \ell^*_i\}$. Formally, flipping the color means that odd-numbered color $2\ell-1$ becomes color $2\ell$, and even-numbered color $2\ell$ becomes color $2\ell-1$. Then, we obtain $2^{c_i}$ colorings. These colorings have the same $(\ell_1, \dots, \ell_{i-1})$, and all of them appear with the same probability because the related coloring appears when all of the probabilistic decisions in $\mathcal{C}'$ are inverted. However, $C_i$ is matched for only two of them. Therefore, the probability that $\ell_i = \ell^*_i$ is $\frac{2}{2^{c_i}} = 2^{-(c_i-1)}$, and otherwise $\ell_i$ becomes $\ell^*_i+1$.
\end{proof}

The expected number of colors in \textsc{RandomizedLST} is, obviously, between $2 \mathbb{E}[\max(\ell_1, \dots, \ell_n)] - 1$ and $2 \mathbb{E}[\max(\ell_1, \dots, \ell_n)]$. Especially when $T$ is a tree, $\ell_n = \max(\ell_1, \dots, \ell_n)$ because vertex $n$ is the root of $T$; in this case, it is crucial to estimate $\mathbb{E}[\ell_n]$. It turns out that, when we search for the graphs with the worst expected number of colors, we only have to consider the case when $G$ is connected (i.e., $T$ is a tree). This is because, when $G$ is not connected, we can modify the graph by adding edges between $v_n$ and all the other connected components, and the levels will not decrease.

\subsection{Preliminaries for the Analysis}

In this subsection, we prove the following lemma, which shows that it is sufficient to consider when $T$ is a binary tree. Hereafter, we call $\ell_v$ the level of vertex $v$ (of $T$) and denote the root of $T$ as $\mathrm{root}(T)$.

\begin{lemma}
    There is a binary tree $T$ that maximizes $\mathbb{E}[\ell_{\mathrm{root}(T)}]$ among all trees with $m$ leaves.
    \label{lem:binarytree}
\end{lemma}

First, we define the following preorder $\preceq$ on the set of rooted trees $\mathcal{T}$:

\begin{definition}
    For $T_1, T_2 \in \mathcal{T}$, $T_1 \preceq T_2$ if $\mathrm{Pr}_{T=T_1}[\ell_{\mathrm{root}(T_1)} \geq t] \geq \mathrm{Pr}_{T=T_2}[\ell_{\mathrm{root}(T_2)} \geq t]$ for all $t$.
\end{definition}

It is easy to see that the relation $\preceq$ satisfies reflexivity and transitivity. Also, if $T_1 \preceq T_2$, then $\mathbb{E}[\ell_{\mathrm{root}(T_1)}] \geq \mathbb{E}[\ell_{\mathrm{root}(T_2)}]$ holds because $\mathbb{E}[X] = \sum_{t=1}^\infty \mathrm{Pr}[X \geq t]$ for any random variable $X$ that takes a positive integer.

\begin{lemma}
    Let $T_1$ be a tree, and let $T_2$ be a tree created by replacing a subtree $T'_1$ (of $T_1$) with a tree $T'_2$. If $T'_1 \preceq T'_2$, then $T_1 \preceq T_2$ holds.
    \label{lem:monotone}
\end{lemma}

\begin{proof}
    Let $r = \mathrm{root}(T_1) \ (= \mathrm{root}(T_2))$. We first consider the case that $\mathrm{root}(T'_1)$ (also $\mathrm{root}(T'_2)$) is a child of $r$. Let $x_1, \dots, x_k \ (x_1 = \mathrm{root}(T'_1))$ be the children of $r$. If $k = 1$, then $\ell_r = \ell_{x_1}$, so it is obvious that $T_1 \preceq T_2$ holds given that $T'_1 \preceq T'_2$. So suppose $k \geq 2$. If the levels of $x_2, \dots, x_k$ are fixed, the level of $r$ is decided in the following way, where $\ell' = \max(\ell_{x_2}, \dots, \ell_{x_k})$, and $c' = (\text{number of $j$'s $(j \geq 2)$ that $\ell_{x_j} = \ell'$})$:
    \begin{equation*}
        \ell_r = \begin{cases}
            \text{$\ell'$ with probability $2^{-(c'-1)}$, and $\ell'+1$ with probability $1-2^{-(c'-1)}$} & (\ell_{x_1} < \ell') \\
            \text{$\ell'$ with probability $2^{-c'}$, and $\ell'+1$ with probability $1-2^{-c'}$} & (\ell_{x_1} = \ell') \\
            \ell_{x_1} & (\ell_{x_1} > \ell')
        \end{cases}
    \end{equation*}
    Therefore:
    \begin{equation*}
        \mathrm{Pr}[\ell_r \geq t] =
        \begin{cases}
            1 & (t \leq \ell') \\
            (1 - 2^{-(c'-1)}) + 2^{-c'} \cdot \mathrm{Pr}[\ell_{x_1} \geq \ell'] + 2^{-c'} \cdot \mathrm{Pr}[\ell_{x_1} \geq \ell'+1] & (t = \ell'+1) \\
            \mathrm{Pr}[\ell_{x_1} \geq t] & (t \geq \ell'+2)
        \end{cases}
    \end{equation*}
    When $T'_1$ is replaced by $T'_2$, $\mathrm{Pr}[\ell_{x_1} \geq t]$ does not decrease for all $t$, so $\mathrm{Pr}[\ell_r \geq t]$ does not decrease for all $t$, which means that $T_1 \preceq T_2$.

    Next, we prove the general case. Let $v_0, v_1, \dots, v_k \ (v_k = r)$ be a path from $\mathrm{root}(T'_1)$ to $r$, and let $T_{i, j}$ be the subtree of $v_j$ in $T_i$. The assumption $T'_1 \preceq T'_2$ means $T_{1, v_0} \preceq T_{2, v_0}$. Also, if $T_{1, v_j} \preceq T_{2, v_j}$, then $T_{1, v_{j+1}} \preceq T_{2, v_{j+1}}$, as shown above. Therefore, $(T_1 =) \ T_{1, v_k} \preceq T_{2, v_k} \ (= T_2)$ is shown by induction.
\end{proof}

\begin{proof}[Proof of \autoref{lem:binarytree}]
    Let $T_0$ be a tree that $\mathbb{E}[\ell_{\mathrm{root}(T_0)}]$ takes the maximum value among all trees with $m$ leaves. If $T_0$ has some vertex $v$ that has one child or three or more children, we perform the following operation to make a new tree $T_1$. Let $p$ be the parent of $v$ (if it exists), and let $x_1, \dots, x_k$ be the children of $v$.
    \begin{itemize}
        \item Case $k = 1$: Delete vertex $v$, and set the parent of $x_1$ to $p$ (if it exists). We define this operation as \emph{contraction}. See the left of \autoref{fig:binarytree}. 
        \item Case $k \geq 3$: Delete vertex $v$, and instead create vertices $w_1, \dots, w_{k-1}$. Set the parents of $x_1$ and $x_2$ to $w_1$, and for $j = 1, \dots, k-2$, set the parents of $w_j$ and $x_{j+2}$ to $w_{j+1}$, and finally set the parent of $w_{k-1}$ to $p$ (if it exists).  See the right of \autoref{fig:binarytree}. 
    \end{itemize}

    \begin{figure}[htbp]
        \centering
        \includegraphics[width=0.8\linewidth]{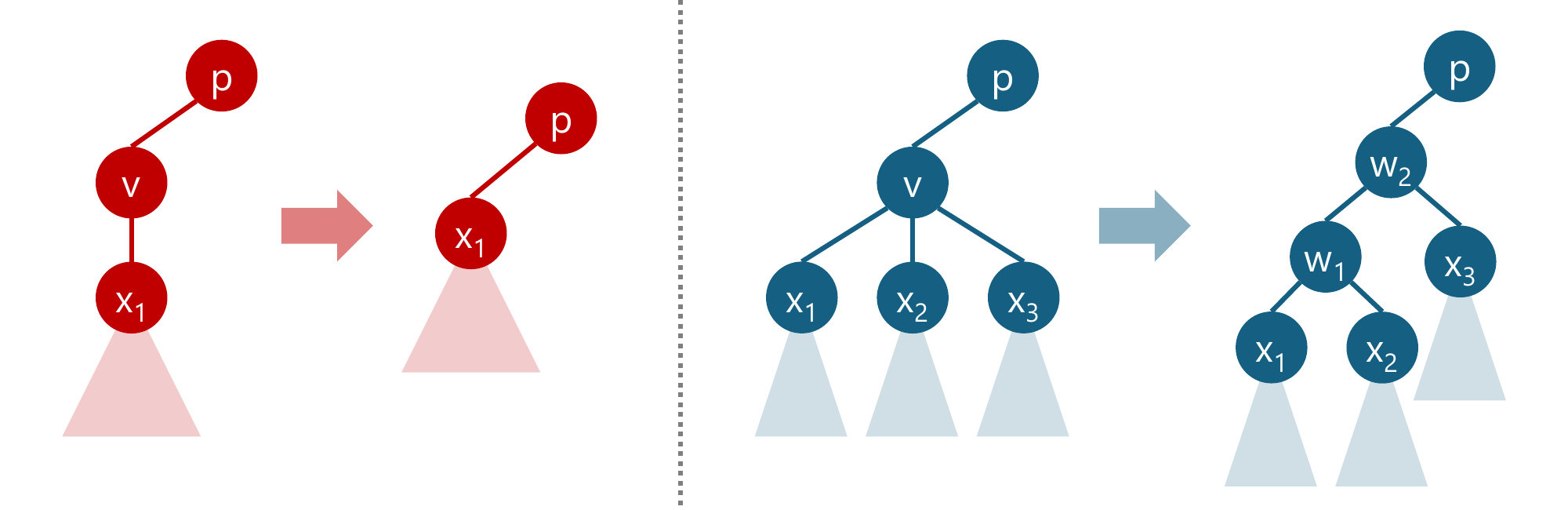}
        \caption{The operations to modify a tree, for $k = 1$ case (left) and $k \geq 3$ case (right)}
        \label{fig:binarytree}
    \end{figure}
    
    We prove that $T_0 \preceq T_1$ holds.
    \begin{itemize}
        \item Case $k = 1$: The operation is to replace $T'_0$ with $T'_1$, where $T'_0$ is the subtree of $v$ in $T_0$, and $T'_1$ is the subtree of $x_1$ in $T_1$. Since $v$ has one child, $\ell_v = \ell_{x_1}$ holds; therefore, $T'_0 \preceq T'_1$. By  \autoref{lem:monotone}, $T_0 \preceq T_1$.
        \item Case $k \geq 3$: The operation is to replace $T'_0$ with $T'_1$, where $T'_0$ is the subtree of $v$ in $T_0$, and $T'_1$ is the subtree of $w_{k-1}$ in $T_1$. Let $\ell^* = \max(\ell_{x_1}, \dots, \ell_{x_k})$, and $c = (\text{number of $j$'s that $\ell_{x_j} = \ell^*$})$. Then, in $T_0$, the probability that $\ell_v = \ell^*$ is $2^{-(c-1)}$, and otherwise $\ell_v = \ell^*+1$. However, in $T_1$, to become $\ell_{w_{k-1}} = \ell^*$, at least $c-1$ events that ``a vertex is leveled $\ell^*$ from two level-$\ell^*$ children'' must happen inside $w_1, \dots, w_{k-1}$. The scenario happens with probability $2^{-(c-1)}$ or less. Otherwise, $\ell_{w_{k-1}}$ becomes $\ell^*+1$ or more. Therefore, $T'_0 \preceq T'_1$. By \autoref{lem:monotone}, $T_0 \preceq T_1$.
    \end{itemize}
    
    We repeatedly perform the operations, starting from $T_0$, until the tree becomes a binary tree. Note that the number of leaves remains unchanged by the operations. Let $T$ be the resulting tree. By the transitivity of $(\mathcal{T}, \preceq)$, $T_0 \preceq T$ holds, which means that $\mathbb{E}[\ell_{\mathrm{root}(T_0)}] \leq \mathbb{E}[\ell_{\mathrm{root}(T)}]$. By the maximality of $T_0$, the binary tree $T$ is another tree that $\mathbb{E}[\ell_{\mathrm{root}(T)}]$ takes the maximum value among all trees with $m$ leaves.
\end{proof}

Now, it is crucial to estimate the maximum value of $\mathbb{E}[\ell_{\mathrm{root}(T)}]$ among all binary trees $T$ with at most $m$ leaves; let this value be $f(m)$. Then, we know that the expected performance of \textsc{RandomizedLST} is upper-bounded by $2f(n)$ colors because $n \geq m$.

\subsection{Analysis 1: Considering the Level-2 Terminals}

We start analyzing the worst-case performance of \textsc{RandomizedLST}. Instead of directly analyzing the maximum value of $\mathbb{E}[\ell_{\mathrm{root}(T)}]$, we examine how quickly the level-$2$ vertices appear in the worst case.

\begin{definition}
    For $i \geq 1$, a vertex is a \emph{level-$i$ terminal} if it has level $i$ and both of its two children have level $i-1$ (or is a leaf if $i = 1$).
\end{definition}

\begin{figure}[htbp]
    \centering
    \includegraphics[width=0.7\linewidth]{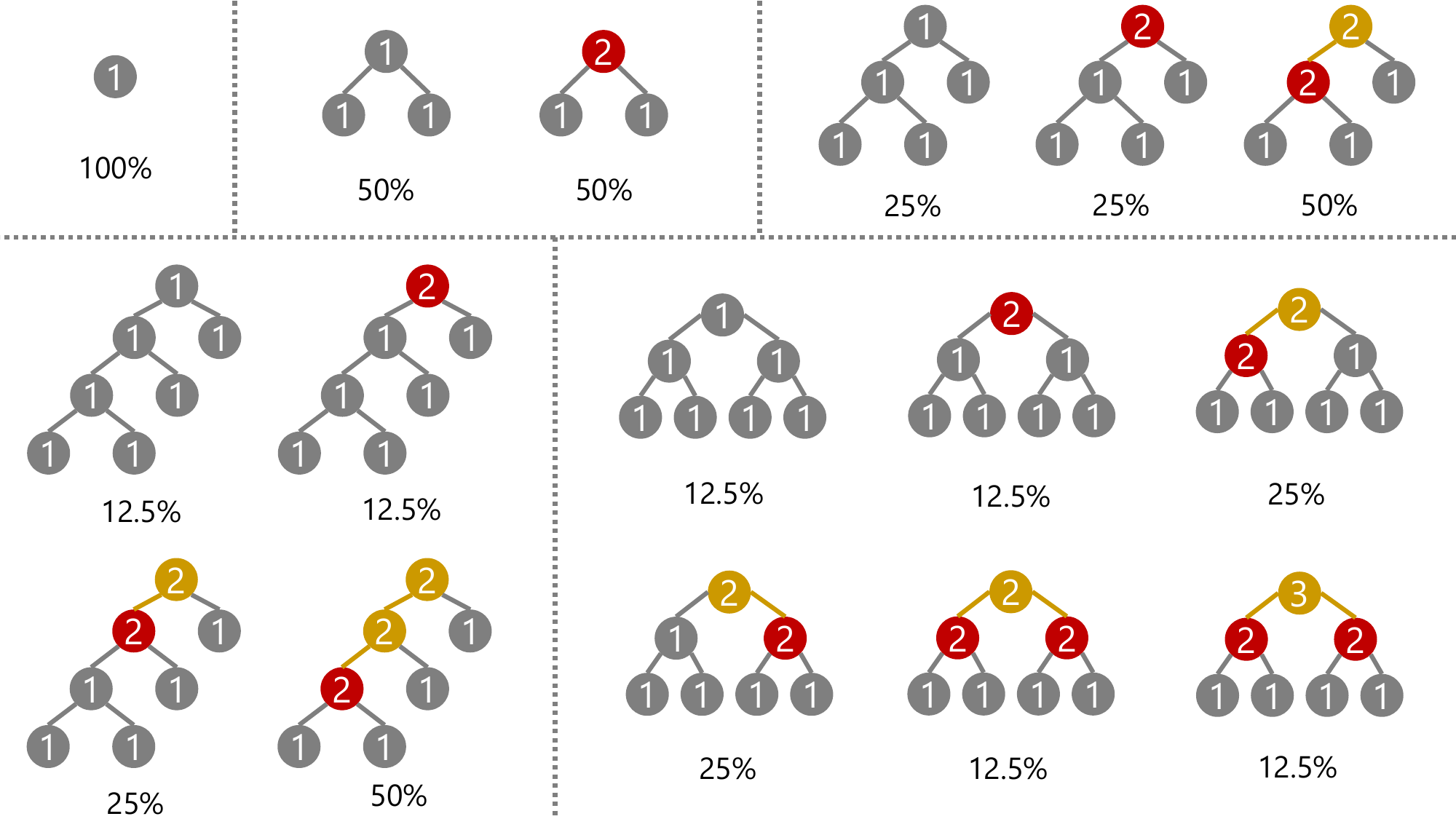}
    \caption{The outcomes of levels for all possible trees with four or fewer leaves (percentage = probability of the corresponding outcome, red vertices = level-$2$ terminals, yellow vertices = other level $2$+ vertices). There are two binary trees with $m = 4$ leaves, and the expected number of level-$2$ terminals are $\frac{7}{8}$ and $\frac{9}{8}$ for the bottom-left and the bottom-right cases, respectively.}
    \label{fig:level2terminal}
\end{figure}

Let $a_m$ be the maximum possible expected number of level-$2$ terminals for a tree with exactly $m$ leaves. For example, $a_1 = 0, a_2 = \frac{1}{2}, a_3 = \frac{3}{4}$, and $a_4 = \frac{9}{8}$, as in \autoref{fig:level2terminal}. The following lemma shows that these values also navigate the number of level-$i$ terminals for $i \geq 3$.

\begin{lemma}
    Let $\gamma$ be a real number that satisfies $a_{m'} \leq \gamma \cdot m'$ for all $m' \geq 1$. Then, the expected number of level-$i$ terminals for any tree $T$ with $m$ leaves is at most $\gamma^{i-1} \cdot m$.
    \label{lem:level2a}
\end{lemma}

\begin{proof}
    Let $X_i$ be the number of level-$i$ terminals. By the assumption, $\mathbb{E}[X_1] = m$. Let $T_i$ be a subgraph induced by the set of vertices with level $i$ or higher.\footnote{In \autoref{fig:level2terminal}, $T_2$ corresponds to the subgraph with red/yellow vertices and yellow edges.}
    
    We show a relation between the number of leaves in $T_i$ (which is equal to $X_i$) and the number of level-$(i+1)$ terminals (which is equal to $X_{i+1}$). For a vertex $v$ that has one child in $T_i$, the level of $v$ is the same as the level of the only child because the level of ``another child'' in $T$ is $i-1$ or less; therefore, contracting such a vertex $v$ does not affect the analysis. After repeating the contraction, $T_i$ becomes a binary tree with $X_i$ leaves, and the way that levels are assigned in this binary tree is identical to that of the normal binary tree case, except that levels start from $i$. Therefore, the expected number of level-$(i+1)$ terminals is at most $a_{X_i}$. It follows that $\mathbb{E}[X_{i+1}] \leq \mathbb{E}[a_{X_i}] \leq \gamma \cdot \mathbb{E}[X_i]$. By induction, $\mathbb{E}[X_i] \leq \gamma^{i-1} \cdot m$ is shown.
\end{proof}

By this lemma, we can upper-bound the expected level of the root vertex using $\gamma$.

\begin{lemma}
    $\mathbb{E}[\ell_{\mathrm{root}(T)}] \leq \frac{1}{\log (1/\gamma)} \cdot \log m + O(1)$.
    \label{lem:level2b}
\end{lemma}

\begin{proof}
    If the level of the root is $i$ or higher, there exists at least one level-$i$ terminal, so $\mathrm{Pr}[\ell_{\mathrm{root}(T)} \geq i] \leq \min(\mathbb{E}[X_i], 1) \leq \min(\gamma^{i-1} \cdot m, 1)$ (the last inequality is from \autoref{lem:level2a}). Therefore:
    \begin{align*}
        \mathbb{E}[\ell_{\mathrm{root}(T)}] & = \sum_{i=1}^\infty \mathrm{Pr}[\ell_{\mathrm{root}(T)} \geq i] \\
        & \leq \sum_{i=1}^\infty \min(\gamma^{i-1} \cdot m, 1) \\
        & \leq \frac{1}{\log (1/\gamma)} \cdot \log m + \left(1 + \frac{1}{1-\gamma}\right)
    \end{align*}
    which shows that $\mathbb{E}[\ell_{\mathrm{root}(T)}] \leq \frac{1}{\log (1/\gamma)} \cdot \log m + O(1)$.
\end{proof}

The remaining work for this subsection is to estimate the value of $\gamma$.

\begin{theorem} \label{thm:level2c}
    The minimum possible value of $\gamma$ is given by the following:
    \begin{equation*}
        \gamma = \sum_{i=1}^\infty 2^{-(2^i-1+i)} = \frac{1}{2^2} + \frac{1}{2^5} + \frac{1}{2^{10}} + \frac{1}{2^{19}} + \cdots < 0.282229
    \end{equation*}
\end{theorem}

\begin{proof}
    For a binary tree $T$, let $a(T)$ be the expected number of level-$2$ terminals. This can be calculated by:
    \begin{equation*}
        a(T) = \sum_{v \in V(T)} 2^{-(s_v-1)}
    \end{equation*}
    where $s_v$ is the number of leaves in the subtree rooted at vertex $v$. This is because the probability for $v$ to become a level-$2$ terminal is $2^{-(s_v-1)}$; all non-leaf vertices of the subtree except $v$ (there are $s_v-2$ such vertices) should be leveled $1$ from two level-$1$ children, and $v$ should be leveled $2$ from two level-$1$ children.

    Let $T$ be a tree that maximizes $a(T)$ among all trees with $m = 2^k$ leaves $(k \geq 1)$. We say that two leaves are \emph{paired} if they share the parents. Suppose that there exist at least two \emph{unpaired} leaves, and let $u_1$ and $u_2$ be two of them (since $2^k$ is an even number, the number of unpaired leaves is always even.) Then, we perform the following operations on $T$:
    \begin{enumerate}
        \item Delete $u_1$ and contract the parent $p$ of $u_1$. This operation decreases $a(T)$ by at most $2^{-(3-1)} = \frac{1}{4}$ because this deletes $p$, and since $u_1$ is unpaired, $s_p \geq 3$ holds. (\autoref{fig-523} left to middle)
        \item Create two new vertices and set their parents to $u_2$. This operation increases $a(T)$ by at least $\frac{1}{2} - (\frac{1}{8} + \frac{1}{16} + \frac{1}{32} + \cdots + \frac{1}{2^m}) = \frac{1}{4} + \frac{1}{2^m}$ because by this operation, $s_{u_2}$ becomes $2$, and $s_v$ for all ancestors of $u_2$ increases by $1$. Note that, since $u_2$ is unpaired, $s_v \geq 3$ holds for the parent of $u_2$. (\autoref{fig-523} middle to right)
    \end{enumerate}
    Overall, $a(T)$ increases by at least $\frac{1}{2^m}$, contradicting the maximality of $a(T)$. Therefore, every leaf is paired to make a subtree with two leaves.

    \begin{figure}[htbp]
        \centering
        \includegraphics[width=\linewidth]{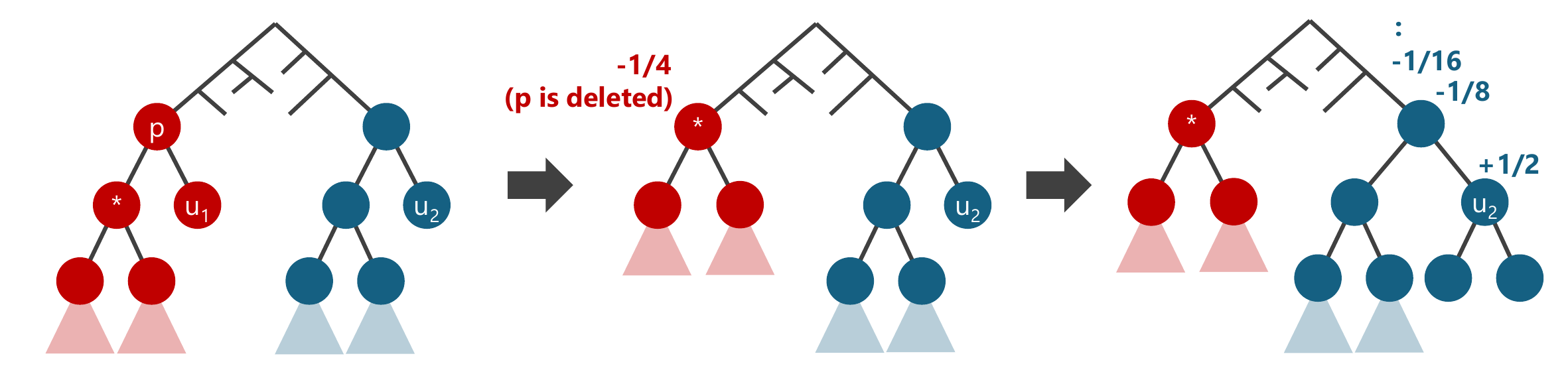}
        \caption{The sketch of operations.}
        \label{fig-523}
    \end{figure}

    Using the same argument, we can show inductively that every subtree with $2^i$ leaves must be paired, for $i = 1, 2, \dots, k-1$. Therefore, $T$ is a complete binary tree; a binary tree that all $2^k$ leaves are at depth $k$. Calculating $a(T)$ for the complete binary tree gives:
    \begin{equation*}
        a_{2^k} = \sum_{i=1}^k 2^{-(2^i-1)} \cdot 2^{k-i} = 2^k \cdot \sum_{i=1}^k 2^{-(2^i-1+i)}
    \end{equation*}
    because there are $2^{k-i}$ vertices such that $s_v = 2^i$. It follows that $a_{2^k} \leq \gamma \cdot 2^k$ for $\gamma$ given in the statement, and the given $\gamma$ is the minimum possible value.

    Suppose $a_m > \gamma \cdot m$ for some $m$. Then, for a large enough $k$, $a_m \cdot \lfloor \frac{2^k}{m} \rfloor > \gamma \cdot 2^k$ holds. For such $k$, we can make a binary tree $T$ with $2^k$ leaves, which contains $\lfloor \frac{2^k}{m} \rfloor$ subtrees with $m$ leaves. This tree can satisfy $a(T) > \gamma \cdot 2^k$, which is a contradiction to $a_{2^k} \leq \gamma \cdot 2^k$. Therefore, $a_m \leq \gamma \cdot m$ for all $m$.
\end{proof}

By assigning this $\gamma$ to \autoref{lem:level2b}, we make partial progress to the main problem of this section.

\begin{corollary}
    For any bipartite graph $G$, \textsc{RandomizedLST} uses at most $1.096 \log n + O(1)$ colors in expectation.
    \label{col:upperbound_partial}
\end{corollary}

\subsection{Analysis 2: Increasing the Layers}

In this subsection, we further improve the upper bound. To this end, we analyze how quickly the level-($L+1$) vertices appear in the worst case for a constant $L \geq 2$. Similar to the $L = 1$ case, we consider the maximum expected number of level-$(L+1)$ terminals for a tree with $m$ leaves, $a^{(L)}_m$. Let $\gamma^{(L)}$ be the minimum $\gamma$ that $a^{(L)}_m \leq \gamma \cdot m$ for every $m \geq 1$; for example, we have shown $\gamma^{(1)} \approx 0.282228$ (see \autoref{thm:level2c}). We can show the following generalization of \autoref{lem:level2b}.

\begin{lemma}
    $\mathbb{E}[\ell_{\mathrm{root}}(T)] \leq \frac{L}{\log_2(1/\gamma^{(L)})} \cdot \log m + O(1)$.
    \label{lem:levelk}
\end{lemma}

\begin{proof}
    We can prove this fact similarly to the proofs of \autoref{lem:level2a} and \autoref{lem:level2b}.
\end{proof}

Estimating the value of $\gamma^{(L)}$ is difficult for $L \geq 2$, so we try to obtain a good upper bound for $\gamma^{(L)}$ by computer check. Obviously, simple brute force is impossible because trees can be infinitely large. Instead, our idea is to sum up the performances of small enough subtrees. The following lemma demonstrates that the results of small trees can be used to upper-bound the value of $\gamma^{(L)}$:

\begin{lemma}
    Let $b^{(L)}_m$ be the maximum value of (expected number of level-$(L+1)$ terminals) + $\mathrm{Pr}[\ell_{\mathrm{root}(T)} \leq L]$ among all binary trees $T$ with $m$ leaves. Then,
    
    $$
    \gamma^{(L)} \leq \max\left(\frac{b^{(L)}_B}{B}, \dots, \frac{b^{(L)}_{2B-1}}{2B-1}\right)
    $$
    
    \noindent
    holds for any $B \geq 1$.
    \label{lem:small_upperbound}
\end{lemma}

\begin{proof}
    Consider any binary tree $T$. For explanation, we color vertex $v$ red if $s_v \geq B$ and $s_x < B$ for all the children $x$ (of $v$). Note that as in the proof of \autoref{thm:level2c}, $s_v$ is the number of leaves in the subtree rooted at vertex $v$. For each red-colored vertex $v$, let $U_v = (\text{the set of vertices in the subtree of $v$}) \cup (\text{the set of ancestors of $v$})$. We claim that the expected number of level-$(L+1)$ terminals in $U_v$ is bounded by $b^{(L)}_{s_v}$. Since there can be at most one level-$(L+1)$ terminals in the ancestors of $v$ (only when $\ell_v \leq L$), the claim follows. Since $B \leq s_v \leq 2B-1$ and the subtrees of red vertices are disjoint, the expected number of level-$(L+1)$ terminals in $\bigcup_{v: \mathrm{red}} U_v$ is at most:
    
    $$
    \max\left(\frac{b^{(L)}_B}{B}, \dots, \frac{b^{(L)}_{2B-1}}{2B-1}\right) \cdot \sum_{v: \mathrm{red}} s_v.
    $$

    \noindent
    Some vertices are in $V(T) \setminus \bigcup_{v: \mathrm{red}} U_v$, but they are composed of subtrees with $B-1$ or fewer leaves. Thus, the expected number of level-$(L+1)$ terminals in these vertices is at most:
    
    $$
    \max\left(\frac{a^{(L)}_1}{1}, \dots, \frac{a^{(L)}_{B-1}}{B-1}\right) \cdot \left(m - \sum_{v: \mathrm{red}} s_v\right).
    $$

    \noindent
    Overall, the expected number of level-$(L+1)$ terminals in $T$ can be upper-bounded by
    
    $$
    \max\left(\frac{a^{(L)}_1}{1}, \dots, \frac{a^{(L)}_{B-1}}{B-1}, \frac{b^{(L)}_B}{B}, \dots, \frac{b^{(L)}_{2B-1}}{2B-1}\right) \cdot m,
    $$

    \noindent
    and since $a^{(L)}_m/m \leq \gamma^{(L)}$, the statement of the lemma holds.
\end{proof}

So, how do we calculate $b^{(L)}_m$? Brute-forcing all binary trees to calculate $b^{(L)}_m$ is realistic only for $m \leq 40$, even using a computer check. Instead, we attempt to obtain a good upper bound for $b^{(L)}_m$. First, we use the following lemma as a tool.

\begin{lemma}
    Let $p_{m, t}$ be the maximum value of $\mathrm{Pr}[\ell_{\mathrm{root}(T)} \geq t]$ among all binary trees $T$ with $m$ leaves. Define $p'_{m, t} \ (m \geq 1, t \geq 1)$ by the recurrence relation that $p'_{m, 1} = 1$ and
    \begin{equation*}
        p'_{m, t} = \max_{m_l + m_r = m} \left\{1 - (1-p'_{m_l,t})(1-p'_{m_r,t}) + \frac{1}{2}(p'_{m_l,t-1}-p'_{m_l,t})(p'_{m_r,t-1}-p'_{m_r,t})\right\}
    \end{equation*}
    for $m \geq 2$. Then, $p_{m, t} \leq p'_{m, t}$ holds for all $m, t$.
\end{lemma}

\begin{proof}
    Let $v_l, v_r$ be the children of $\mathrm{root}(T)$, and let $m_l, m_r$ the number of leaves in subtree of $v_l$ and $v_r$, respectively. Here, the following equation holds:
    \begin{equation*}
        \mathrm{Pr}[\ell_{\mathrm{root}(T)} \geq t] = (1 - \mathrm{Pr}[\ell_{v_l} < t] \cdot \mathrm{Pr}[\ell_{v_r} < t]) + \frac{1}{2} \cdot \mathrm{Pr}[\ell_{v_l} = t-1] \cdot \mathrm{Pr}[\ell_{v_r} = t-1]
    \end{equation*}
    This is because, $\ell_{\mathrm{root}(T)} \geq t$ if exactly one of the following happens:
    \begin{itemize}
        \item Either $\ell_{v_l} \geq t$ or $\ell_{v_r} \geq t$ (probability $1 - \mathrm{Pr}[\ell_{v_l} < t] \cdot \mathrm{Pr}[\ell_{v_r} < t]$)
        \item $\ell_{v_l} = \ell_{v_r} = t-1$ and $\ell_{\mathrm{root}(T)} = t$ (probability $\frac{1}{2} \cdot \mathrm{Pr}[\ell_{v_l} = t-1] \cdot \mathrm{Pr}[\ell_{v_r} = t-1]$)
    \end{itemize}
    By \autoref{lem:monotone}, $\mathrm{Pr}[\ell_{\mathrm{root}(T)} \geq t]$ is a (non-strictly) increasing function with respect to $\mathrm{Pr}[\ell_{v_l} \geq i]$ and $\mathrm{Pr}[\ell_{v_r} \geq i]$ for each $i \in \mathbb{N}$. Since $\mathrm{Pr}[\ell_v < t] = 1 - \mathrm{Pr}[\ell_v \geq t]$ and $\mathrm{Pr}[\ell_v = t] = \mathrm{Pr}[\ell_v \geq t-1] - \mathrm{Pr}[\ell_v \geq t]$ hold for any vertex $v$, the following holds:
    \begin{equation*}
        \mathrm{Pr}[\ell_{\mathrm{root}(T)} \geq t] \leq \left\{1 - (1-p'_{m_l,t})(1-p'_{m_r,t})\right\} + \frac{1}{2}(p'_{m_l,t-1}-p'_{m_l,t})(p'_{m_r,t-1}-p'_{m_r,t})
    \end{equation*}
    which gives an upper bound of $p_{m, t}$. We note that $p'_{m, t}$ can be calculated in $O(m^2 t)$ time.
\end{proof}

Next, we show an alternative way to count the expected number of level-$k$ terminals. Let $q_v = \mathrm{Pr}[\ell_v \geq L+1]$. Then, the probability that $v$ is a level-$(L+1)$ terminal is the following, where $\mathrm{left}(v), \mathrm{right}(v)$ are children of $v$:
\begin{equation*}
    q_v - (1 - (1 - q_{\mathrm{left}(v)})(1 - q_{\mathrm{right}(v)}))
\end{equation*}
This is because $v$ is a level-$k$ terminal when $v$ has level $L+1$ or more (probability $q_v$) but not ``either $\mathrm{left}(v)$ or $v_r$ have level $L+1$ or more (probability $1 - (1 - q_{\mathrm{left}(v)})(1 - q_{\mathrm{right}(v)})$)''. The number of level-$(L+1)$ terminals in $T$ is the sum of this value for all $v \in V(T)$.

For $b^{(L)}_m$, we considered (expected number of level-$(L+1)$ terminals) + $\mathrm{Pr}[\ell_{\mathrm{root}(T)} \leq L]$. The value of $\mathrm{Pr}[\ell_{\mathrm{root}(T)} \leq L] \ (= 1 - q_{\mathrm{root}(T)})$ is equal to the sum of
\begin{equation*}
    -q_v + q_{\mathrm{left}(v)} + q_{\mathrm{right}(v)}
\end{equation*}
across all non-leaf vertex $v$, plus $1$. This is because $q_v = 0$ when $v$ is a leaf, and $q_v$ for all non-root, non-leaf vertices are canceled out, only $-q_{\mathrm{root}(T)}$ to remain. Therefore:
\begin{equation*}
    \text{(expected number of level-$(L+1)$ terminals)} + \mathrm{Pr}[\ell_{\mathrm{root}(T)} \leq L] = 1 + \sum_{v: \text{non-leaf}} q_{\mathrm{left}(v)} q_{\mathrm{right}(v)}
\end{equation*}
because $(q_v - (1 - (1 - q_{\mathrm{left}(v)})(1 - q_{\mathrm{right}(v)}))) + (-q_v + q_{\mathrm{left}(v)} + q_{\mathrm{right}(v)}) = q_{\mathrm{left}(v)} q_{\mathrm{right}(v)}$. So, we must estimate the maximum value of $1 + \sum_{v: \text{non-leaf}} q_{\mathrm{left}(v)} q_{\mathrm{right}(v)}$.

Here, $q_v$ can be between $0$ and $p'_{s_v, L+1}$, where $s_v$ is the number of leaves in the subtree of $v$. We consider the relaxed problem to maximize $1 + \sum_{v: \text{non-leaf}} q_{\mathrm{left}(v)} q_{\mathrm{right}(v)}$ under $0 \leq q_v \leq p'_{s_v, L+1}$ for all $v \in V(T)$. Let $b'^{(L)}_m$ be the maximum value for the relaxed problem among all tree $T$.

If the tree $T$ is fixed, it is obvious that $q_v = p'_{s_v, L+1}$ is the optimal solution. Using this fact, the dynamic programming formula for the optimal value $b'^{(L)}_m$ can be easily found:

\begin{lemma}
    $b'^{(L)}_m$ satisfies the following recurrence relation:
    \begin{equation*}
        b'^{(L)}_1 = 1, \quad b'^{(L)}_m = 1 + \max_{m_l + m_r = m} ((b'^{(L)}_{m_l} - 1) + (b'^{(L)}_{m_r} - 1) + p'_{m_l, L+1} \cdot p'_{m_r, L+1})
    \end{equation*}
\end{lemma}

We note that $b'^{(L)}_m$ can be calculated in $O(m^2)$ time. The calculated $b'^{(L)}_m$ is an upper bound for $b^{(L)}_m$. Hence, assigning $b'^{(L)}_m$ to $b^{(L)}_m$ in \autoref{lem:small_upperbound}, we give an upper bound (or ``upper bound of upper bound'') of $\gamma^{(L)}$, leading to the better analysis for the number of colors in $\textsc{RandomizedLST}$.

\paragraph{Tackling numerical errors.}
In order to upper bound $\gamma^{(L)}$, we must calculate $p'_{m, t}$ and $b'^{(L)}_m$. Conventionally, these values are represented by floating-point numbers, but it may cause numerical errors and create a hole in the proof. It is ideal to calculate everything with integers. So, we \emph{round up} each calculation of $p'_{m, t}$ and $b'^{(L)}_m$ to a rational number of the form $\frac{n}{D} \ (n \in \mathbb{Z})$, where $D$ is a fixed integer parameter. When we run the dynamic programming program, the calculated $p'_{m, t}$ and $b'^{(L)}_m$ will not be lower than the actual $p'_{m, t}$ and $b'^{(L)}_m$. Therefore, we can calculate an upper bound of $\gamma^{(L)}$, and becomes more precise when $D$ is larger.

\paragraph{Results.}
We computed $b'^{(L)}_m$ for $L = 1, \dots, 10$ and $m \leq 2 \cdot 2^{22} + 1$ with $D = 2^{30}$. We note that it is reasonable to set $B = 2^k + 1$ for some integer $k$, because $b'^{(L)}_m/m$ tends to be especially large when $m$ is a power of two. The resulting upper bounds on the expected number of colors in \textsc{RandomizedLST}, is shown in \autoref{tab:upperbound}.

\begin{table}[htbp]
    \caption{Computed $\gamma' = \max(b'^{(L)}_B / B, \dots, b'^{(L)}_{2B-1} / (2B-1))$ and the corresponding upper bounds on the number of colors in \textsc{RandomizedLST} (divided by $\log n$), rounded up to $6$ decimal places}
    \begin{minipage}{0.495\linewidth}
        \centering
        \begin{tabular}{|c|c|c|c|} \hline
            $L$ & $B$ & $\gamma'$ & colors \\ \hline
            $1$ & $2^4 + 1$ & \texttt{2.822285e-1} & \texttt{1.095852} \\ \hline
            $2$ & $2^6 + 1$ & \texttt{7.373281e-2} & \texttt{1.063392} \\ \hline
            $3$ & $2^8 + 1$ & \texttt{1.912694e-2} & \texttt{1.051111} \\ \hline
            $4$ & $2^{10} + 1$ & \texttt{4.957865e-3} & \texttt{1.044924} \\ \hline
            $5$ & $2^{12} + 1$ & \texttt{1.284998e-3} & \texttt{1.041231} \\ \hline
        \end{tabular}
    \end{minipage}
    \begin{minipage}{0.495\linewidth}
        \begin{tabular}{|c|c|c|c|} \hline
            $L$ & $B$ & $\gamma'$ & colors \\ \hline
            $6$ & $2^{14} + 1$ & \texttt{3.330478e-4} & \texttt{1.038783} \\ \hline
            $7$ & $2^{16} + 1$ & \texttt{8.632023e-5} & \texttt{1.037042} \\ \hline
            $8$ & $2^{18} + 1$ & \texttt{2.237284e-5} & \texttt{1.035741} \\ \hline
            $9$ & $2^{20} + 1$ & \texttt{5.798723e-6} & \texttt{1.034731} \\ \hline
            $10$ & $2^{22} + 1$ & \texttt{1.502954e-6} & \texttt{1.033925} \\ \hline
        \end{tabular}
    \end{minipage}
    \label{tab:upperbound}
\end{table}

Now, we obtain $\gamma^{(10)} \leq 1.502954 \times 10^{-6}$. Therefore, by \autoref{lem:levelk}, the following theorem holds:

\begin{theorem}
    For any bipartite graph $G$, \textsc{RandomizedLST} uses at most $1.034 \log n + O(1)$ colors in expectation.
    \label{thm:upperbound}
\end{theorem}

\paragraph{Performance of \textsc{RandomizedLST}.} We estimate $\gamma^{(L)}$ for $L \leq 10$, but if we can further increase $L$, it would improve the upper bound on the performance of \textsc{RandomizedLST}. It seems to us that the worst-case input graph for \textsc{RandomizedLST} is a complete binary tree. In this case, the experiment with dynamic programming shows that it uses around $1.027 \log n$ colors. Therefore, we conjecture that \textsc{RandomizedLST} uses at most $1.027 \log n + O(1)$ colors for any bipartite graph $G$ (\autoref{conj:ai-solve-3}).

\paragraph{Computer check and programs.} The programs used for the analysis and their results in this section can be downloaded at \url{https://github.com/square1001/online-bipartite-coloring}.

\section{Lower Bound for Bipartite Graphs}\label{sec:k2-lower}

\subsection{The Lower Bound Instance}

In this section, we show the limit of randomized algorithms (against an oblivious adversary). We prove the following theorem.

\begin{theorem}\label{thm:lowerbound1}
    Any randomized online coloring algorithm for bipartite graphs requires at least $\frac{91}{96} \log n - O(1)$ colors in expectation in the worst case.
\end{theorem}

By Yao's lemma \cite{Yao77}, the goal is to give a distribution of bipartite graphs that any \emph{deterministic} algorithm uses at least $\frac{91}{96} \log n - O(1)$ colors in expectation. We construct such input graphs as in \autoref{fig:instance}, having structures similar to a binary tree. The grade-$h$ instance, which is constructed by merging two disjoint grade-$(h-1)$ instances with two extra vertices, contains $4 \cdot 2^h - 2$ vertices.

\begin{figure}[htbp]
    \centering
    \includegraphics[width=0.8\linewidth]{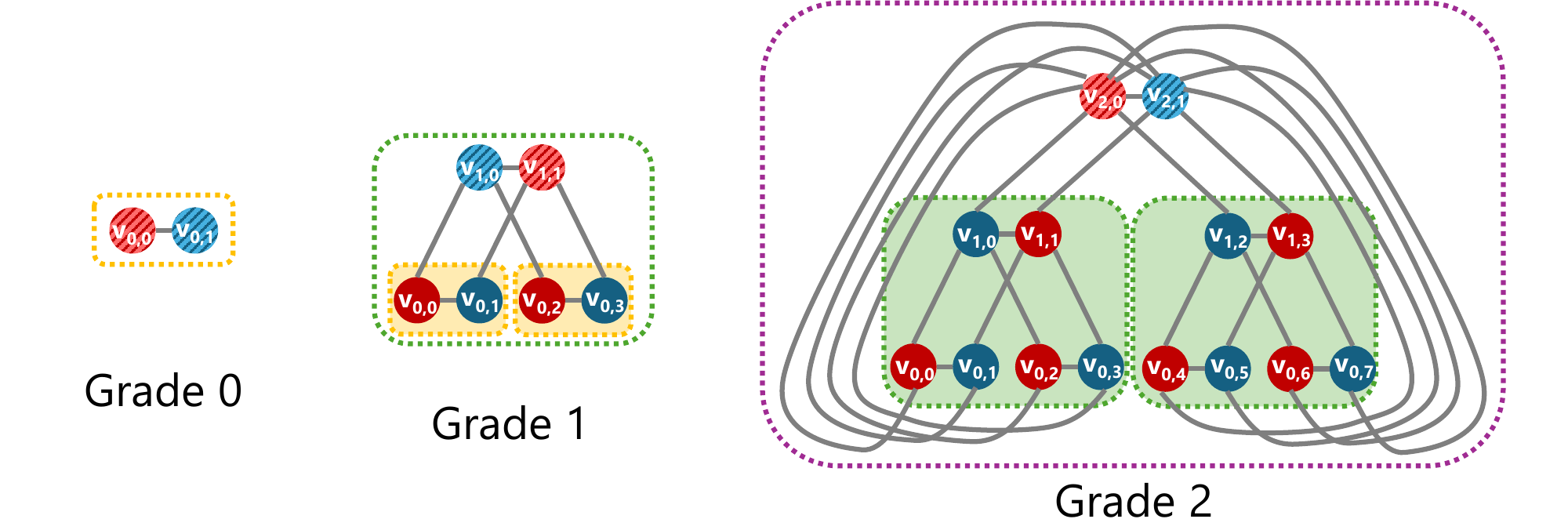}
    \caption{The instances to give a lower bound for $h = 0, 1, 2$. The orange, green, and purple regions correspond to grade-$0$, grade-$1$, and grade-$2$ graphs, respectively. The two extra vertices are in stripe. The labels of vertices can  be changed, depending on the random choice.}
    \label{fig:instance}
\end{figure}

We formally explain how to construct the grade-$k$ instance. The vertices are $v_{i, j} \ (0 \leq i \leq h, 0 \leq j < 2^{h-i+1})$. We refer to the phase that $v_{i, 0}, \dots, v_{i, 2^{h-i+1}-1}$ arrive as ``phase $i$''. At phase $i$, there are $2^{h-i+1}$ components of grade-$(i-1)$ graphs. For $j = 0, 2, \dots, 2^{h-i+1}-2$, we randomly select two of the remaining grade-$(i-1)$ components, and ``merge'' them with two new vertices $v_{i, j}$ and $v_{i, j+1}$. Then, after phase $i$, there are $2^{h-i}$ components of grade-$i$ graphs. 

Formally, when we merge two connected components (say $C_1$ and $C_2$) with two new vertices (say $v_a$ and $v_b$), for each vertex in $C_1$ and $C_2$, we add an edge between it and \emph{either} $v_a$ or $v_b$ in a way that the resulting graph remains bipartite, and we also add an edge between $v_a$ and $v_b$. There are four possible resulting graphs because we can choose which side of bipartition of $C_i$ will be linked to $v_a$ (and to $v_b$), independently for $i = 1, 2$. The four choices will be selected with probability $\frac{1}{4}$ each. We denote this procedure as $\textsc{Merge}(C_1, C_2, v_a, v_b)$, which returns the resulting component.

The goal is to prove the following result, which implies \autoref{thm:lowerbound1}.

\begin{theorem}
    Any deterministic online coloring algorithm requires at least expected $\frac{91}{96} h - O(1)$ colors for the grade-$h$ instance.
    \label{thm:lowerbound2}
\end{theorem}

\subsection{Introducing Potential}

A classic idea to lower-bound the performance of algorithms is to define a value called the \emph{potential} for the current state and say that the potential always increases by a certain amount in each operation.

We define the potential for each connected component $C$. Define the state of $C$ to be $(X, Y)$, where $X$ and $Y$ are the sets of colors used in each bipartition of $C$. Later, we may also use $(X, Y)$ to represent $C$ itself. We consider the following potential $\phi_1$:

\begin{equation*}
    \phi_1(C) := \frac{1}{2} (|X| + |Y|)
\end{equation*}

We will show that, for each phase, the average potential of the components will increase by at least $\frac{3}{4}$ in expectation. Since $|X \cup Y| \geq \phi_1(C)$ always holds, this proves that the expected number of colors used for the grade-$h$ instance is at least $\frac{3}{4} h + 1$. To this end, we first prove that $X$ and $Y$ are not in the inclusion relation.

\begin{lemma}\label{lem:potential_lemma2}
    Let $C_1$ and $C_2$ be two connected components, and let $C := \textsc{Merge}(C_1, C_2, v_a, v_b)$. Then, the state $(X, Y)$ of $C$ neither satisfies $X \subseteq Y$ nor $Y \subseteq X$.
\end{lemma}

\begin{proof}
    To this end, we will show that $X \subseteq Y$ is impossible. Let $(X_0, Y_0)$ be the state of $C$ \emph{before} coloring $v_a$ and $v_b$ (they are ``uncolored'' at this moment). Let $c_a$ and $c_b$ be the colors used for $v_a$ and $v_b$, respectively. Without loss of generality, $v_a$ is adjacent to vertices with a color in $Y_0$. Thus, $c_a \notin Y_0$. Since $c_a \neq c_b$, $c_a \notin Y_0 \cup \{c_b\} = Y$. Therefore, $X = X_0 \cup \{c_a\} \not\subseteq Y$. By symmetry, $Y \subseteq X$ is also impossible.
\end{proof}

This lemma implies that, for the grade-$h$ instance, the state $(X, Y)$ of any component at any time neither satisfies $X \subseteq Y$ nor $Y \subseteq X$. Now, we start proving the $\frac{3}{4}$ lower bound. It suffices to prove the following lemma.

\begin{lemma}\label{lem:potential_lemma1}
    Let $C_1$ and $C_2$ be two connected components, and let $C := \textsc{Merge}(C_1, C_2, v_a, v_b)$. Then, $\mathbb{E}[\phi_1(C)] \geq \frac{1}{2} (\phi_1(C_1) + \phi_1(C_2)) + \frac{3}{4}$ when the state $(X, Y)$ for each $C_i$ neither satisfies $X \subseteq Y$ nor $Y \subseteq X$.
\end{lemma}

\begin{proof}
    Let $(X_1, Y_1)$ and $(X_2, Y_2)$ be the states of $C_1$ and $C_2$, respectively. Let $(X_0, Y_0)$ and $(X, Y)$ be the states of $C$ before and after coloring $\{v_a, v_b\}$, respectively. Here, $(X_0, Y_0)$ can be any of $(X_1 \cup X_2, Y_1 \cup Y_2), (X_1 \cup Y_2, Y_1 \cup X_2), (Y_1 \cup Y_2, X_1 \cup X_2), (Y_1 \cup X_2, X_1 \cup Y_2)$, with probability $\frac{1}{4}$ each. We consider the following value $\Delta_0$. Since $X_0 \subseteq X$ and $Y_0 \subseteq Y$, $\Delta_0$ gives a lower bound of $\mathbb{E}[\phi_1(C)] - \frac{1}{2} (\phi_1(C_1) + \phi_1(C_2))$.
    \begin{equation}\label{eq:potential_lemma1_eq}
        \Delta_0 := \mathbb{E}\left[\phi_1((X_0, Y_0))\right] - \frac{1}{2} (\phi_1(C_1) + \phi_1(C_2)) \quad \left(\phi_1((X_0, Y_0)) = \frac{1}{2} (|X_0| + |Y_0|)\right)
    \end{equation}
    
    We see how much each color $c$ contributes to $\Delta_0$, for the first and the second terms of \autoref{eq:potential_lemma1_eq}. Excluding the symmetric patterns, there are six cases to consider, shown in \autoref{tab:potential_lemma1}:
    \begin{table}[htbp]
        \centering
        \begin{tabular}{|c|c|c|c|}
            Case & 1st term & 2nd term & $\Delta_0$ \\ \hline
            $c \notin X_1, Y_1, X_2, Y_2$ & $0$ & $0$ & $0$ \\
            $c \in X_1, c \notin Y_1, X_2, Y_2$ & $1/2$ & $1/4$ & $1/4$ \\
            $c \in X_1, X_2, c \notin Y_1, Y_2$ & $3/4$ & $1/2$ & $1/4$ \\
            $c \in X_1, Y_1, c \notin X_2, Y_2$ & $1$ & $1/2$ & $1/2$ \\
            $c \in X_1, Y_1, X_2, c \notin Y_2$ & $1$ & $3/4$ & $1/4$ \\
            $c \in X_1, Y_1, X_2, Y_2$ & $1$ & $1$ & $0$
        \end{tabular}
        \caption{For each case, contribution to the 1st and 2nd term of (\ref{eq:potential_lemma1_eq}), along with their difference}
        \label{tab:potential_lemma1}
    \end{table}

    By \autoref{lem:potential_lemma2}, $|X_1 \oplus Y_1| \geq 2$ and $|X_2 \oplus Y_2| \geq 2$, where $X_i \oplus Y_i = (X_i \setminus Y_i) \cup (Y_i \setminus X_i)$ is the symmetric difference of $X_i$ and $Y_i$. The 2nd, 3rd, and 5th patterns in \autoref{tab:potential_lemma1} correspond to the case that $c \in X_i \oplus Y_i$ for an $i \in \{1, 2\}$ (the 3rd pattern is double-counted); for each case, each element in $X_i \oplus Y_i \ (i = 1, 2)$ contributes to $\Delta_0$ by $\frac{1}{4}, \frac{1}{8}, \frac{1}{4}$, respectively. Therefore, $\Delta_0 \geq 4 \times \frac{1}{8} = \frac{1}{2}$, and except for the case that ``two $c$'s are in the 3rd pattern and no $c$'s are in the 2nd, 4th, and 5th patterns,'' $\Delta_0 \geq \frac{3}{4}$, which means that the expected average potential increases by $\frac{3}{4}$. Note that there is no case that $\Delta_0 = \frac{1}{4} + \frac{1}{8} + \frac{1}{8} + \frac{1}{8} = \frac{5}{8}$ because the number of ``$\frac{1}{8}$'' must be even.

    The only exceptional case is, without loss of generality, $X_1 = \{1, \dots, c, c+1\}, Y_1 = \{1, \dots, c, c+2\}, X_2 = \{1, \dots, c, c+1\}, Y_2 = \{1, \dots, c, c+2\}$ for some $c$. In this case:
    \begin{itemize}
        \item If $(X_0, Y_0) = (X_1 \cup X_2, Y_1 \cup Y_2)$, we can color $v_a, v_b$ by color $c+1, c+2$, respectively, and $(X, Y) = (\{1, \dots, c, c+1\}, \{1, \dots, c, c+2\})$. The average potential stays at $c+1$.
        \item If $(X_0, Y_0) = (X_1 \cup Y_2, Y_1 \cup X_2)$, $X_0 = Y_0 = \{1, \dots, c+2\}$, so we need two extra colors to color $v_a$ and $v_b$. The average potential increases from $c+1$ to $c+3$.
        \item The cases of $(X_0, Y_0) = (Y_1 \cup Y_2, X_1 \cup X_2), (Y_1 \cup X_2, X_1 \cup Y_2)$ are symmetric to the first and the second cases, respectively.
    \end{itemize}
    Therefore, the expected average potential increases by $1$ in this case.
\end{proof}

\subsection{Two-Phase Analysis} \label{sec:k2-lower-twophase}

In the previous subsection, we showed that any randomized algorithm requires at least $\frac{3}{4} \log n - \frac{1}{2}$ colors, using potential $\phi_1$. Unfortunately, this is the best possible potential among functions of $|X \cap Y|, |X \setminus Y|, |Y \setminus X|$. In order to improve the lower bound, we need to develop more sophisticated analysis methods.

In this subsection, we see how much the expected average potential of the components increases in two phases. We consider the model with four connected components, $C_1, C_2, C_3, C_4$. Two pairs of components are merged in the first phase, and the two ``merged'' components are merged in the second phase. There are essentially $3 \times 2 \times 2^3 = 48$ outcomes, considering how components are paired ($3$ ways) and are merged in which order ($2$ ways), along with how merges happen ($2^3 = 8$ ways, because $3$ merges happen).\footnote{When merging components with states $(X_1, Y_1)$ and $(X_2, Y_2)$, there are essentially two possible resulting states: $(X_1 \cup X_2, Y_1 \cup Y_2)$ and $(X_1 \cup Y_2, Y_1 \cup X_2)$, as we can regard $(X, Y)$ and $(Y, X)$ as the same states.} The ``player'' (algorithm) must decide the color of two added vertices right after every merge. Therefore, the minimum expected increase in potential $\phi$ when the player plays optimally, denoted by $\textsc{PotentialIncrease}(C_1, C_2, C_3, C_4, \phi)$, can be computed in the expected minimax algorithm (\autoref{alg:minimax}). The complex nature of this procedure makes it more difficult to create cases with low increases in potential.

\begin{algorithm}[htbp]
	\caption{$\textsc{PotentialIncrease}(C_1, C_2, C_3, C_4, \phi)$}
    \label{alg:minimax}
    \begin{algorithmic}[1]
        \State $x \gets 0$
        \ForAll{$(C'_1, C'_2, C'_3, C'_4)$, a permutation of $(C_1, C_2, C_3, C_4)$, that $C'_1$ and $C'_2$ are merged first and $C'_3$ and $C'_4$ are merged second (there are $3 \times 2 = 6$ ways)}
            \State Let $(X'_i, Y'_i)$ be the state of $C'_i$ for $i = 1, 2, 3, 4$
            \State $a_0 \gets 0$
            \For{$(X_5, Y_5) = (X'_1 \cup X'_2, Y'_1 \cup Y'_2), (X'_1 \cup Y'_2, Y'_1 \cup X'_2)$}
                \State $b_0 \gets +\infty$
                \ForAll{$(c^{(1)}_a, c^{(1)}_b)$, the colors of two extra vertices when merging $C'_1$ and $C'_2$}
                    \State $(X'_5, Y'_5) \gets (X_5 \cup \{c^{(1)}_a\}, Y_5 \cup \{c^{(1)}_b\})$
                    \State $a_1 \gets 0$
                    \For{$(X_6, Y_6) = (X'_3 \cup X'_4, Y'_3 \cup Y'_4), (X'_3 \cup Y'_4, Y'_3 \cup X'_4)$}
                        \State $b_1 \gets +\infty$
                        \ForAll{$(c^{(2)}_a, c^{(2)}_b)$, the colors of two extra vertices when merging $C'_3$ and $C'_4$}
                            \State $(X'_6, Y'_6) \gets (X_6 \cup \{c^{(2)}_a\}, Y_6 \cup \{c^{(2)}_b\})$
                            \State $a_2 \gets 0$
                            \For{$(X_7, Y_7) = (X'_5 \cup X'_6, Y'_5 \cup Y'_6), (X'_5 \cup Y'_6, Y'_5 \cup X'_6)$}
                                \State $b_2 \gets +\infty$
                                \ForAll{$(c^{(3)}_a, c^{(3)}_b)$, the colors of two extra vertices in the final merge}
                                    \State $(X'_7, Y'_7) \gets (X_7 \cup \{c^{(3)}_a\}, Y_7 \cup \{c^{(3)}_b\})$
                                    \State $\Delta \gets \phi((X'_7, Y'_7)) - \frac{1}{4} (\phi(C_1) + \phi(C_2) + \phi(C_3) + \phi(C_4))$
                                    \State $b_2 \gets \min(b_2, \Delta)$
                                \EndFor
                                \State $a_2 \gets a_2 + \frac{1}{2} b_2$
                            \EndFor
                            \State $b_1 \gets \min(b_1, a_2)$
                        \EndFor
                        \State $a_1 \gets a_1 + \frac{1}{2} b_1$
                    \EndFor
                    \State $b_0 \gets \min(b_0, a_1)$
                \EndFor
                \State $a_0 \gets a_0 + \frac{1}{2} b_0$
            \EndFor
            \State $x \gets x + \frac{1}{6} a_0$
        \EndFor
        \State \Return $x$
    \end{algorithmic}
\end{algorithm}

We introduce the following potential $\phi_2$, which slightly modifies $\phi_1$.\footnote{Even with $\phi_1$, we can prove the lower bound of $\frac{59}{64} \log_2 n - O(1)$, which is only slightly worse than $\frac{89}{96} \log_2 n - O(1)$.} The goal is to prove \autoref{lem:potential_lemma3}, which directly leads to a lower bound of $\frac{89}{96} h - O(1)$ colors.
\begin{equation*}
    \phi_2(C) = |X \cap Y| + \frac{11}{21} |X \oplus Y|
\end{equation*}

\begin{lemma}
    $\textsc{PotentialIncrease}(C_1, C_2, C_3, C_4, \phi_2) \geq \frac{89}{48}$ for any components $C_1, C_2, C_3, C_4$ where state $(X, Y)$ neither satisfies $X \subseteq Y$ nor $Y \subseteq X$.
    \label{lem:potential_lemma3}
\end{lemma}

\begin{proof}
    This lemma can be shown by a computer search. First, we need to make the number of candidates for the combination of states of $C_1, C_2, C_3, C_4$ computable (or at least finite). With all the ideas explained in the next subsection, we reduce the number of combinations to consider to $16829$. We compute $\textsc{PotentialIncrease}(C_1, C_2, C_3, C_4, \phi_2)$ for all of them, and each call returned a result of $\frac{89}{48}$ or more.
\end{proof}

\paragraph{Remarks on \autoref{alg:minimax}.}
When brute-forcing colors $(c^{(i)}_a, c^{(i)}_b) \ (i = 1, 2, 3)$ in \autoref{alg:minimax}, we may assume that we do not ``jump'' the colors; when there are only colors $1, \dots, c$ in the current graph, we only color the next vertex using any of color $1, \dots, c+1$. This makes the number of search states finite. Also, when $X_{i+4} \setminus Y_{i+4}$ is not empty, it is obvious that choosing $c^{(i)}_a$ from $X_{i+4} \setminus Y_{i+4}$ is an optimal strategy (then, it becomes $X'_{i+4} = X_{i+4}$). Similarly, when $Y_{i+4} \setminus X_{i+4}$ is not empty, it is optimal to choose $c^{(i)}_b$ from $Y_{i+4} \setminus X_{i+4}$. Then, we can further reduce the number of possibilities to consider.

\subsection{Further Improvement: Limiting the Number of Combinations} \label{sec:k2-lower-search}

In this subsection, we complete the proof of \autoref{lem:potential_lemma3}; indeed, we explain how to enumerate all candidates of combinations of states of $C_1, C_2, C_3, C_4$ such that $\textsc{PotentialIncrease}(C_1, C_2, C_3, C_4, \phi_2)$ is less than $\frac{89}{48}$, using a computer search.

\paragraph{The searching framework.}
The following lemma, the ``potential $\phi_2$ version'' of \autoref{lem:potential_lemma1}, shows that $\phi_2$ always increases by $\frac{31}{42}$ in the second phase.
\begin{lemma}
    Let $C_1$ and $C_2$ be connected components, and let $C := \textsc{Merge}(C_1, C_2, v_a, v_b)$. Then, $\mathbb{E}[\phi_2(C)] \geq \frac{1}{2} (\phi_2(C_1) + \phi_2(C_2)) + \frac{31}{42}$ when the state $(X, Y)$ for each $C_i$ neither satisfies $X \subseteq Y$ nor $Y \subseteq X$.
    \label{lem:potential_lemma4}
\end{lemma}

We can prove this lemma similarly to \autoref{lem:potential_lemma1}, but we can also prove by running \autoref{alg:lowerbound} (by calling $g((\emptyset, \emptyset), (\emptyset, \emptyset))$, which will be explained later). Therefore, it suffices to search all the cases where, in the first phase, the expected average potential increases by less than $\frac{89}{48} - \frac{31}{42} = \frac{125}{112}$. Let $f(C_i, C_j)$ be the expected increase of a potential $\phi_2$ by merging components $C_i$ and $C_j$. We enumerate all the cases such that:

\begin{equation*}
    \Delta_1 = \frac{1}{6} \left(f(C_1, C_2) + f(C_1, C_3) + f(C_1, C_4) + f(C_2, C_3) + f(C_2, C_4) + f(C_3, C_4)\right) < \frac{125}{112}
\end{equation*}

\paragraph{Implementation and the state matrix.}
In order to implement the brute force of the cases such that $\Delta_1 < \frac{125}{112}$, we define \emph{state matrix} $M$ to represent $(X_i, Y_i)$, the state of $C_i$, for $i = 1, 2, 3, 4$. $M$ is a $4 \times m$ matrix where $m$ is the number of colors, defined in the following way.
\begin{equation*}
    M_{i, j} = \begin{cases}
        0 & (j \notin X_i, Y_i) \\
        1 & (j \in X_i, j \notin Y_i) \\
        2 & (j \in Y_i, j \notin X_i) \\
        3 & (j \in X_i, Y_i)
    \end{cases}
\end{equation*}
Here, we can assume that each column is not ``all $0$'' or ``all $3$'' (so there are $254^m$ matrices with $m$ columns). This is because if a color is in none or all of $X_1, Y_1, \dots, X_4, Y_4$, there is no effect on the increase of potential. We attempt to brute-force these matrices by DFS (depth-first search), which appends one column to the right in each step (corresponds to adding a new color).

\paragraph{Lower-bounding for the DFS.}
In order to execute the DFS in a finite and realistic time, we need to apply ``pruning'' by showing that, for the current state matrix, it is impossible to achieve $\Delta_1 < \frac{125}{112}$ no matter how subsequent columns are appended.

Let $g(C_i, C_j)$ be the minimum value of $f(C_i, C_j)$ when we can freely add subsequent colors to $(X_i, Y_i)$ and $(X_j, Y_j)$. Then, $\mathrm{LB} := \frac{1}{6} (g(C_1, C_2) + g(C_1, C_3) + g(C_1, C_4) + g(C_2, C_3) + g(C_2, C_4) + g(C_3, C_4))$ gives the lower bound of $\Delta_1$ for any subsequent state matrices. Therefore, once $\mathrm{LB} \geq \frac{125}{112}$ is met, we do not need to perform the DFS for subsequent columns and can apply pruning. We can compute $g(C_i, C_j)$ using \autoref{alg:lowerbound}, which is given by $\textsc{LowerBound}(C_i, C_j, m, +\infty)$.

\begin{algorithm}[htbp]
    \caption{$\textsc{LowerBound}(C_1, C_2, m, \beta)$: Function to calculate the minimum value of $f(C_1, C_2)$ when colors $m+1, m+2, \dots$ are added, or report that the minimum value is $\beta$ or more}
    \label{alg:lowerbound}
    \begin{algorithmic}[1]
        \State $d_0 \gets$ (the value of $\Delta_0$ defined in \autoref{lem:potential_lemma1} (for potential $\phi_2$) for the current $C_1, C_2$)
        \If{$d_0 \geq \beta$}
            \State \Return $\beta$ \Comment{$f(C_1, C_2) \geq \beta$ no matter how subsequent colors are added}
        \EndIf
        \State $(X^*_1, Y^*_1), (X^*_2, Y^*_2) \gets$ (states of components which adds minimum possible subsequent colors into $C_1, C_2$ to make it satisfy $X_1 \not\subseteq Y_1, Y_1 \not\subseteq X_1, X_2 \not\subseteq Y_2, Y_2 \not\subseteq X_2$)
        \State $\beta \gets \min(\beta, f((X^*_1, Y^*_1), (X^*_2, Y^*_2)))$
        \ForAll{$(X'_1, Y'_1), (X'_2, Y'_2)$, the states of $C_1, C_2$ after adding color $m+1$ (there are $2^4 - 2 = 14$ ways)}
            \State $d \gets \textsc{LowerBound}((X'_1, Y'_1), (X'_2, Y'_2), m+1, \beta)$
            \State $\beta \gets \min(\beta, d)$
        \EndFor
        \State \Return $\beta$
    \end{algorithmic}
\end{algorithm}

\paragraph{Remarks on \autoref{alg:lowerbound}.}
The idea of this algorithm is that $\Delta_0$ defined in \autoref{lem:potential_lemma1} not only gives a lower bound to $f(C_1, C_2)$; no matter how subsequent colors are added to $(X_1, Y_1)$ and $(X_2, Y_2)$, $f(C_1, C_2)$ will not be lower than $\Delta_0$. To compute $g(C_1, C_2)$, we search all possibilities on subsequent colors by DFS, but once $\Delta_0$ exceeds or equals to the current minimum value of $f(C_1, C_2)$, we do not need to perform the DFS for subsequent colors and can apply pruning. Also, since most of the calls to $\textsc{LowerBound}$ are identical (up to swapping colors), we can apply memoization to reduce redundant calculations of $\textsc{LowerBound}$.

\paragraph{Utilizing symmetry.}
Next, we reduce the number of combinations by utilizing the ``symmetry'' that some state matrices represent essentially the same state. Specifically, the following operations on a state matrix $M$ do not essentially change the state:
\begin{enumerate}
    \item Swap two rows of $M$. (Corresponds to swapping $C_i$ and $C_j$)
    \item Swap two columns of $M$. (Corresponds to swapping the indices of two colors)
    \item Choose one row, and for each element of the row, change $1$ to $2$ and $2$ to $1$. (Corresponds to swapping $X_i$ and $Y_i$)
\end{enumerate}
We say that $M$ is in \emph{standard form} if $(M_{1, 1}, \dots, M_{4, 1}, \dots, M_{1, m}, \dots, M_{4, m})$ is lexicographically earliest among the state matrices that can be obtained by repeating these three kinds of operations. It is easy to see that, once $M$ becomes a matrix not in standard form, it will never be in standard form again after adding subsequent columns. In this case, we can apply pruning in the DFS.

We can check if $M$ is in standard form by brute force. We brute-force the choice of how the rows are permuted and which rows we flip $1$ and $2$ in (there are $4! \times 2^4 = 384$ ways), and for each choice, we sort columns in lexicographical order to obtain a candidate of the standard form.

\paragraph{Results.} In the DFS, we search $62195$ state matrices that satisfy $\mathrm{LB} < \frac{125}{112}$ and are in standard form. Among them, $22558$ state matrices contain $1$ and $2$ in every row, and $16829$ of them actually satisfy $\Delta_1 < \frac{125}{112}$.

\subsection{The Final Piece: Potential Decomposition}

To further improve the analysis of the lower bound, we introduce the idea that the potentials in the first and second phases can be different. We consider decomposing potential $\phi$ as $\phi = \phi_A + \phi_B$. Let $p_i, a_i, b_i \ (i = 0, \dots, h)$ be the average potentials $\phi, \phi_A, \phi_B$ of components after phase $i$. Then, the following equation holds due to the telescoping sum:
\begin{equation*}
    p_h - p_0 = (a_h - a_{h-1}) + (b_1 - b_0) + \sum_{i=1}^{h-1} \{(a_i - a_{i-1}) + (b_{i+1} - b_i)\}
\end{equation*}

The term $(a_i - a_{i-1}) + (b_{i+1} - b_i)$ represents the increase of $\phi_A$ at phase $i$ plus the increase of $\phi_B$ at phase $i+1$. Therefore, when merging $C_1, C_2, C_3, C_4$ as in \autoref{sec:k2-lower-twophase}, if we know that the expected (increase of average $\phi_A$ in the first phase) + (increase of average $\phi_B$ in the second phase) is always $x$ or more, we know that $p_k \geq xk - O(1)$. The new ``potential increase'' can be calculated by modifying Line 19 of \autoref{alg:minimax} to $\Delta \gets \{\phi_B((X'_7, Y'_7)) - \frac{1}{2}(\phi_B((X'_5, Y'_5)) + \phi_B((X'_6, Y'_6)))\} + \{\frac{1}{2}(\phi_A((X'_5, Y'_5)) + \phi_A((X'_6, Y'_6))) - \frac{1}{4}(\phi_A(C_1) + \phi_A(C_2) + \phi_A(C_3) + \phi_A(C_4))\}$. The new algorithm is denoted as $\textsc{PotentialIncrease}(C_1, C_2, C_3, C_4, \phi_A, \phi_B)$. In Subsection 6.3, we only considered the case that $\phi_A = \phi_B$, so by increasing the degrees of freedom, we can expect a better lower bound.

\paragraph{The potential setting.}
We consider setting $\phi_A, \phi_B$ in the following way:
\begin{align*}
    \phi_A(C) & = \frac{1}{2} |X \cap Y| + \begin{cases}
        \frac{17}{24} & ((|X \setminus Y|, |Y \setminus X|) = (2, 1), (1, 2)) \\
        \frac{5}{6} & ((|X \setminus Y|, |Y \setminus X|) = (3, 1), (1, 3)) \\
        \frac{1}{4} |X \oplus Y| & (\text{otherwise}) \\
    \end{cases} \\
    \phi_B(C) & = \frac{1}{2} |X \cap Y| + \frac{1}{3} |X \oplus Y|
\end{align*}
Then, the following lemma holds:
\begin{lemma}
    $\textsc{PotentialIncrease}(C_1, C_2, C_3, C_4, \phi_A, \phi_B) \geq \frac{91}{96}$ for any components $C_1, C_2, C_3, C_4$ where state $(X, Y)$ neither satisfies $X \subseteq Y$ nor $Y \subseteq X$.
    \label{lem:potential_lemma5}
\end{lemma}

We assume that in \autoref{lem:potential_lemma5}, the function $\textsc{PotentialIncrease}(C_1, C_2, C_3, C_4, \phi_A, \phi_B)$ is calculated in the way explained in ``Remarks on \autoref{alg:minimax}'' in \autoref{sec:k2-lower-twophase}. We note that the result may change when we allow to choose a new color for $c^{(i)}_a$ when $X_{i+4} \setminus Y_{i+4}$ is not empty (and similarly for $c^{(i)}_b$), even if this is not an optimal strategy, because of the difference of $\phi_A$ and $\phi_B$. Below, we give a computer assisted proof. 

\paragraph{The searching framework.}
We prove \autoref{lem:potential_lemma5} in a similar way to \autoref{sec:k2-lower-search}. In the second phase, it can be proven that $\phi_B$ always increases by $\frac{1}{3}$, similarly to \autoref{lem:potential_lemma4}. Therefore, we have to enumerate all the cases such that $\phi_A$ increases by at most $\frac{91}{96} - \frac{1}{3} = \frac{59}{96}$ in the first phase.

\paragraph{The lower-bounding for $\phi_A$.}
We aim to calculate $g(C_1, C_2)$, defined in \autoref{sec:k2-lower-search}. However, the difference is that $\phi_A$ is no longer a linear function of $|X \oplus Y|$, so each $c \in X_i \oplus Y_i$ does not directly contribute to $\Delta_0$. In order to cope with this issue, we consider another potential function $\phi_A'(C) = \frac{1}{2} |X \cap Y| + \frac{1}{4} |X \oplus Y|$ (which is identical to $\frac{1}{2} \phi_1(C)$). Note that $\phi'_A(C) - \frac{1}{6} \leq \phi_A(C) \leq \phi'_A(C)$ always holds. Then, if we define $\Delta_0$ for $\phi'_A$ (referring to \autoref{tab:potential_lemma1}, the contribution to $\Delta_0$ for each case (from the top) becomes $0, \frac{1}{8}, \frac{1}{8}, \frac{1}{4}, \frac{1}{8}, 0$), the increase of $\phi_A$ can be lower-bounded by $\Delta_0 - \frac{1}{6}$. Therefore, we can calculate $g(C_1, C_2)$ by changing Line 1 of \autoref{alg:lowerbound} to $d_0 \gets (\text{the value of $\Delta_0$ for potential $\phi'_A$ for the current $C_1, C_2$}) - \frac{1}{6}$.

\paragraph{Results.} In the DFS, we search $1773334$ state matrices that satisfy $\mathrm{LB} < \frac{59}{96}$ (where $\mathrm{LB}$ is defined for $\phi_A$) and are in standard form. Among them, $700415$ state matrices contain $1$ and $2$ in every row, and $415942$ of them actually satisfy $\Delta_1 < \frac{59}{96}$ (where $\Delta_1$ is defined for $\phi_A$). We compute $\textsc{PotentialIncrease}(C_1, C_2, C_3, C_4, \phi_A, \phi_B)$ for all of them, and each call returned a result of $\frac{91}{96}$ or more. In conclusion, it is shown that any randomized online coloring algorithm for bipartite graphs requires at least $\frac{91}{96} \log n - O(1)$ colors (\autoref{thm:lowerbound1}).

\paragraph{Computer checks and programs.} The programs used for the analysis and their results in this section can be downloaded at \url{https://github.com/square1001/online-bipartite-coloring}, which is the same URL as that in \autoref{sec:k2-upper}.

\section{Conclusion}

In this paper, we studied the online coloring of $k$-colorable graphs for $k \geq 5, k = 4$, and $k = 2$.

\begin{itemize}
    \item In \autoref{sec:k5}, we presented a deterministic online algorithm to color $k$-colorable graphs with $\widetilde{O}(n^{1-\frac{1}{k(k-1)/2}})$ colors. The key was to create an online algorithm for locally $k$-colorable graphs with $O(n^{1-\frac{1}{k(k-1)/2+1}})$ colors. We also improved the competitive ratio to $O(\frac{n}{\log \log n})$.
    \item In \autoref{sec:k4}, we presented a deterministic online algorithm to color $4$-colorable graphs with $O(n^{14/17})$ colors. The key was to make use of the second neighborhoods with the \emph{double greedy method} that uses First-Fit twice. We also applied the \emph{Common \& Simplify Technique} to take advantage of dense subgraph structures.
    \item In \autoref{sec:k2-upper}, we showed that a randomization of the algorithm by Lov\'{a}sz, Saks, and Trotter \cite{LST89} improves the performance to $1.034 \log_2 n + O(1)$ colors. We also showed that, in \autoref{sec:k2-lower}, no randomized algorithms can achieve $\frac{91}{96} \log_2 n - O(1)$ colors.
\end{itemize}

We were unable to improve the state-of-the-art bound $\widetilde{O}(n^{2/3})$ colors for $k = 3$ by \cite{Kie98}. The main reason is that the graph of degree $n^{2/3}$ is too sparse to use the Common \& Simplify technique, which we described in \autoref{subsec:k4-dense}. Indeed, to achieve $\widetilde{O}(n^{2/3-\varepsilon})$ colors for some $\varepsilon > 0$ with this technique, we must set parameters $(\alpha, \beta)$ satisfying $\alpha + \beta = 1-\varepsilon$, but the present technique requires $\widetilde{O}(n^{2-(\alpha+\beta)})$ colors, which is much above the requirement. Especially, we are still unable to solve the special case where there exists a 3-coloring of $G$ that for every $i \in \{1, 2, \dots, \frac{n}{k}\}$ ($k = n^{1/3-\varepsilon}$), the $(i-1)k+1, (i-1)k+2, \dots, ik$-th vertices (in the arrival order) have the same color.

Finally, we conclude this paper by highlighting several important open problems related to online graph coloring. We conjecture the following:

\begin{conjecture}\label{conj:ai-solve-1}
    There exists a deterministic online algorithm to color $3$-colorable graphs with $\widetilde{O}(n^{2/3-\varepsilon})$ colors (for some $\varepsilon > 0$).
\end{conjecture}

\begin{conjecture}\label{conj:ai-solve-2}
    There exists a deterministic online algorithm to color $k$-colorable graphs with $n^{1-1/o(k^2)}$ colors.
\end{conjecture}

\begin{conjecture}\label{conj:ai-solve-3}
    \textsc{RandomizedLST} uses at most $1.027 \log_2 n + O(1)$ colors for any bipartite graph $G$.
\end{conjecture}

\begin{conjecture}\label{conj:ai-solve-4}
    The optimal randomized online algorithm for coloring bipartite graphs is \textsc{RandomizedLST}, up to constant number of colors; therefore, no algorithm can achieve $1.026 \log_2 n - O(1)$ colors.
\end{conjecture}

We hope that future researchers (or possibly AIs) solve these questions.

\section*{Acknowledgments}

The authors thank our laboratory members Kaisei Deguchi, Tomofumi Ikeura, Yuta Inoue, Hiroaki Mori, and Akihito Yoneyama for their helpful suggestions for improving this paper.

\printbibliography

\end{document}